\documentclass[amsart]{article}
\usepackage{bm}
\usepackage{tikz}
\usepackage{tikz,pgfplots}
\usepackage{caption}
\usepackage{graphicx}
\usepackage{color}
\usepackage{amsmath}
\usepackage{amssymb}\usepackage{amsthm}
\usepackage{amsfonts}
\usepackage{amscd}
\begin{document}
\theoremstyle{plain}
\newtheorem*{ithm}{Theorem}
\newtheorem*{iprop}{Proposition}
\newtheorem*{idefn}{Definition}
\newtheorem{thm}{Theorem}[section]
\newtheorem{lem}[thm]{Lemma}
\newtheorem{dlem}[thm]{Lemma/Definition}
\newtheorem{prop}[thm]{Proposition}
\newtheorem{set}[thm]{Setting}
\newtheorem{cor}[thm]{Corollary}
\newtheorem*{icor}{Corollary}
\theoremstyle{definition}
\newtheorem{assum}[thm]{Assumption}
\newtheorem{notation}[thm]{Notation}
\newtheorem{setting}[thm]{Setting}
\newtheorem{defn}[thm]{Definition}
\newtheorem{clm}[thm]{Claim}
\newtheorem{ex}[thm]{Example}
\theoremstyle{remark}
\newtheorem{rem}[thm]{Remark}
\newcommand{\unit}{\mathbb I}
\newcommand{\ali}[1]{{\mathfrak A}_{[ #1 ,\infty)}}
\newcommand{\alm}[1]{{\mathfrak A}_{(-\infty, #1 ]}}
\newcommand{\nn}[1]{\lV #1 \rV}
\newcommand{\br}{{\mathbb R}}
\newcommand{\dm}{{\rm dom}\mu}
\newcommand{\lb}{l_{\bb}(n,n_0,k_R,k_L,\lal,\bbD,\bbG,Y)}
\newcommand{\Ad}{\mathop{\mathrm{Ad}}\nolimits}
\newcommand{\Proj}{\mathop{\mathrm{Proj}}\nolimits}
\newcommand{\RRe}{\mathop{\mathrm{Re}}\nolimits}
\newcommand{\RIm}{\mathop{\mathrm{Im}}\nolimits}
\newcommand{\Wo}{\mathop{\mathrm{Wo}}\nolimits}
\newcommand{\Prim}{\mathop{\mathrm{Prim}_1}\nolimits}
\newcommand{\Primz}{\mathop{\mathrm{Prim}}\nolimits}
\newcommand{\ClassA}{\mathop{\mathrm{ClassA}}\nolimits}
\newcommand{\Class}{\mathop{\mathrm{Class}}\nolimits}
\newcommand{\diam}{\mathop{\mathrm{diam}}\nolimits}
\def\qed{{\unskip\nobreak\hfil\penalty50
\hskip2em\hbox{}\nobreak\hfil$\square$
\parfillskip=0pt \finalhyphendemerits=0\par}\medskip}
\def\proof{\trivlist \item[\hskip \labelsep{\bf Proof.\ }]}
\def\endproof{\null\hfill\qed\endtrivlist\noindent}
\def\proofof[#1]{\trivlist \item[\hskip \labelsep{\bf Proof of #1.\ }]}
\def\endproofof{\null\hfill\qed\endtrivlist\noindent}
\numberwithin{equation}{section}

\newcommand{\kakunin}[1]{{\color{red}}}

\newcommand{\oo}{{\boldsymbol\omega}}
\newcommand{\ctv}{\caC_{(\theta,\varphi)}}
\newcommand{\btv}{\caB_{(\theta,\varphi)}}
\newcommand{\amf}{\mathfrak A}
\newcommand{\at}{\caA_{\bbZ^2}}
\newcommand{\oz}{\caO_0}
\newcommand{\caA}{{\mathcal A}}
\newcommand{\caB}{{\mathcal B}}
\newcommand{\caC}{{\mathcal C}}
\newcommand{\caD}{{\mathcal D}}
\newcommand{\caE}{{\mathcal E}}
\newcommand{\caF}{{\mathcal F}}
\newcommand{\caG}{{\mathcal G}}
\newcommand{\caH}{{\mathcal H}}
\newcommand{\caI}{{\mathcal I}}
\newcommand{\caJ}{{\mathcal J}}
\newcommand{\caK}{{\mathcal K}}
\newcommand{\caL}{{\mathcal L}}
\newcommand{\caM}{{\mathcal M}}
\newcommand{\caN}{{\mathcal N}}
\newcommand{\caO}{{\mathcal O}}
\newcommand{\caP}{{\mathcal P}}
\newcommand{\caQ}{{\mathcal Q}}
\newcommand{\caR}{{\mathcal R}}
\newcommand{\caS}{{\mathcal S}}
\newcommand{\caT}{{\mathcal T}}
\newcommand{\caU}{{\mathcal U}}
\newcommand{\caV}{{\mathcal V}}
\newcommand{\caW}{{\mathcal W}}
\newcommand{\caX}{{\mathcal X}}
\newcommand{\caY}{{\mathcal Y}}
\newcommand{\caZ}{{\mathcal Z}}
\newcommand{\bba}{{\mathbb a}}
\newcommand{\bbA}{{\mathbb A}}
\newcommand{\bbB}{{\mathbb B}}
\newcommand{\bbC}{{\mathbb C}}
\newcommand{\bbD}{{\mathbb D}}
\newcommand{\bbE}{{\mathbb E}}
\newcommand{\bbF}{{\mathbb F}}
\newcommand{\bbG}{{\mathbb G}}
\newcommand{\bbH}{{\mathbb H}}
\newcommand{\bbI}{{\mathbb I}}
\newcommand{\bbJ}{{\mathbb J}}
\newcommand{\bbK}{{\mathbb K}}
\newcommand{\bbL}{{\mathbb L}}
\newcommand{\bbM}{{\mathbb M}}
\newcommand{\bbN}{{\mathbb N}}
\newcommand{\bbO}{{\mathbb O}}
\newcommand{\bbP}{{\mathbb P}}
\newcommand{\bbQ}{{\mathbb Q}}
\newcommand{\bbR}{{\mathbb R}}
\newcommand{\bbS}{{\mathbb S}}
\newcommand{\bbT}{{\mathbb T}}
\newcommand{\bbU}{{\mathbb U}}
\newcommand{\bbV}{{\mathbb V}}
\newcommand{\bbW}{{\mathbb W}}
\newcommand{\bbX}{{\mathbb X}}
\newcommand{\bbY}{{\mathbb Y}}
\newcommand{\bbZ}{{\mathbb Z}}
\newcommand{\str}{^*}
\newcommand{\lv}{\left \vert}
\newcommand{\rv}{\right \vert}
\newcommand{\lV}{\left \Vert}
\newcommand{\rV}{\right \Vert}
\newcommand{\la}{\left \langle}
\newcommand{\ra}{\right \rangle}
\newcommand{\ltm}{\left \{}
\newcommand{\rtm}{\right \}}
\newcommand{\lcm}{\left [}
\newcommand{\rcm}{\right ]}
\newcommand{\ket}[1]{\lv #1 \ra}
\newcommand{\bra}[1]{\la #1 \rv}
\newcommand{\lmk}{\left (}
\newcommand{\rmk}{\right )}
\newcommand{\al}{{\mathcal A}}
\newcommand{\md}{M_d({\mathbb C})}
\newcommand{\ainn}{\mathop{\mathrm{AInn}}\nolimits}
\newcommand{\id}{\mathop{\mathrm{id}}\nolimits}
\newcommand{\Tr}{\mathop{\mathrm{Tr}}\nolimits}
\newcommand{\Ran}{\mathop{\mathrm{Ran}}\nolimits}
\newcommand{\Ker}{\mathop{\mathrm{Ker}}\nolimits}
\newcommand{\Aut}{\mathop{\mathrm{Aut}}\nolimits}
\newcommand{\spn}{\mathop{\mathrm{span}}\nolimits}
\newcommand{\Mat}{\mathop{\mathrm{M}}\nolimits}
\newcommand{\UT}{\mathop{\mathrm{UT}}\nolimits}
\newcommand{\DT}{\mathop{\mathrm{DT}}\nolimits}
\newcommand{\GL}{\mathop{\mathrm{GL}}\nolimits}
\newcommand{\spa}{\mathop{\mathrm{span}}\nolimits}
\newcommand{\supp}{\mathop{\mathrm{supp}}\nolimits}
\newcommand{\rank}{\mathop{\mathrm{rank}}\nolimits}
\newcommand{\idd}{\mathop{\mathrm{id}}\nolimits}
\newcommand{\ran}{\mathop{\mathrm{Ran}}\nolimits}
\newcommand{\dr}{ \mathop{\mathrm{d}_{{\mathbb R}^k}}\nolimits} 
\newcommand{\dc}{ \mathop{\mathrm{d}_{\cc}}\nolimits} \newcommand{\drr}{ \mathop{\mathrm{d}_{\rr}}\nolimits} 
\newcommand{\zin}{\mathbb{Z}}
\newcommand{\rr}{\mathbb{R}}
\newcommand{\cc}{\mathbb{C}}
\newcommand{\ww}{\mathbb{W}}
\newcommand{\nan}{\mathbb{N}}\newcommand{\bb}{\mathbb{B}}
\newcommand{\aaa}{\mathbb{A}}\newcommand{\ee}{\mathbb{E}}
\newcommand{\pp}{\mathbb{P}}
\newcommand{\wks}{\mathop{\mathrm{wk^*-}}\nolimits}
\newcommand{\mk}{{\Mat_k}}
\newcommand{\mnz}{\Mat_{n_0}}
\newcommand{\mn}{\Mat_{n}}
\newcommand{\dist}{\dc}
\newcommand{\braket}[2]{\left\langle#1,#2\right\rangle}
\newcommand{\ketbra}[2]{\left\vert #1\right \rangle \left\langle #2\right\vert}
\newcommand{\abs}[1]{\left\vert#1\right\vert}
\newcommand{\trl}[2]
{T_{#1}^{(\theta,\varphi), \Lambda_{#2},\bar V_{#1,\Lambda_{#2}}}}
\newcommand{\trlz}[1]
{T_{#1}^{(\theta,\varphi), \Lambda_{0},\unit}}
\newcommand{\trlt}[2]
{T_{#1}^{(\theta,\varphi), \Lambda_{#2}+t_{#2}\bm e_{\Lambda_{#2}},\bar V_{#1,\Lambda_{#2}+t_{#2}\bm e_{\Lambda_{#2}}}}}
\newcommand{\trltj}[4]
{T_{#1}^{(\theta,\varphi), \Lambda_{#2}^{(#3)}+t_{#4}\bm e_{\Lambda_{#2}^{(#3)}},\bar V_{#1,\Lambda_{#2}^{(#3)}+t_{#4}\bm e_{\Lambda_{#2}^{(#3)}}}}}
\newcommand{\trltjp}[4]
{T_{#1}^{(\theta,\varphi), {\Lambda'}_{#2}^{(#3)}+t_{#4}'\bm e_{{\Lambda'}_{#2}^{(#3)}},\bar V_{#1,{\Lambda'}_{#2}^{(#3)}+t_{#4}'\bm e_{{\Lambda'}_{#2}^{(#3)}}}}}
\newcommand{\trlta}[2]
{T_{#1}^{(\theta,\varphi), \Lambda_{#2}^{t_{#2}},\bar V_{#1,\Lambda_{#2}^{t_{#2}}}}}
\newcommand{\trltb}[2]
{T_{#1}^{(\theta,\varphi), \Lambda_{#2}+t\bm e_{\Lambda_{#2}},\bar V_{#1,\Lambda_{#2}+t\bm e_{\Lambda_{#2}}}}}

\newcommand{\trlpt}[2]
{T_{#1}^{(\theta,\varphi), \Lambda_{#2}'+t_{#2}'\bm e_{\Lambda_{#2}'},\bar V_{#1,\Lambda_{#2}'+t_{#2}'\bm e_{\Lambda_{#2}'}}}}

\newcommand{\trll}[3]
{T_{#1, #3}^{(\theta,\varphi), \Lambda_{#2},\bar V_{#1,\Lambda_{#2}}}}
\newcommand{\trlp}[2]
{T_{#1}^{(\theta,\varphi), \Lambda_{#2}',\bar V_{#1,\Lambda_{#2}'}}}
\newcommand{\trlpp}[2]
{T_{#1}^{(\theta,\varphi), \Lambda_{#2}'',\bar V_{#1,\Lambda_{#2}''}}}
\newcommand{\trlj}[3]
{T_{\rho_{#1}}^{(\theta,\varphi), \Lambda_{#2}^{(#3)}, V_{\rho_{#1},\Lambda_{#2}^{(#3)}}}}
\newcommand{\trljp}[3]
{T_{{\rho'}_{#1}}^{(\theta,\varphi), {\Lambda'}_{#2}^{(#3)},V_{\rho'_{#1},{\Lambda'}_{#2}^{(#3)}}}}
\newcommand{\wod}[3]
{W_{#1\Lambda_{#2}\Lambda_{#3}}}
\newcommand{\wodt}[3]
{{W^{\bm t}}_{#1\Lambda_{#2}\Lambda_{#3}}}
\newcommand{\comp}[2]
{{(\theta_{#1},\varphi_{#1}), \Lambda_{#2},\{\bar V_{\eta,\Lambda_{#2}}\}_\eta}}
\newcommand{\ltj}[2]{\Lambda_{#1}+{#2} \bm e_{\Lambda_{#1}} }
\newcommand{\ltjp}[2]{{\Lambda'}_{#1}+{#2} \bm e_{\Lambda_{#1}} }
\newtheorem{nota}{Notation}[section]
\def\qed{{\unskip\nobreak\hfil\penalty50
\hskip2em\hbox{}\nobreak\hfil$\square$
\parfillskip=0pt \finalhyphendemerits=0\par}\medskip}
\def\proof{\trivlist \item[\hskip \labelsep{\bf Proof.\ }]}
\def\endproof{\null\hfill\qed\endtrivlist\noindent}
\def\proofof[#1]{\trivlist \item[\hskip \labelsep{\bf Proof of #1.\ }]}
\def\endproofof{\null\hfill\qed\endtrivlist\noindent}
\newcommand{\wrl}[2]{Y_{#1}^{\Lambda_0^{(#2)}}}
\newcommand{\wrlt}[2]{\tilde Y_{#1}^{\Lambda_0^{(#2)}}}
\newcommand{\ZZ}{\bbZ_2\times\bbZ_2}
\newcommand{\SSS}{\mathcal{S}}
\newcommand{\cs}{S}
\newcommand{\ct}{t}
\newcommand{\hS}{S}
\newcommand{\vv}{{\boldsymbol v}}
\newcommand{\ala}{a}
\newcommand{\bet}{b}
\newcommand{\gam}{c}
\newcommand{\alphas}{\alpha}
\newcommand{\alphai}{\alpha^{(\sigma_{1})}}
\newcommand{\alphan}{\alpha^{(\sigma_{2})}}
\newcommand{\betas}{\beta}
\newcommand{\betai}{\beta^{(\sigma_{1})}}
\newcommand{\betan}{\beta^{(\sigma_{2})}}
\newcommand{\alphass}{\alpha^{{(\sigma)}}}
\newcommand{\uu}{V}
\newcommand{\vp}{\varsigma}
\newcommand{\vpr}{R}
\newcommand{\tg}{\tau_{\Gamma}}
\newcommand{\sgg}{\Sigma_{\Gamma}}
\newcommand{\nh}{t28}
\newcommand{\rk}{6}
\newcommand{\nii}{2}
\newcommand{\nhh}{28}
\newcommand{\sjt}{30}
\newcommand{\sjtg}{30}
\newcommand{\bcg}{\caB(\caH_{\alpha})\otimes  C^{*}(\Sigma)}
\newcommand{\pza}[1]{\pi_0\lmk\caA_{\Lambda_{#1}}\rmk''}
\newcommand{\pzac}[1]{\pi_0\lmk\caA_{\Lambda_{#1}}\rmk'}
\newcommand{\pzacc}[1]{\pi_0\lmk\caA_{\Lambda_{#1}^c}\rmk'}
\newcommand{\trlzi}[2]{T_{#1}^{(\theta,\varphi) \Lambda_0^{(#2)}\unit}}

\newcommand{\obk}{\omega_{\mathop{\mathrm{bk}}}}
\newcommand{\obd}{\omega_{\mathop{\mathrm{bd}}}}
\newcommand{\obdm}{\omega_{\mathop{\mathrm{bd}(-)}}}
\newcommand{\abk}{\caA_{\mathbb Z^2}}
\newcommand{\hu}{\mathop {\mathrm H_{U}}}
\newcommand{\hd}{\mathop {\mathrm H_{D}}}

\newcommand{\abd}{\caA_{\hu}}
\newcommand{\aloch}{\caA_{\hu,\mathop{\mathrm {loc}}}}
\newcommand{\alocg}[1]{\caA_{#1,\mathop{\mathrm {loc}}}}
\newcommand{\hbk}{\caH_{\mathop{\mathrm{bk}}}}
\newcommand{\hbd}{\caH_{\mathop{\mathrm{bd}}}}
\newcommand{\hbdm}{\caH_{\mathop{\mathrm{bd}(-)}}}
\newcommand{\pbk}{\pi_{\mathop{\mathrm{bk}}}}
\newcommand{\pbd}{\pi_{\mathop{\mathrm{bd}}}}
\newcommand{\pbdm}{\pi_{\mathop{\mathrm{bd}(-)}}}
\newcommand{\mopbk}{{\mathop{\mathrm{bk}}}}
\newcommand{\mopbd}{{\mathop{\mathrm{bd}}}}
\newcommand{\Obk}{O_{\mathop{\mathrm{bk}}}}
\newcommand{\OUbk}{O_{\mathop{\mathrm{bk}}}^U}
\newcommand{\Orbd}{O^r_{\mathop{\mathrm{bd}}}}
\newcommand{\OrUbd}{O^{rU}_{\mathop{\mathrm{bd}}}}
\newcommand{\Obkl}{O_{\mathop{\mathrm{bk},\Lambda_0}}}
\newcommand{\OUbkl}{O_{\mathop{\mathrm{bk},\Lambda_0}}^U}
\newcommand{\Orbdl}{O^r_{\mathop{\mathrm{bd},\Lambda_{r0}}}}
\newcommand{\OrUbdl}{O^{rU}_{\mathop{\mathrm{bd},\Lambda_{r0}}}}
\newcommand{\Obu}{O_{\mathop{\mathrm{bd}}}^{\mathop{\mathrm{bu}}}}
\newcommand{\Obul}{O_{\mathop{\mathrm{bd}},\lz}^{\mathop{\mathrm{bu}}}}
\newcommand{\Odl}{O_{\lzr}^{\caD}}
\newcommand{\fd}{F^{\caD}}
\newcommand{\hilb}{\mathop{\mathrm {Hilb}}_f}
\newcommand{\Obun}[1]{O_{\mathop{\mathrm{bd}#1}}^{\mathop{\mathrm{bu}}}}
\newcommand{\OrUbdn}[1]{O^{rU}_{\mathop{\mathrm{bd}#1 }}}
\newcommand{\OUbkm}{O_{\mathop{\mathrm{bk}}}^{U(-)}}
\newcommand{\Orbdm}{O^{r(-)_{\mathop{\mathrm{bd} }}}}
\newcommand{\OrUbdm}{O^{rU(-)}_{\mathop{\mathrm{bd}}}}
\newcommand{\OUbklm}{O_{\mathop{\mathrm{bk},\Lambda_{0(-)}}}^{U(-)}}
\newcommand{\Orbdlm}{O^{r(-)}_{\mathop{\mathrm{bd},\Lambda_{r0(-)}}}}
\newcommand{\OrUbdlm}{O^{rU(-)}_{\mathop{\mathrm{bd},\Lambda_{r0(-)}}}}
\newcommand{\Obum}{O_{\mathop{\mathrm{bd}}}^{\mathop{\mathrm{bu}(-)}}}
\newcommand{\Obulm}{O_{\mathop{\mathrm{bd}},{\lm {0(-)}}}^{\mathop{\mathrm{bu}(-)}}}
\newcommand{\Otot}{O_{\mathop{\mathrm{total}}}}
\newcommand{\Ototl}{O_{\mathop{\mathrm{total}}}^{\lz,\lzm}}

\newcommand{\Obj}{{\mathop{\mathrm{Obj}}}}
\newcommand{\Mor}{{\mathop{\mathrm{Mor}}}}
\newcommand{\ti}[4]{\tilde\iota\lmk(#1,#2), (#3,#4)\rmk}

\newcommand{\Cabkl}{C_{\mathop{\mathrm{bk},\lz}}}
\newcommand{\Cabk}{C_{\mathop{\mathrm{bk}}}}
\newcommand{\CaUbk}{C_{\mathop{\mathrm{bk}}}^U}
\newcommand{\Carbd}{C^r_{\mathop{\mathrm{bd}}}}
\newcommand{\CarUbd}{C^{rU}_{\mathop{\mathrm{bd}}}}
\newcommand{\CaUbkl}{C_{\mathop{\mathrm{bk},\Lambda_0}}^U}
\newcommand{\Carbdl}{C^r_{\mathop{\mathrm{bd},\Lambda_{r0}}}}
\newcommand{\CarUbdl}{C^{rU}_{\mathop{\mathrm{bd},\Lambda_{r0}}}}
\newcommand{\Cabu}{C_{\mathop{\mathrm{bd}}}^{\mathop{\mathrm{bu}}}}
\newcommand{\Cabul}{C_{\mathop{\mathrm{bd}},\lz}^{\mathop{\mathrm{bu}}}}
\newcommand{\Cadl}{C_{\lzr}^{\caD}}
\newcommand{\Hom}{\mathop{\mathrm{Hom}}}
\newcommand{\Catotl}{C_{\mathop{\mathrm{total}}}^{\lz,\lzm}}
\newcommand{\Cafin}{C_{\mathop{\mathrm{bd}},\lz}^{\mathop{\mathrm{fin}}}}

\newcommand{\mm}{\Mor_{\tilde\caM}}
\newcommand{\om}{\Obj\lmk {\tilde\caM}\rmk}
\newcommand{\omt}{\otimes_{\tilde\caM}}

\newcommand{\CaUbkm}{C_{\mathop{\mathrm{bk}}}^{U(-)}}
\newcommand{\Carbdm}{C^{r(-)}_{\mathop{\mathrm{bd}}}}
\newcommand{\CarUbdm}{C^{rU(m)}_{\mathop{\mathrm{bd}}}}
\newcommand{\CaUbklm}{C_{\mathop{\mathrm{bk},\Lambda_{0(-)}}}^{U(-)}}
\newcommand{\Carbdlm}{C^{r(-)}_{\mathop{\mathrm{bd},\Lambda_{r0(-)}}}}
\newcommand{\CarUbdlm}{C^{rU(-)}_{\mathop{\mathrm{bd},\Lambda_{r0(-)}}}}
\newcommand{\Cabum}{C_{\mathop{\mathrm{bd}}}^{\mathop{\mathrm{bu}(-)}}}
\newcommand{\Cabuln}{C_{\mathop{\mathrm{bd}},\lm{0(-)}}^{\mathop{\mathrm{bu}(-)}}}
\newcommand{\Irr}{\mathop{\mathrm{Irr}}}


\newcommand{\Tbk}[3]{T_{#1}^{(\frac{3\pi}2,\frac\pi 2), {#2},{#3}}}
\newcommand{\Vrl}[2]{V_{#1,#2}}
\newcommand{\Tbkv}[2]{T_{#1}^{(\frac{3\pi}2,\frac\pi 2), #2, V_{#1,#2}}}
\newcommand{\Tbd}[3]{T_{#1}^{\mathrm{(l)} #2 #3}}
\newcommand{\Tbdv}[2]{T_{#1}^{\mathrm{(l)} #2\Vrl{{#1}}{#2}}}

\newcommand{\zc}{\caZ\lmk\Cabul\rmk}
\newcommand{\hi}[2]{\hat\iota\lmk #1: #2\rmk}

\newcommand{\lz}{\Lambda_0}
\newcommand{\lzr}{\Lambda_{r0}}
\newcommand{\lm}[1]{{\Lambda_{#1}}}
\newcommand{\llz}{(\lz,\lzr)}
\newcommand{\lmr}[1]{{\Lambda_{r#1}}}
\newcommand{\hlm}[1]{{\hat\Lambda_{#1}}}
\newcommand{\tlm}[1]{{\tilde\Lambda_{#1}}}
\newcommand{\ld}{\Lambda}
\newcommand{\pc}{\mathcal{PC}}
\newcommand{\Cbk}{\caC_{\mopbk}}
\newcommand{\CUbk}{\caC_{\mopbk}^U}
\newcommand{\Crbd}{\caC_{\mopbd}^r}
\newcommand{\Clbd}{\caC_{\mopbd}^l}
\newcommand{\Vbk}[2]{\caV_{#1#2}^{\mopbk}}
\newcommand{\Vbd}[2]{\caV_{#1#2}^{\mopbd}}
\newcommand{\VUbd}[2]{\caV_{#1#2}^{\mopbd U} }
\newcommand{\Vbu}[2]{\caV_{#1#2}^{\mathop{\mathrm{bu}}}}
\newcommand{\bl}{\caB_l}
\newcommand{\fbk}{\caF_{\mopbk}^U}
\newcommand{\fbd}{\caF_{\mopbd}^U}
\newcommand{\gu}{\caG}
\newcommand{\hb}[2]{\iota^{(\lz,\lzr)}\lmk#1: #2\rmk}
\newcommand{\hfc}{\hat F^{\llz}}
\newcommand{\ffc}{F^{\llz}}

\newcommand{\lzm}{\Lambda_{0(-)}}
\newcommand{\lzrm}{\Lambda_{r0(-)}}
\newcommand{\pcm}{\mathcal{PC(-)}}
\newcommand{\CUbkm}{\caC_{\mopbk}^{U(-)}}
\newcommand{\Crbdm}{\caC_{\mopbd}^{r(-)}}
\newcommand{\Clbdm}{\caC_{\mopbd}^{l(-)}}
\newcommand{\Vbkm}[2]{\caV_{#1#2}^{\mopbk (-)}}
\newcommand{\Vbdm}[2]{\caV_{#1#2}^{\mopbd(-)}}
\newcommand{\VUbdm}[2]{\caV_{#1#2}^{\mopbd U(-)} }
\newcommand{\Vbum}[2]{\caV_{#1#2}^{\mathop{\mathrm{bu}(-)}}}
\newcommand{\fbkm}{\caF_{\mopbk}^{U(-)}}
\newcommand{\fbdm}{\caF_{\mopbd}^{U(-)}}

\newcommand{\lr}[2]{\Lambda_{(#1,0),#2,#2}}
\newcommand{\lef}[2]{\overline{\Lambda_{(#1,0),\pi-#2,#2}}}
\newcommand{\lrhu}{(\Lambda_r)^c\cap \hu}
\newcommand{\lhuc}{\Lambda^c\cap \hu}
\newcommand{\lhu}{\Lambda\cap \hu}
\newcommand{\lrhuz}{(\Lambda_{r0})^c\cap \hu}
\newcommand{\lhucz}{\Lambda_{0}^c\cap \hu}
\newcommand{\lc}[1]{(\lm #1)^c\cap\hu}

\newcommand{\zam}{\caZ_a\lmk\tilde\caM\rmk }
\newcommand{\ozam}{\Obj\lmk \zam\rmk}
\newcommand{\mzam}{\Mor_{\zam}}
\newcommand{\rpc}[2]{\lmk(#1,#2), C_{(#1,#2)}\rmk}
\newcommand{\ozt}{\otimes_{\zam}}

\newcommand{\change}[1]{#1}

\title{Boundary states of a bulk gapped ground state in $2$-D quantum spin systems }

\author{Yoshiko Ogata \\Research Institute for Mathematical Sciences\\
 Kyoto University, Kyoto 606-8502 JAPAN}
\maketitle

\begin{abstract}
We introduce a natural mathematical definition of
 {\it boundary states} of a bulk gapped ground state in the operator algebraic framework
 of $2$-D quantum spin systems.
With the approximate Haag duality at the boundary,
we derive a $C^*$-tensor category $\tilde \caM$ out of such boundary state.
Under a non-triviality condition of the braiding in the bulk,
we show that the Drinfeld center (with an asymptotic constraint) of
$\tilde \caM$ is equivalent to the bulk braided $C^*$-tensor category
derived in \cite{MTC}.
\end{abstract}

\section{Introduction}
Since the discovery of Kitaev's beautiful model \cite{kitaev2006anyons}, the classification problem of topologically ordered gapped
systems has attracted
a lot of attention.
Ingenious models were made  \cite{levin2005string}, and 
a lot of categorical arguments were carried out.
On the other hand,  Naaijkens carried out an analysis of
Kitaev quantum double  models from the operator algebraic/ algebraic quantum field theoretic (DHR theory)
point of view \cite{N1}, \cite{N2},\cite{FN}, \cite{naaijkens2013kosaki}. 
The advantage of this direction is that DHR theory allows us to
treat non-solvable models, hence it allows us to think of the classification.
In \cite{MTC}, we derived braided $C^*$-tensor categories out of
gapped ground states on two-dimensional quantum spin systems
satisfying the approximate Haag duality.
There, following the DHR theory \cite{BF},\cite{DHRI},\cite{fredenhagen1989superselection},\cite{frohlich1990braid},\cite{longo1989index}
our objects are representations satisfying the superselection criterions.
In the same paper, we showed that this braided $C^*$-tensor category is an
invariant of the classification of gapped ground state phases.
This can be regarded as an analytical realization of the statement 
that {\it anyons are invariant of topological phases}.
In particular, the existence of anyons means the ground state is
long-range entangled \cite{NaOg}, as was shown in \cite{haah2016invariant}
for a class of commuting Hamiltonians.

Bravyi and Kitaev introduced two gapped boundaries of the toric code model \cite{bravyi1998quantum}.
In \cite{kitaev2012models}, Kitaev and Kong studied the model from the categorical point of view.
There, Kitaev and Kong conjectured that the bulk theory is the Drinfeld center of the 
gapped boundary theory, and ``physically proved'' it
 for all 2+1D Levin-Wen models. 
 {It is also stated in $\rm Theorem^{ph}$ 4.2.20 of \cite{konginvitation}.}
Recently,  Wallick \cite{wa} analyzed the corresponding theory 
in the toric code model, in terms of superselection sectors, following the corresponding bulk analysis
by Naaijkens \cite{N1}\cite{N2}.
There, he showed the boundary theory is a module tensor category over the bulk.

Although it is a broadly studied subject, as far as we know, there is no mathematical formulation to specify what {\it a boundary state of a bulk gapped ground state} means in the operator algebraic framework.
In this paper, we suggest a natural mathematical definition of it. See Definition \ref{boundary}.
Physically, a pure state $\obd$
 on the half-plane is a boundary state of a bulk gapped ground state $\obk$ if
 they are macroscopically the same away from the boundary.
 We analyze such a boundary state $\obd$ in this paper.
We derive a $C^*$-tensor category $\tilde \caM$
for $\obd$
under the boundary version of approximate Haag duality.
We then construct a Drinfeld center $\zam$
of $\tilde\caM$, with some asymptotic constraint (Definition \ref{hato}).
We show that  
$\zam$ is a braided $C^*$-tensor category equivalent to the bulk braided $C^*$-tensor category
derived in \cite{MTC},
if the braiding is nontrivial in a certain sense (Assumption \ref{raichi} and Remark \ref{dil}).
\change{ If we interpret our nontriviality condition related to gapped boundary condition}, it gives a proof for the Kitaev and Kong's conjecture \cite{kitaev2012models}, \cite{konginvitation}, in our setting.

\subsection{Setting}\label{setting2}
Throughout this paper, we fix some $2\le d\in\nan$.
We denote the algebra of $d\times d$ matrices by $\Mat_{d}$.
For each $z\in\bbZ^2$,  let $\caA_{\{z\}}$ be an isomorphic copy of $\Mat_{d}$, and for any finite subset $\Lambda\subset\bbZ^2$, we set $\caA_{\Lambda} = \bigotimes_{z\in\Lambda}\caA_{\{z\}}$.
For finite $\Lambda$, the algebra $\caA_{\Lambda} $ can be regarded as the set of all bounded operators acting on
the Hilbert space $\bigotimes_{z\in\Lambda}{\bbC}^{d}$.
We use this identification freely.
If $\Lambda_1\subset\Lambda_2$, the algebra $\caA_{\Lambda_1}$ is naturally embedded in $\caA_{\Lambda_2}$ by tensoring its elements with the identity. 
For an infinite subset $\Gamma\subset \bbZ^{2}$,
$\caA_{\Gamma}$
is given as the inductive limit of the algebras $\caA_{\Lambda}$ with $\Lambda$, finite subsets of $\Gamma$.
We call $\caA_{\Gamma}$ the quantum spin system on $\Gamma$.
For a subset $\Gamma_1$ of $\Gamma\subset\bbZ^{2}$,
the algebra $\caA_{\Gamma_1}$ can be regarded as a subalgebra of $\caA_{\Gamma}$. 
For $\Gamma\subset \bbR^2$, with a bit of abuse of notation, we write $\caA_{\Gamma}$
to denote $\caA_{\Gamma\cap \bbZ^2}$.
Also, $\Gamma^c$ denotes the complement of $\Gamma$ in $\bbR^2$.
We denote the algebra of local elements in $\Gamma$ by $\caA_{\Gamma, \rm{loc}}$,
i.e. $\caA_{\Gamma, \rm{loc}}=\cup_{\Lambda\Subset\Gamma}\caA_\Lambda$.
The upper half-plane $\bbR\times [0,\infty)$ is denoted by $\hu$,
and the lower half-plane $\bbR\times (-\infty,0)$ is denoted by $\hd$.
For
each $\bm a\in \bbR^2$, $\theta\in\bbR$ and $\varphi\in (0,\pi)$,
we set 
\begin{align*}
\Lambda_{\bm a, \theta,\varphi}
:
=&\lmk \bm a+\left\{
t\bm e_{\beta}\mid t>0,\quad \beta\in (\theta-\varphi,\theta+\varphi)\right\}\rmk.
\end{align*}
where  $\bm e_{\beta}=(\cos\beta,\sin\beta)$.
We call a set of this form a cone, and denote by $\Cbk$ the set of all cones.
We set $\arg\Lambda_{\bm a, \theta,\varphi}:=(\theta-\varphi,\theta+\varphi)$ and $|\arg\Lambda_{\bm a, \theta,\varphi}|=2\varphi$.
We also set $\bm e_\Lambda:=\bm e_{\theta}$ for $\Lambda=\Lambda_{\bm a, \theta,\varphi}$.
For $\varepsilon>0$ and $\Lambda=\Lambda_{\bm a, \theta,\varphi}$, $\Lambda_\varepsilon$ 
denotes $\Lambda_\varepsilon=\Lambda_{\bm a, \theta,\varphi+\varepsilon}$.
We consider the following sets of cones:
\begin{align}
\begin{split}
&\CUbk:=\left\{
\Lambda_{\bm a,\theta,\varphi}\mid \bm a\in \bbR^2,\;
0<\theta-\varphi<\theta+\varphi<\pi
\right\},\\
&\Crbd:=\left\{
\Lambda_{(a,0),\varphi,\varphi}\mid a\in\bbR,\; 0<\varphi<\frac\pi 2
\right\},\\
&\Clbd:=\left\{
\hu\cap \lmk \Lambda_r\rmk^c\mid \Lambda_r\in\Crbd
\right\}=\left\{
\overline{\Lambda_{(a,0),\frac\pi 2+\varphi,\frac\pi 2-\varphi}}\mid a\in\bbR,\; 0<\varphi<\frac\pi 2
\right\},\\
&\caC_{(\theta,\varphi)}:=\left\{
\Lambda\mid \overline{\arg\Lambda}\cap{[\theta-\varphi,\theta+\varphi]}=\emptyset
\right\}.
\end{split}
\end{align}
See Figures \ref{fune}-\ref{ikada}.
Note that they are upward-filtering sets with respect to the inclusion relation.
We denote by $\pc$ the set of all pairs $(\lz,\lzr)\in \CUbk\times \Crbd$
with $\lz\subset \lzr$.

Let $\obk$ be a pure state on $\abk$ with a GNS triple
$(\hbk, \pbk,\Omega_{\mopbk} )$.
Let $\obd$ be a pure state on $\abd$ with a GNS triple
$(\hbd, \pbd,\Omega_{\mopbd} )$.
Now we introduce our definition of boundary states of a bulk state.
\begin{defn}\label{boundary}
We say that $\obd$ is a boundary state of the bulk state $\obk$ if 
$\pbd\vert_{\caA_{\lhu}}$ and $\pbk\vert_{\caA_{\lhu}}$ are quasi-equivalent
for any $\Lambda\in\CUbk$.
\end{defn}
Recall that in quantum spin systems, two quasi-equivalent states can be approximately represented by
 to each other's local perturbations.
From this point of view, this definition is telling us that  $\obd$ is a boundary state of the bulk state $\obk$ if they are macroscopically the same away from the boundary.
Namely, if we restrict these states to the regions of the form Figure \ref{fune} or Figure \ref{boat},
then they look macroscopically the same.
We show this property for ground states of frustration-free Hamiltonians satisfying ``strict cone
TQO''. See section \ref{gappedboundary}. 
Although it is not the main target of this paper, we take this definition also for the case that the boundary state has a current at the edge.

\begin{figure}[htbp]
\centering
\begin{tabular}{cccc}
\begin{minipage}{0.3\hsize}
\begin{tikzpicture}[domain=-4:4, scale=0.4]
   \filldraw[fill=lightgray, opacity=.5, very thick] (-3.5,5)--(-1,1)--(1.5,5);
   \draw[dashed] (-5,0)--(5,0);
\end{tikzpicture}
\caption{\small $\Lambda\cap\hu$ for $\ld\in \CUbk$}
\label{fune}
\end{minipage}
\begin{minipage}{0.3\hsize}
\begin{tikzpicture}[domain=-4:4, scale=0.4]
   \filldraw[fill=lightgray, opacity=.5, very thick] (-3.5,5)--(-2,0)--(1,0)--(2.5,5);
   \draw[dashed] (-5,0)--(5,0);
\end{tikzpicture}
\caption{\small $\Lambda\cap\hu$ for $\ld\in \CUbk$}
\label{boat}
\end{minipage}
\begin{minipage}{0.3\hsize}
\begin{tikzpicture}[domain=-4:4, scale=0.4]
   \filldraw[fill=lightgray, opacity=.5, very thick] (-5,0)--(1.5,0)--(-5,5);
   \draw[dashed] (-5,0)--(5,0);
\end{tikzpicture}
\caption{$\lm l\in \Clbd$}
\label{ahiru}
\end{minipage}
\begin{minipage}{0.3\hsize}
\begin{tikzpicture}[domain=-4:4, scale=0.4]
   \filldraw[fill=gray, opacity=.5, very thick] (5,0)--(-1.5,0)--(5,5);
   \draw[dashed] (-5,0)--(5,0);
\end{tikzpicture}
\caption{$\lm r\in \Crbd$}
\label{ikada}
\end{minipage}

\end{tabular}
\end{figure}

We set
\begin{align}
\begin{split}
&\caF_{\mopbk}^{U(0)}:=\cup_{\Lambda\in\CUbk}\pbk\lmk\caA_{\Lambda \cap \hu }\rmk'',\quad
\caF_{\mopbk}^U:=\overline{\caF_{\mopbk}^{U(0)}}^n,\\
&\caF_{\mopbd}^{U(0)}:=\cup_{\Lambda\in\CUbk}\pbd\lmk\caA_{\Lambda \cap \hu}\rmk'',\quad
\caF_{\mopbd}^U:=\overline{\caF_{\mopbd}^{U(0)}}^n,\\
&\caG^{(0)}:=\cup_{\Lambda\in\CUbk}\pbk\lmk\caA_{\Lambda}\rmk'',\quad
\caG:=\overline{\caG^{(0)}}^n,\\
&\caB_l^{(0)}:=\cup_{\Lambda_l\in \Clbd}\pbd\lmk\caA_{\Lambda_l}\rmk'',\quad
\caB_l:=\overline{\caB_l^{(0)}}^n,\\
&\caB_{(\theta,\varphi)}:=
\overline{\cup_{\Lambda\in\caC_{(\theta,\varphi)} }\pbk(\caA_{\Lambda})'' }^n.
\end{split}
\end{align}
where ${\overline{\cdot}}^n$ indicates the norm closure.
Note that $\caG=\caB_{(\frac{3\pi}2,\frac\pi 2)}$.
Because $\CUbk,\Clbd,\caC_{(\theta,\varphi)}$ are upward-filtering sets, 
$\caF^U_{\mopbk}$, $\caF^U_{\mopbd}$, $\caG$, $\caB_l$, 
$\caB_{(\theta,\varphi)}$ are $C^*$-algebras.
\change{Note that $\caB_{(\theta,\varphi)}$ and $\caB_l$ are introduced in \cite{MTC} and
\cite{wa} respectively.}
For a $C^*$-algebra $\mathfrak C$, we denote by $\caU(\mathfrak C)$ the set of all unitaries in $\mathfrak C$.

For example, to construct $\caF_{\mopbd}^{U(0)}$, we collect all the von Neumann algebras
$\pbd\lmk \caA_{\ld \cap\hu}\rmk''$ corresponding to the areas
of the shape Figure \ref{fune} or Figure \ref{boat}.
Then take the norm closure and obtain $\fbd$.
We regard this $C^*$-algebra $\fbd$, the $C^*$-algebra
corresponding to {\it operations using only the bulk}.
To construct $\caB_l^{(0)}$, we collect all the von Neumann algebras
$\pbd\lmk\caA_{\lm l}\rmk''$ corresponding to the areas
of the shape Figure \ref{ahiru}.
Then take the norm closure and obtain $\bl$.
We regard this $C^*$-algebra $\bl$, the $C^*$-algebra
corresponding to {\it operations using the left boundary}.
By definition, we can easily see 
\begin{align}
\pbk\lmk \abd\rmk\subset \fbk,\quad \pbd\lmk \abd\rmk\subset \fbd,\quad \pbk\lmk \abk\rmk\subset \caG.
\end{align}

\change{Next, we introduce representations of our interest. They are versions of localized and transportable representations in the DHR-theory. One important point that is not found in the DHR-theory is we sometimes consider only certain 
charge transporters ($\Vbd{\rho}{\Lambda_r}$ etc. below).}

For a representation $\rho$ of $\abk$ on $\hbk$ and $\Lambda\in \Cbk$,
we define a set of unitaries
\begin{align}
\Vbk{\rho}{\Lambda}:=\left\{
V_{\rho\Lambda}\in\caU\lmk \caH_{\mopbk}\rmk
\mid \left.\Ad\lmk V_{\rho\Lambda}\rmk\circ\rho\right\vert_{\caA_{\Lambda^c}}
=\left.\pbk \right\vert_{\caA_{\Lambda^c}}
\right\}.
\end{align}
We denote by $\Obk$ the set of all representations $\rho$ of $\abk$ on $\hbk$,
which has a non-empty $\Vbk{\rho}{\Lambda}$ for all $\Lambda\in \Cbk$.
We denote by $\OUbk$ the set of all representations $\rho$ of $\abk$ on $\hbk$,
which has a non-empty $\Vbk{\rho}{\Lambda}$ for any $\Lambda\in \CUbk$.
Note that $\Obk\subset\OUbk$.
For $\Lambda_0\in \CUbk$, we denote by $\OUbkl$
the set of all $\rho\in \OUbk$ such that $\rho\vert_{\caA_{\Lambda_0^c}}=\pbk\vert_{\caA_{\Lambda_0^c}}$.
\change{(The excitation is localized in $\ld_0$.) Physically, the excitation of $\rho\in \Obk$ can be moved to any cone in $\bbZ^2$, while for
$\rho\in \OUbk$, it is guaranteed only for cones in $\CUbk$. Later in Assumption \ref{oo}, we will require a kind of homogeneity, namely, $\Obk=\OUbk$ : excitations which can be localized to any $\ld\in \CUbk$ can be localized to any cone.}

For a representation $\rho$ of $\abd$ on $\hbd$ and $\Lambda_r\in \Crbd$,
we define sets of unitaries
\begin{align}
\begin{split}
&\Vbd{\rho}{\Lambda_r}:=
\left\{
V_{\rho\Lambda_r}\in\caU\lmk \bl\rmk
\mid \left.\Ad\lmk V_{\rho\Lambda_r}\rmk\circ\rho\right\vert_{\caA_{\Lambda_r^c\cap\hu}}
=\left.\pbd \right\vert_{\caA_{\Lambda_r^c\cap\hu}}
\right\}\\
&\VUbd{\rho}{\Lambda_r}:=
\left\{
V_{\rho\Lambda_r}\in\caU\lmk \fbd\rmk
\mid \left.\Ad\lmk V_{\rho\Lambda_r}\rmk\circ\rho\right\vert_{\caA_{\Lambda_r^c\cap\hu}}
=\left.\pbd \right\vert_{\caA_{\Lambda_r^c\cap\hu}}
\right\}.
\end{split}
\end{align}
We denote by $\Orbd$ the set of all representations $\rho$ of $\abd$ on $\hbd$,
which has a non-empty $\Vbd{\rho}{\Lambda_r}$ for any $\Lambda_r\in \Crbd$.
We denote by $\OrUbd$ the set of all representations $\rho$ of $\abd$ on $\hbd$,
which has a non-empty $\VUbd{\rho}{\Lambda_r}$ for any $\Lambda_r\in \Crbd$.
Note that $\OrUbd\subset \Orbd$.
Physically, objects in $\OrUbd$ (resp. $\Orbd$) correspond to low energy excitations 
which can be localized into the area of the form Figure \ref{ikada},
with {\it operations using only the bulk} (resp. {\it operations using the left boundary}).

For $\Lambda_{r0}\in \Crbd$, we denote by $\Orbdl$
(resp. $\OrUbdl$)
the set of all $\rho\in \Orbd$ (resp. $\rho\in \OrUbd$) such that $\rho\vert_{\caA_{\Lambda_{r0}^c}}=\pbd\vert_{\caA_{\Lambda_{r0}^c}}$.
\change{The excitation is localized in $\Lambda_{r0}$.}

For a representation $\rho$ of $\abd$ on $\hbd$ and $\Lambda\in \CUbk$,
we define sets of unitaries
\begin{align}
\begin{split}
&\Vbu{\rho}{\Lambda}:=
\left\{
V_{{\rho}\Lambda}\in\caU\lmk \fbd\rmk
\mid \left.\Ad\lmk V_{\rho\Lambda}\rmk\circ\rho\right\vert_{\caA_{\Lambda^c\cap\hu}}
=\left.\pbd \right\vert_{\caA_{\Lambda^c\cap\hu}}
\right\}.
\end{split}
\end{align}
We denote by $\Obu$ the set of all representations $\rho$ of $\abd$ on $\hbd$,
which has a non-empty $\Vbu{\rho}{\Lambda}$ for any $\Lambda\in \CUbk$.
Physically, objects in $\Obul$ correspond to low energy excitations 
which can be localized into the area of the form Figure \ref{fune} and Figuare \ref{boat},
with {\it operations using only the bulk}.
For $\lz\in \CUbk$, we denote by
$\Obul$ the set of all $\rho\in \Obu$ such that
$\rho\vert_{\caA_{\lhucz}}=\pbd\vert_{\caA_{\lhucz}}$.
\change{(The excitation is localized in $\ld_0\cap\hu$.)}
For any $\lm r\in\Crbd$, there exists a $\ld\in \CUbk$ such that
$(\ld,\lm r)\in\pc$.
Then we have $\Vbu\rho\ld\subset\Vbd\rho{\lm r}$,
and we have $\Obul\subset \Orbdl$
for $(\lz,\lzr)\in \pc$.


For each $\rho,\sigma\in \OUbk$,
the set of intertwiners between them are denoted by
\begin{align}
(\rho,\sigma):=
\left\{
R\in \caB(\hbk) \mid R\rho(A)=\sigma(A)R,\;\; \text{for all}\;\; A\;\in \abk
\right\}.
\end{align}
For any $\rho,\sigma\in \Orbd$, we set
\begin{align}
\begin{split}
&(\rho,\sigma)_l:=
\left\{
R\in \bl \mid R\rho(A)=\sigma(A)R,\;\; \text{for all}\;\; A\;\in \abd
\right\},\\
&(\rho,\sigma)_U:=
\left\{
R\in \fbd \mid R\rho(A)=\sigma(A)R,\;\; \text{for all}\;\; A\;\in \abd
\right\}.
\end{split}
\end{align}
Note that $(\rho,\sigma)_U\subset (\rho,\sigma)_l$.

For representations $\rho,\sigma$ of $\abd$ on $\hbd$, we write  $\rho\simeq_l\sigma$
(resp. $\rho\simeq_U\sigma$)
if there exists a unitary $U\in\bl$ (resp. $U\in\fbk$ ) such that $\rho=\Ad U\sigma$.
It is important to note that $\simeq_l$ is coarser than $\simeq_U$.
Physically, we understand $\rho\simeq_U\sigma$ means we can map $\rho$ to $\sigma$
using only ``bulk operation''.
On the other hand, $\rho\simeq_l\sigma$ can require ``truely boundary operation'' to map
$\rho$ to $\sigma$.
\change{
\begin{ex}\label{toric}
Let us consider the smooth boundary of the Toric code studied in \cite{wa}.
For reader's convenience, we briefly recall the morphisms introduced in \cite{wa}.
See \cite{wa} Section 6 for more detailed information.
For an infinite path $\gamma$ on the dual lattice and an infinite path $\gamma'$ 
on the lattice,
let $\rho_\gamma^X$, $\rho_{\gamma'}^Z$ be the endomorphisms given by
adjoints of  
products of Pauli $X$ along $\gamma$, products of Pauli $Z$ along $\gamma'$ respectively.
The representations $\pi_\gamma^X:=\pbd\rho_\gamma^X$,
$\pi_{\gamma'}^Z:=\pbd\rho_{\gamma'}^Z$ correspond to $m$-particle and $e$-particle, respectively.

Let $\ld,\ld_1\in \CUbk$ and ${\gamma'},\gamma_1'$ be infinite paths on the lattice
in $\ld\cap\hu$, $\ld_1\cap\hu$ respectively.
Note that there exists a cone $\ld_2\in \CUbk$ such that $\ld\cup\ld_1\subset \ld_2$.
In \cite{wa}, Wallick constructed a transporter from $\pi_{\gamma'}^Z$ to $\pi_{\gamma_1'}^Z$ 
as follows. :
Let $\zeta_n'$, $n\in\bbN$ be a sequence of paths on the lattice in $\ld_2\cap\hu$
from the $n$-th vertex of $\gamma'$ to the $n$-th vertex of $\gamma_1'$.
We take $\zeta_n'$ so that the distances from $\zeta_n'$ to the starting sites of 
$\gamma'$, $\gamma_1'$ go to infinity as $n\to\infty$. (See Figure \ref{same})
Then, as was shown in \cite{wa},  the weak limit $V$ of 
$\pbd\lmk \Gamma_{\gamma_n'}^Z\Gamma_{\zeta_n'}^Z \Gamma_{\gamma_{1n}'}^Z\rmk$
exists. 
Here, $\gamma_n'$ and $\gamma_{1n}'$ are the first $n$-step of 
$\gamma',\gamma_1'$  and $\Gamma_{\gamma_n'}^Z$, 
$\Gamma_{\zeta_n'}^Z$, $ \Gamma_{\gamma_{1n}'}^Z$
are the product of Pauli $Z$ along $\gamma_n'$, $\zeta_n'$, $\gamma_{1n}'$.
This $V$ is  a unitary, and
we have $\pi_{\gamma'_1}^Z=\Ad(V)\pi_{\gamma'}^Z$.
Because all of $\gamma_n'$, $\zeta_n'$, $\gamma_{1n}'$ are in
$\ld_2$, the operator  
$ \Gamma_{\gamma_n'}^Z\Gamma_{\zeta_n'}^Z \Gamma_{\gamma_{1n}'}^Z$
belongs to $\caA_{\ld_2\cap\hu}$. Therefore, we have
\begin{align}
\begin{split}
V={\mathop{\mathrm {w-}}}\lim_{n\to\infty}\pbd\lmk \Gamma_{\gamma_n'}^Z\Gamma_{\zeta_n'}^Z \Gamma_{\gamma_{1n}'}^Z\rmk
\in \pbd\lmk\caA_{\ld_2\cap\hu}\rmk''\subset\fbd.
\end{split}
\end{align}
From this and
\begin{align}
\begin{split}
\Ad(V)\pi_{\gamma'}^Z\vert_{\caA_{\ld_1^c\cap \hu}}=\pi_{\gamma'_1}^Z\vert_{\caA_{\ld_1^c\cap\hu}}=
\pbd\vert_{\caA_{\ld_1^c}},
\end{split}
\end{align}
we obtain $V\in \Vbu{\pi_{\gamma'}^Z}{{\Lambda}_1\cap\hu}$

In \cite{wa}, it is also shown that the $m$-particle is unitarily equivalent to the vacuum:
Let us consider an infinite path $\gamma$ on the dual lattice in an intersection of $\hu$
and $\ld\in \CUbk$, starting at the boundary.
Note that there is a cone $\ld_l\in \Clbd$ including $\ld$.
Let $\gamma_n$, $n\in\bbN$ be the path consisting of the first $n$ bonds
of $\gamma$.
 Let $\zeta_n$ be a
 path on the dual lattice in $\ld_l$ from the end of $\gamma_n$ to the boundary $(a_n,0)$.
 (See Figure \ref{kurage})
 We require the distance from $\zeta_n$ to the starting point of $\gamma$ goes to
 infinity as $n\to\infty$. In particular, $a_n\to-\infty$ as $n\to\infty$.
 Then, as was shown in \cite{wa},  the weak limit $V$ of 
$\pbd\lmk \Gamma_{\gamma_n}^X\Gamma_{\zeta_n}^X \rmk$
exists. 
Here,  $\Gamma_{\gamma_n}^X$, 
$\Gamma_{\zeta_n}^X$
are the product of Pauli $X$ along $\gamma_n$, $\zeta_n$.
This $V$ is  a unitary, and
we have $\pi_{\gamma}^X=\Ad(V)\pbd$.
Because both of $\gamma_n$, $\zeta_n$ are in
$\ld_l$, the operator  
$ \Gamma_{\gamma_n}^X\Gamma_{\zeta_n}^X$
belongs to $\caA_{\ld_l\cap\hu}$. Therefore, we have
\begin{align}
\begin{split}
V={\mathop{\mathrm {w-}}}\lim_{n\to\infty}
\pbd\lmk \Gamma_{\gamma_n}^X\Gamma_{\zeta_n}^X \rmk
\in \pbd\lmk\caA_{\ld_l\cap\hu}\rmk''\subset\caB_l.
\end{split}
\end{align}
This means $\pi_\gamma^X\simeq_l\pbd$.
Note also that this $V$ belongs to $\pbd\lmk\caA_{\ld^c\cap \hu}\rmk'$
because
\begin{align}
\begin{split}
\pbd\vert_{\caA_{\ld^c\cap\hu}}
=\Ad V^*\pi_\gamma^X\vert_{\caA_{\ld^c\cap\hu}}
=\Ad V^* \pbd \vert_{\caA_{\ld^c\cap\hu}},
\end{split}
\end{align}
since $\gamma$ is inside of $\ld$.
An important fact is that 
\begin{align}
\begin{split}
\pbd\lmk\caA_{\ld\cap\hu}\rmk''\subsetneq \pbd\lmk\caA_{\ld^c\cap \hu}\rmk'
\end{split}
\end{align}
by Proposition 12.7 of \cite{wa}. The right-hand side is truely bigger than the left-hand side. 
Note that the left-hand side $\pbd\lmk\caA_{\ld\cap\hu}\rmk''$
is included in $\caF$.
In fact as we will see in Remark \ref{bToric},
$\pi_\gamma^X$ and $\pbd$ are not equivalent with respect to $\simeq_\caF$.
In particular, $V\notin\caF$ and $V\in \pbd\lmk\caA_{\ld^c\cap \hu}\rmk'\setminus \pbd\lmk\caA_{\ld\cap\hu}\rmk'$.

\end{ex}
\begin{figure}[htbp]
  \begin{minipage}[b]{0.5\linewidth}
    \centering
    \includegraphics[width=5cm]{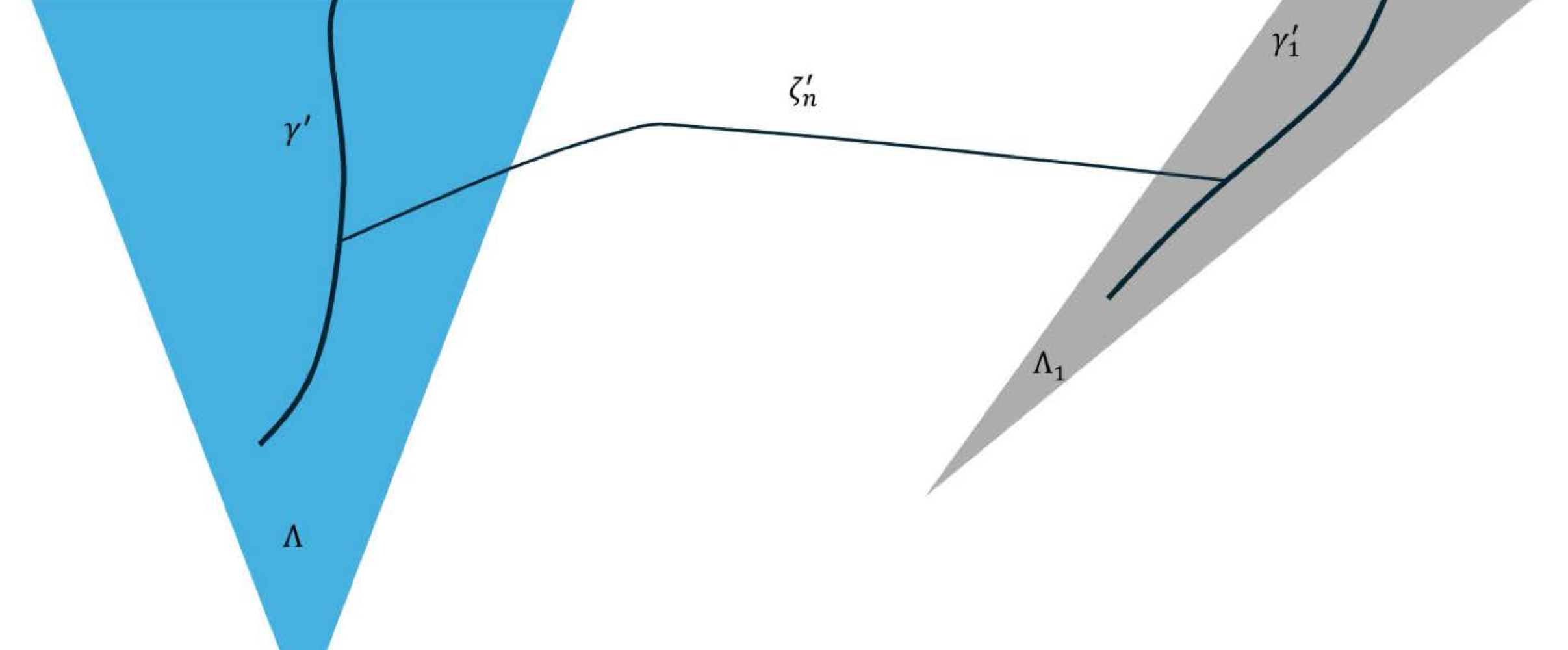}\\
    \quad\quad \quad\quad\quad \quad\quad \quad\quad\quad\caption{$\pi_{\gamma'_1}^Z=\Ad(V)\pi_{\gamma'}^Z$}
    \label{same}
  \end{minipage}
  \begin{minipage}[b]{0.5\linewidth}
    \centering
    \includegraphics[width=3cm]{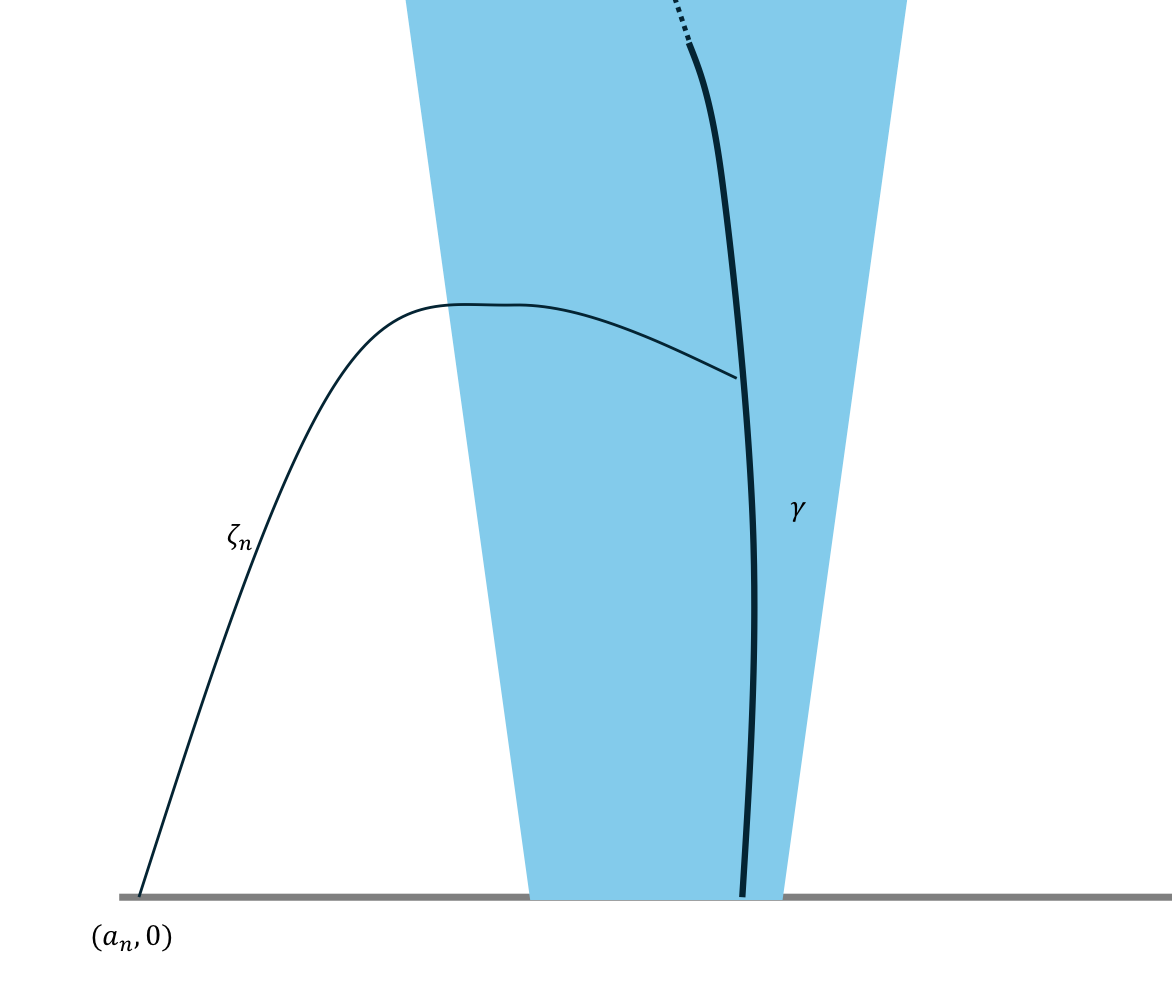}
   \quad\quad \quad\quad\quad\quad\quad \quad\quad\quad \caption{$\pi_{\gamma}^X=\Ad(V)\pbd$}
    \label{kurage}
  \end{minipage}
\end{figure}

}

\subsection{Assumptions}
Let us consider Setting \ref{setting2}.
We introduce our assumptions in this section.
First of all, we consider the situation that $\obd$ is a boundary state of $\obk$.
\begin{assum}\label{assump7}
The state $\obd$ is a boundary state of $\obk$.
Namely, for any $\Lambda\in\CUbk$, representations 
$\pbd\vert_{\caA_{\lhu}}$ and $\pbk\vert_{\caA_{\lhu}}$ are quasi-equivalent.
\end{assum}

\change{In the DHR theory, the Haag duality is assumed.
In the quantum spin setting, the Haag duality is defined as follows:
a representation $\pi$ of $\caA_{\bbZ^2}$ satisfies the Haag duality if
\begin{align}
\pi\lmk\caA_{\ld^c}\rmk'=\pi\lmk\caA_{\ld}\rmk''
\end{align}
for any cone $\ld$.\\
Note that $\pi(\caA_{\ld})''\subset \pi\lmk\caA_{\ld^c}\rmk'$ always holds.
The requirement here is that they are actually equal.
}
We introduced the following approximate  Haag duality in \cite{MTC}.
\change{This is a relaxation of the Haag duality.}
\begin{assum}\label{assum3}
Consider the setting in subsection \ref{setting2}.
For any $\varphi\in (0,\pi)$ and 
 $\varepsilon>0$ with
$2\varphi+4\varepsilon<2\pi$,
there is some $R^{\mopbk}_{\varphi,\varepsilon}\ge 0$ and decreasing
functions $f^{\mopbk}_{\varphi,\varepsilon,\delta}(t)$, $\delta>0$
on $\bbR_{\ge 0}$
with $\lim_{t\to\infty}f^{\mopbk}_{\varphi,\varepsilon,\delta}(t)=0$
such that
\begin{description}
\item[(i)]
for any cone $\Lambda$ with $|\arg\Lambda|=2\varphi$, there is a unitary 
$U^{\mopbk}_{\Lambda,\varepsilon}\in \caU(\hbk)$
satisfying
\begin{align}\label{lem7p}
\pbk\lmk\caA_{\Lambda^c}\rmk'\subset 
\Ad\lmk U^{\mopbk}_{\Lambda,\varepsilon}\rmk\lmk 
\pbk\lmk \caA_{\lmk \Lambda-R^{\mopbk}_{\varphi,\varepsilon}\bm e_\Lambda\rmk_\varepsilon}\rmk''
\rmk,
\end{align}
and 
\item[(ii)]
 for any $\delta>0$ and $t\ge 0$, there is a unitary 
 $\tilde U^{\mopbk}_{\Lambda,\varepsilon,\delta,t}\in \pbk\lmk \caA_{\Lambda_{\varepsilon+\delta}-t\bm e_{\Lambda}}\rmk''$
 satisfying
\begin{align}\label{uappro}
\lV
U^{\mopbk}_{\Lambda,\varepsilon}-\tilde U^{\mopbk}_{\Lambda,\varepsilon,\delta,t}
\rV\le f^{\mopbk}_{\varphi,\varepsilon,\delta}(t).
\end{align}
\end{description}
\end{assum}
\change{
Note that if the Haag duality holds, then the approximate Haag duality holds with
$U^{\mopbk}_{\Lambda,\varepsilon}=\tilde U^{\mopbk}_{\Lambda,\varepsilon,\delta,t}=\unit$
and $f^{\mopbk}_{\varphi,\varepsilon,\delta}(t)=0$, $R^{\mopbk}_{\varphi,\varepsilon}=0$.
}

\change{In Figure \ref{AH}, the dark colored cone is $\ld$,
 and the light colored cone is ${\Lambda_{\varepsilon+\delta}-t\bm e_{\Lambda}}$.
 The unitary $\tilde U^{\mopbk}_{\Lambda,\varepsilon,\delta,t}$ is localized in the light cone area.
 (i) combined with (ii) tells us the following:
 unlike the Haag duality, $\pbk(\caA_{\ld^c})'$ may not be equal to
 $\pbk(\caA_{\ld})''$. However, if we consider the enlarged cone
 ${\Lambda_{\varepsilon+\delta}-t\bm e_{\Lambda}}$, then
 $\pbk(\caA_{\ld^c})'$ is ``approximately'' included in 
 the von Neumann algebra corresponding to that cone $\pbk\lmk \caA_{\Lambda_{\varepsilon+\delta}-t\bm e_{\Lambda}}\rmk''$.
 More precisely, for any $a\in \pbk(\caA_{\ld^c})'$,
 there exists a $b\in \pbk\lmk \caA_{\Lambda_{\varepsilon+\delta}-t\bm e_{\Lambda}}\rmk''$
 with $\lV b\rV=\lV a\rV$
 such that
 \begin{align}
 \begin{split}
 \lV b-a\rV\le 2 f^{\mopbk}_{\varphi,\varepsilon,\delta}(t) \lV a\rV.
 \end{split}
 \end{align}
 The advantage of the approximate Haag duality over the Haag duality is that we know that the approximate Haag duality is stable under local unitary operations \cite{MTC}.
 }

The following Lemma is immediate from the definition.
\begin{lem}\label{lem11}
Consider the setting in subsection \ref{setting2} and assume Assumption \ref{assum3}.
Then for any $\Lambda\in \CUbk$, we have
$\pbk\lmk\caA_{\Lambda^c}\rmk'\subset\gu$.
\end{lem}
The following two assumptions are boundary version of the approximate Haag duality.
\begin{assum}\label{assum80}
For any $\varphi\in (0,\frac\pi 2)$, $\varepsilon>0$ with $2\varphi+2\varepsilon<\pi$,
there is some $R^{(r)\mopbd}_{\varphi,\varepsilon}>0$ and decreasing
functions $f^{(r)\mopbd}_{\varphi,\varepsilon,\delta}(t)$, $\delta>0$
on $\bbR_{\ge 0}$
with $\lim_{t\to\infty}f^{(r)\mopbd}_{\varphi,\varepsilon,\delta}(t)=0$
such that
\begin{description}
\item[(i)]
for any cone $\Lambda_r=\lr a \varphi\in \Crbd$, there is a unitary 
$U^{(r)\mopbd}_{\Lambda_r,\varepsilon}\in \caU(\hbd)$
satisfying
\begin{align}\label{lem7p}
\pbd\lmk\caA_{\lrhu}\rmk'\subset 
\Ad\lmk U^{(r)\mopbd}_{\Lambda_r,\varepsilon}\rmk\lmk 
\pbd\lmk \caA_{ \lr{a-R^{(r)\mopbd}_{\varphi,\varepsilon}}{\varphi+\varepsilon}}\rmk''
\rmk,
\end{align}
and 
\item[(ii)]
 for any $\delta>0$ and $t\ge 0$, there is a unitary 
 $\tilde U^{(r)\mopbd}_{\Lambda_r,\varepsilon,\delta,t}\in \pbd\lmk 
 \caA_{
  \lr{a-t}{\varphi+\varepsilon+\delta}
 }\rmk''\cap \fbd$
 satisfying
\begin{align}\label{uappro}
\lV
U^{(r)\mopbd}_{\Lambda_r,\varepsilon}-\tilde U^{(r)\mopbd}_{\Lambda_r,\varepsilon,\delta,t}
\rV\le f_{\varphi,\varepsilon,\delta}^{(r)\mopbd}(t).
\end{align}
\end{description}
Note in particular, we have $U^{(r)\mopbd}_{\Lambda_r,\varepsilon}\in \fbd$.
\end{assum}
\change{In Figure \ref{AHr}, the dark-colored cone indicates $\ld_r$, and the light-colored cone
indicates $ \lr{a-t}{\varphi+\varepsilon+\delta}$.
The unitary $\tilde U^{(r)\mopbd}_{\Lambda_r,\varepsilon,\delta,t}$ is localized in the light-colored cone.
The new ingredient compared to the Assumption \ref{assum3} is that we require the unitary
$\tilde U^{(r)\mopbd}_{\Lambda_r,\varepsilon,\delta,t}$ to be in $\fbd$
, not only in $\pbd\lmk\caA_{\lr{a-t}{\varphi+\varepsilon+\delta}}\rmk''$.}
\begin{figure}[htbp]
  \begin{minipage}[b]{0.5\linewidth}
    \centering
    \includegraphics[width=5cm]{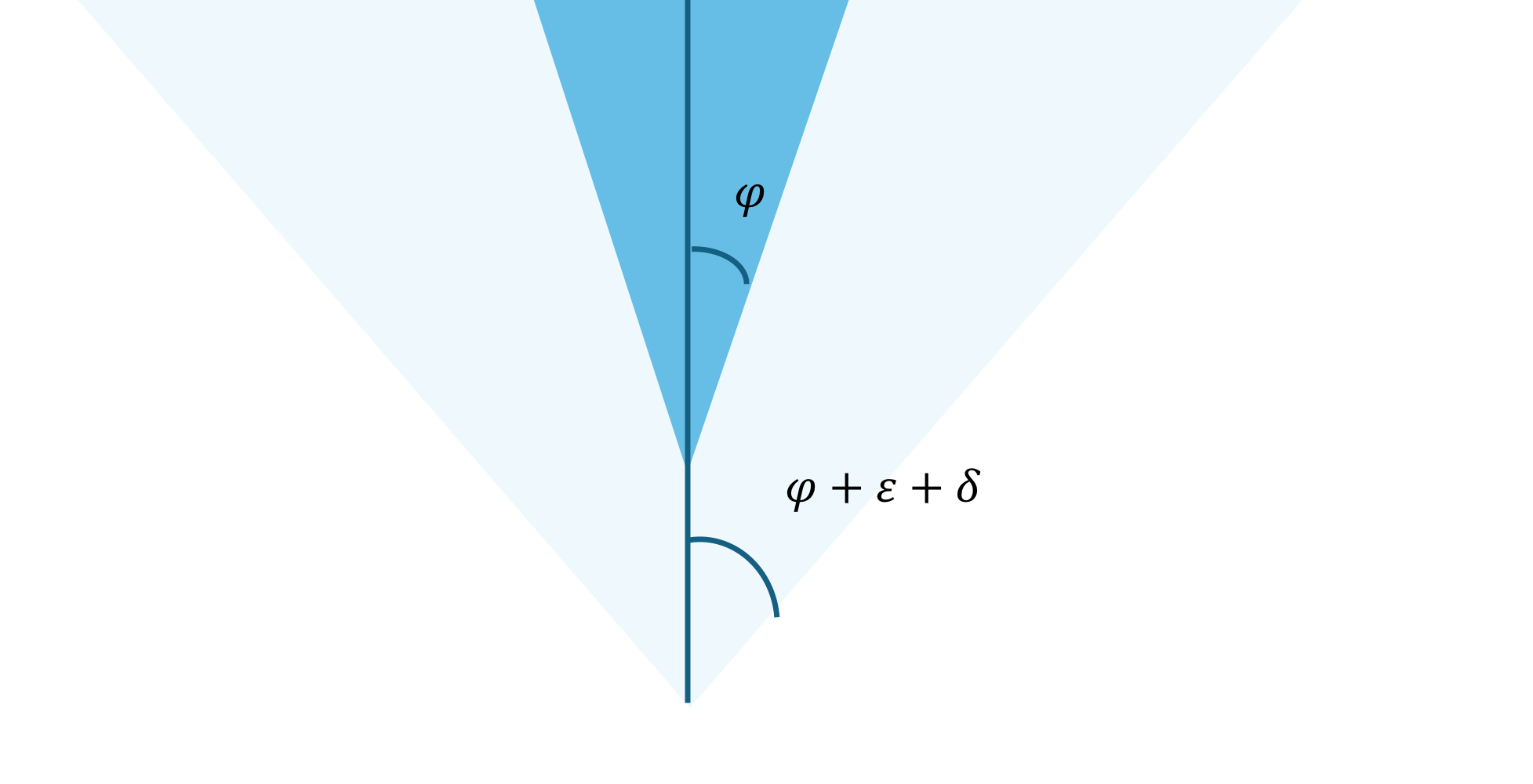}\\
\caption{}   \label{AH}
   \end{minipage}
  \begin{minipage}[b]{0.5\linewidth}
    \centering
    \includegraphics[width=3cm]{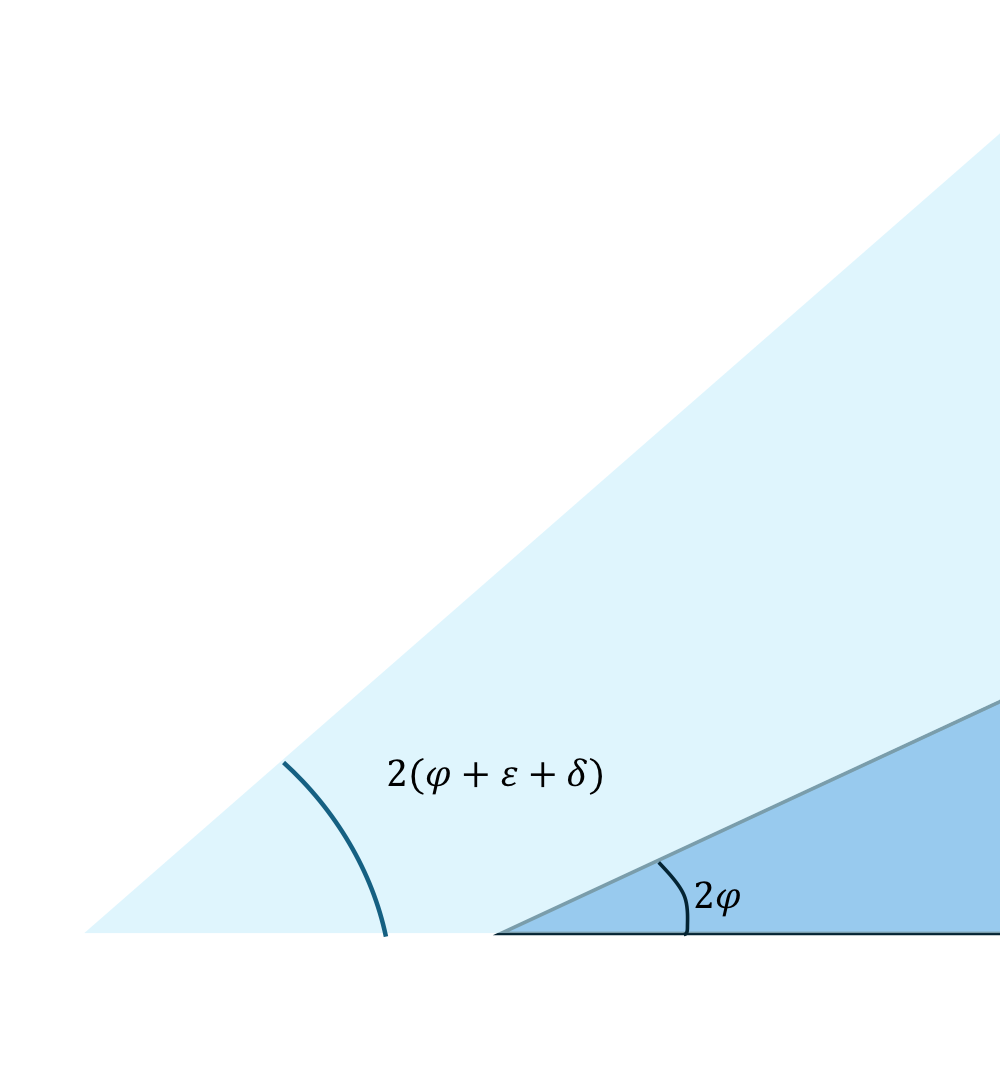}
    \caption{}\label{AHr}
  \end{minipage}
\end{figure}

\change{Here is the left version.}
\begin{assum}\label{assum80l}
For any $\varphi\in (0,\frac\pi 2)$, $\varepsilon>0$ with $2\varphi+2\varepsilon<\pi$,
there is some $R^{(l)\mopbd}_{\varphi,\varepsilon}>0$ and decreasing
functions $f^{(l)\mopbd}_{\varphi,\varepsilon,\delta}(t)$, $\delta>0$
on $\bbR_{\ge 0}$
with $\lim_{t\to\infty}f^{(l)\mopbd}_{\varphi,\varepsilon,\delta}(t)=0$
such that
\begin{description}
\item[(i)]
for any cone $\Lambda_l=\lef a \varphi\in \Clbd$, there is a unitary 
$U^{(l)\mopbd}_{\Lambda_l,\varepsilon}\in \caU(\hbd)$
satisfying
\begin{align}\label{lem7p}
\pbd\lmk\caA_{(\lm l)^c\cap \hu}\rmk'\subset 
\Ad\lmk U^{(l)\mopbd}_{\Lambda_l,\varepsilon}\rmk\lmk 
\pbd\lmk \caA_{ \lef{\lmk a+R^{(l)\mopbd}_{\varphi,\varepsilon}\rmk}{(\varphi+\varepsilon)}}\rmk''
\rmk,
\end{align}
and 
\item[(ii)]
 for any $\delta>0$ and $t\ge 0$, there is a unitary 
 $\tilde U^{(l)\mopbd}_{\Lambda_l,\varepsilon,\delta,t}\in \pbd\lmk 
 \caA_{
  \lef{a+t}{(\varphi+\varepsilon+\delta)}
 }\rmk''\cap \fbd$
 satisfying
\begin{align}\label{uappro}
\lV
U^{(l)\mopbd}_{\Lambda_l,\varepsilon}-\tilde U^{(l)\mopbd}_{\Lambda_l,\varepsilon,\delta,t}
\rV\le f^{(l)\mopbd}_{\varphi,\varepsilon,\delta}(t).
\end{align}
\end{description}
Note in particular, we have $U^{(l)\mopbd}_{\Lambda,\varepsilon}\in \fbd$.
\end{assum}
\change{
\begin{rem}
The non-approximate versions of Assumption \ref{assum80} and Assumption \ref{assum80l} are
\begin{align}
\begin{split}
&\pbd\lmk\caA_{(\lm r)^c\cap\hu }\rmk'
=\pbd\lmk\caA_{\lm r}\rmk'',\quad \lm r\in\Crbd,\\
&\pbd\lmk\caA_{(\lm l)^c\cap\hu }\rmk'
=\pbd\lmk\caA_{\lm l}\rmk'',\quad \lm l\in\Clbd,
\end{split}
\end{align}
respectively.
In \cite{wa}, Wallick showed this non-approximate version for the smooth boundary of the Toric code.
\end{rem}
}

In order to have direct sum and subobjects, we need the following assumption.
\begin{assum}\label{aichi}
For any $\Lambda\in \CUbk$, $\pbk(\caA_{\Lambda})''$ is an infinite factor.
\end{assum}
The boundary version is the following.
\begin{assum}\label{wakayama}
For any $\Lambda\in \CUbk$, $\pbd(\caA_{\Lambda\cap\hu})''$ is an infinite factor.
\end{assum}
From Assumption \ref{wakayama}, we can derive infiniteness of von Neumann algebras of larger regions.
\begin{lem}
Assume Assumption \ref{wakayama}.
Then for any $\lm l\in \Clbd$, $\pbd(\caA_{\lm{l}})''$ is an infinite factor.
\end{lem}
\begin{proof}
Because $\obd$ is pure, $\pbd(\caA_{\lm{l}})''$ is a factor  for any $\lm l\in \Clbd$.
For any $\lm l\in \Clbd$, there exists a $\ld\in \CUbk$ such that
$\ld\cap\hu\subset\lm l$, hence $\pbd(\caA_{\Lambda\cap\hu})''\subset \pbd(\caA_{\lm{l}})''$.
Therefore, if Assumption \ref{wakayama} holds, then $\pbd(\caA_{\lm{l}})''$ is an infinite factor.
\end{proof}
We also note the following.
\begin{lem}\label{daidai}
If Assumption \ref{aichi} and Assumption \ref{assump7} hold,
then Assumption \ref{wakayama} holds. 
\end{lem}
\begin{proof}
This is because
$\Lambda\setminus \lmk \Lambda\cap \hu\rmk$ is finite
for any $\ld\in\CUbk$.
\end{proof}
Combining this and Lemma 5.3 Lemma 5.5 of \cite{MTC}, for boundary states of gapped ground states, Assumption \ref{aichi}
and Assumption \ref{wakayama} hold automatically.
\begin{lem}\
Suppose that $\obk$ is a gapped ground state in the sense it satisfies
the gap inequality (5.1)of \cite{MTC}, and $\obd$ is a boundary state of $\obk$.
Then Assumption \ref{aichi} and Assumption \ref{wakayama} hold.
\end{lem}

The following assumption means that in $\obk$, any quasi-particles 
 in the upper half plane can freely move to the lower half plane.
\begin{assum}\label{oo}
$\OUbk=\Obk$.
\end{assum}
In this paper, we consider such a homogeneous situation.

\change{
\begin{ex}\label{exas}
Assumptions in this subsection are satisfied by the Toric code:
By \cite{bravyi2010topological}, \cite{N2} or with the argument there, we know that the condition 
in Definition \ref{def104} holds for the Toric code.
From Proposition \ref{hotaru}, this guarantees Assumption \ref{assump7}
for the smooth/rough boundaries of Toric code.
Thanks to \cite{N1} and \cite{wa}, we know that Assumption \ref{assum3},
Assumption \ref{assum80} and Assumption \ref{assum80l} are also satisfied.
Assumption \ref{aichi} and Assumption \ref{wakayama} can be found there as well,
although they also follow from the gap condition \cite{MTC}.
Assumption \ref{oo} follows from the argument in \cite{naaijkens2013kosaki} :
the argument  following
\cite{kawahigashi2001multi} already gives a bound for the number of sectors in $\OUbk$.
But as was shown in \cite{naaijkens2013kosaki} this number is already saturated by
sectors coming from $\Obk$. (Recall $\Obk\subset \OUbk$.)

\end{ex}
}
\subsection{Results}
We consider the setting in subsection \ref{setting2}.
The proof in  \cite{MTC} shows the following.
\begin{thm}Consider the setting in subsection \ref{setting2}.
Under Assumption \ref{assum3} and Assumption \ref{aichi},
 $\Obkl$ with their morphisms $(\rho,\sigma)$,
$\rho,\sigma\in \Obkl$
forms a braided $C^*$-tensor category $\Cabkl$.
\end{thm}
See \cite{NT} for the definition of braided $C^*$-tensor category and related notions.

Analogous to this theorem, we can derive various $C^*$-tensor categories
from our boundary state $\obd$. We can consider different categories
depending on our choices of objects (representations of $\caA_{\hu}$) and
morphisms (intertwiners between representations).
Under 
Assumption \ref{wakayama} we can show that
all of the following categories are $C^*$-tensor categories
(Theorem \ref{naha}, Theorem \ref{okinawa}, Theorem \ref{nagoya}).
\begin{enumerate}
\item The category $\Carbdl$ with objects $\Orbdl$
and morphisms $(\rho,\sigma)_l$ between objects $\rho,\sigma\in \Orbdl$.
\item The category $\CarUbdl$ with objects $\OrUbdl$
and morphisms $(\rho,\sigma)_U$ between objects
$\rho,\sigma\in \OrUbdl$.
\item The category $\Cabul$ with objects $\Obul$
and morphisms $(\rho,\sigma)_U$ between objects $\rho,\sigma\in \Obul$.
\end{enumerate}

For each $\rho\in \OrUbdl$ we can define 
a family of natural isomorphisms
$\hb - \rho:\Obul\to (-\otimes \rho,\rho\otimes -)_U$, satisfying the half-braiding equation
(Proposition \ref{lem40}, Lemma \ref {lem41}, Lemma \ref{inu}, Lemma \ref{panama}),
under Assumption \ref{assum80} Assumption \ref{assum80l}.
This $\hb -\rho$  satisfies some asymptotic property (Lemma \ref{neko}). 
When we restrict this to $\Obul$, 
$\Cabul$ becomes braided $C^*$-tensor category (Theorem \ref{matsuyama}).
The category $\Cabul$ can be regarded as a center of $\CarUbdl$ in a certain sense.
See Appendix \ref{Orcen}.

Using the condition that $\obd$ is a boundary of $\obk$ (Assumption \ref{assump7}),
we can show the following. See Corollary \ref{kanazawa} for a more detailed statement.
\begin{ithm}
The braided $C^*$-tensor category $\Cabul$ and the bulk braided $C^*$-tensor category
$\Cabkl$ are equivalent.
\end{ithm}
Hence at the boundary, we have a ``copy'' of the bulk theory.
However, $\Cabul$ is not preferable as the boundary theory.
This is because the morphisms are too limited there, i.e.,
we allow ourselves to use only $\fbd$.
In particular, as we saw in Example \ref{toric} for toric code
 with rough boundary, $m$ particle is {\it not} equivalent to the vacuum
  with respect to
$\simeq_U$, while
it {\it is} equivalent to the vacuum
  with respect to
$\simeq_l$.
Therefore, we prefer to use $\bl$ as morphisms.
Hence in section \ref{drinfeld} we introduce a category  $\caM$ with object $\Obul$, 
and morphisms between objects $\rho,\sigma\in \Obul$
defined by
 $\Mor_{\caM}(\rho,\sigma):=(\rho,\sigma)_l$.
 The category $\caM$ is a tensor category.
 However, it is {\it not} a $C^*$-tensor category, in the sense that we cannot
 construct a subobject in general.
 In order to overcome this difficulty, we introduce the idempotent completion $\tilde\caM$
 of $\caM$.
Then we obtain the following (see Proposition \ref{mtcat}).
\begin{iprop}
This $\tilde\caM$ is a $C^*$-tensor category
\end{iprop}
We regard this $\tilde \caM$ as our boundary theory.
For this $C^*$-tensor category $\tilde\caM$,
we introduce  {\it the Drinfeld center with an asymptotic constraint}
$\caZ_a(\tilde \caM)$ of $\tilde \caM$.
This is the same as the usual Drinfeld center  (Definition \ref{hato}), except that the half-braidings
are required to satisfy the asymptotic property satisfied by
our $\iota^{\llz}$.
Now we assume that the braiding of $\Cabkl$ is non-trivial in the following sense :
if two quasi-particles in the bulk have the asymptotically unitarily equivalent braiding, then they
  are equivalent. For more precise condition, see Assumption \ref{raichi} and Remark \ref{dil}.
 \change{We regard this condition as ``a stable gapped boundary'' condition. See Remark \ref{gapmean}
 for the reasoning.}
Then we obtain the following.
\begin{ithm}
The Drinfeld center with the asymptotic constraint
$\zam$ is a braided $C^*$-tensor category equivalent to the bulk braided $C^*$-tensor category
$\Cabkl$.
\end{ithm}
\change{
\begin{ex}
From Example \ref{exas}, all the assumptions in this subsection are
satisfied in Toric code, and we may apply our theorem.
We will see from Example \ref{ebToric} and Theorem \ref{mainthm} that this equivalence
maps bulk objects $\pbk\rho_{\gamma'}^Z$ (with $\gamma'$ an infinite path in the upper half plane) to a boundary object
$(\pbd\rho_{\gamma'}^Z, \ti - -{\pbd\rho_{\gamma'}^Z}\unit)$
where $\tilde \iota$ is the half-braiding introduced in Lemma \ref{panda}.
This recovers the equivalence shown in \cite{wa}.
\end{ex}
}

\change{
\subsection{Comparisons with former works}
In \cite{wa}, Wallick (besides things we already mentioned about),
showed that the superselection sectors $\pi_\gamma^k$,
given by infinite paths of all types, and their direct sums form a fusion category $\Delta(\ld)$ (section 8 \cite{wa}).
In fact, as was shown in Theorem 12.9 of \cite{wa}, they are basically all the superselection sectors at the boundary. The proof of Theorem 12.9 of \cite{wa} relies on the index theory given by \cite{kawahigashi2001multi}, and requires quite a bit of information of the model.
The functor from the bulk to the boundary constructed in \cite{wa}
is given by a concrete formula, using the information of the path $\gamma$.
In our general framework, it is hard to talk about the index, and there is no path 
$\gamma$ that allows us to define the functor explicitly.
The thing we use instead to introduce our functor is
the quasi-equivalence given by Assumption \ref{assump7}.
This allows us to prove that there is a copy of the bulk theory at the boundary
Corollary \ref{kanazawa}.

Our main theorem proves  that the bulk theory is the Drinfeld center of the 
``stable gapped boundary'' theory, in our setting of superselection sectors.
Because the formulation of \cite{kitaev2012models} is
not given in terms of superselection sectors, we cannot make a direct connection
with our objects and objects studied in  \cite{kitaev2012models}.
However, assuming that the category given by superselection sectors (in our sense)
and that of \cite{kitaev2012models} are identical, we can relate our result to that 
of \cite{kitaev2012models}.
The identification is non-trivial though.
There is yet another approach to boundary theory \cite{jones2023local}, mostly focusing on systems with commuting Hamiltonians, but affording more informations for such cases. It would be interesting to see the relation between their approach and ours.
That might prove the identification we mentioned above.

}

\section{Boundary theories}
\change{In this section, we introduce braided $C^*$-tensor categories 
$\Carbdl$, $\CarUbdl$, $\Cabul$.
The proof follows the standard DHR theory \cite{BF} that is adapted to spin systems 
\cite{MTC} at the boundary. 
}

First, we extend each $\rho\in\Orbdl$ to an endomorphism of $\bl$.
\change{Note, as in the standard DHR theory, we cannot extend $\rho$ to the whole 
$\pbd(\caA_{H_U})''=\caB(\hbd)$ $\sigma$-weak continuously.
If that was possible, $\rho$ would be quasi-equivalent to $\pbd$.}
\begin{lem}\label{lem20}
Consider the setting in subsection \ref{setting2}.
Then for any $\rho\in\Orbd$, $\lzr\in\Crbd$ and $\Vrl\rho\lzr\in \Vbd\rho\lzr$,
there exists a unique $*$-homomorphism $\Tbdv\rho\lzr: \bl\to\caB(\hbd)$
such that
\begin{description}
\item[(i)]
$\Tbdv\rho\lzr\pbd=\Ad\lmk \Vrl\rho\lzr\rmk\rho$,
\item[(ii)]
$\Tbdv\rho\lzr$ is $\sigma$-weak continuous on
$\pbd(\caA_{\Lambda_l})''$ for all $\Lambda_l\in\Clbd$.
\end{description}
It satisfies $\Tbdv\rho\lzr(\bl)\subset \bl$ and defines an endomorphism
$\Tbdv\rho\lzr : \bl\to \bl$.
This endomorphism also satisfies 
\begin{align}\label{soto}
\left.\Tbdv\rho\lzr\right\vert_{\pbd\lmk\caA_{(\lzr)^c\cap\hu}\rmk''}
=\id_{\pbd\lmk\caA_{(\lzr)^c\cap\hu}\rmk''}.
\end{align}
If furthermore $\rho\in\OrUbd$ and $\Vrl\rho\lzr\in \VUbd\rho\lzr$
then we have $\Tbdv\rho\lzr(\fbd)\subset \fbd$.
\end{lem}
\begin{proof}
Each $\lm l\in\Clbd$ can be written as $\lm l=(\lm r)^c\cap \hu$
with some $\lm r=\lr a \varphi\in\Crbd$.
Set
\[
\kappa_{\lm l}:=\left\{
\Gamma_r\mid\Gamma_r=\lr b{\varphi_1},\;
a<b,\; 0<\varphi_1<\varphi
\right\}
\subset\Crbd.
\]
Then, as in Lemma 2.11 of \cite{MTC},
for any $\rho\in\Orbd$, $\lzr\in\Crbd$ and $\Vrl\rho\lzr\in \Vbd\rho\lzr$,
\begin{align}
T_\rho^{(0)}(x):=\Ad\lmk\Vrl\rho{\lzr}\Vrl\rho{K_{\lm l}}^*\rmk(x),\quad\text{if}\;
x\in \pbd\lmk\caA_{\lm l }\rmk'',\;\lm l\in\Clbd
\end{align}
defines a isometric $*$-homomorphism $T_\rho^{(0)}: \bl^{(0)}\to\bl$,
independent of the choice of $K_{\lm l}$, $\Vrl\rho{K_{\lm l}}$, $\lm l$.
As in Lemma 2.13 of \cite{MTC}, this $T_\rho^{(0)}$ extends to
the $\Tbdv\rho\lzr$ with the desired property.
The last statement is trivial from the definition above.
\end{proof}

\change{
\begin{rem}
If we consider the case $\rho\in\Orbdl$ and $\Vrl\rho\lzr=\unit\in  \Vbd\rho\lzr$, then this gives the extension 
$\Tbd\rho{\lzr}\unit$ of $\rho$.

\end{rem}
}
\begin{rem}
We extend the equivalence relations $\simeq_l$,  $\simeq_U$
to endomorphisms of $\bl$ in an obvious way.
\end{rem}
%

By the same proof as that of \cite{MTC} Lemma 2.14, 
for any $\lm {r1},\lm{r2}\in\Crbd$, $\rho\in \Orbd$, $\Vrl{\rho}{\lm {ri}}\in \Vbd\rho{\lm {ri}}$, $i=1,2$,
we have
 $\Tbdv\rho{{\lm {r2}}}=\Ad\lmk\Vrl{\rho}{\lm {r2}}\Vrl{\rho}{\lm {r1}}^*\rmk\Tbdv\rho{{\lm {r1}}}$.
  For $\lmr i\in\Crbd$, $\rho_i\in \Orbd$ (resp. $\rho_i\in \OrUbd$), $\Vrl{\rho_i}{\lmr i}\in \Vbd{\rho_i}{\lmr i}$
  (resp. $\Vrl{\rho_i}{\lmr i}\in \VUbd{\rho_i}{\lmr i}$), $i=1,2$
  with $\rho_1\simeq_l\rho_2$ 
(resp. $\rho_1\simeq_U\rho_2$) we have 
$\Tbdv{\rho_1}{\lmr 1}\simeq_l\Tbdv{\rho_2}{\lmr 2}$
(resp. $\Tbdv{\rho_1}{\lmr 1}\simeq_U\Tbdv{\rho_2}{\lmr 2}$).

\change{The extension allows us to consider compositions, which will give us the tensor structure. }

\begin{lem}\label{lem22}
Consider the setting in subsection \ref{setting2}, and let $\lzr\in \Crbd$.
Then for any $\rho,\sigma\in \Orbdl$, we have
\begin{align}\label{iwate}
\rho\circ_{\lzr} \sigma:=
\Tbd\rho{\lzr}\unit\;\Tbd\sigma{\lzr}\unit\;\pbd\in\Orbdl,
\end{align}
and this defines a composition $\circ_{\lzr} : \Orbdl\times \Orbdl \to \Orbdl$.
This composition $\circ_{\lzr}$ satisfies the following properties.
\begin{description}
\item[(i)]
For $\rho,\sigma\in \Orbdl$, we have $\Tbd{\rho\circ_{\lzr} \sigma}\lzr\unit
=\Tbd{\rho}\lzr\unit\Tbd{\sigma}\lzr\unit$.
\item[(ii)]
For $\rho,\sigma,\gamma\in \Orbdl$,
$\lmk \rho\circ_{\lzr} \sigma\rmk\circ_{\lzr} \gamma=\rho\circ_{\lzr}\lmk \sigma\circ_{\lzr} \gamma\rmk$.
\item[(iii)]
If $\rho,\sigma\in \OrUbdl$, then $\rho\circ_{\lzr} \sigma\in\OrUbdl $.
\item[(iv)] If $\llz\in\pc$ and $\rho,\sigma\in \Obul$, then 
${\rho\circ_{\lzr} \sigma}\in \Obul$.
\end{description}
\end{lem}
\begin{proof}
This follows from the same proof as that of Lemma 3.3, 3.6, 3.7 of \cite{MTC},
and definition of $\Obul$,$\OrUbdl$.
\kakunin{The proof of (\ref{iwate}), (i)(ii) are the same as that of Lemma 3.3, 3.6, 3.7 of \cite{MTC}.
(iii) follows because we have
\begin{align}
\Tbdv\rho{\lm r }\Tbdv\sigma{\lm r }\pbd=\Ad\lmk\Vrl{\rho}{\lm r}\rmk\Tbd\rho{{\lzr}}\unit
\Ad\lmk\Vrl{\sigma}{\lm r}\rmk\Tbd\sigma{{\lzr}}\unit\pbd
=\Ad\lmk\Vrl{\rho}{\lm r}\Tbd\rho{{\lzr}}\unit\lmk \Vrl{\sigma}{\lm r}\rmk\rmk
\Tbd\rho{{\lzr}}\unit
\Tbd\sigma{{\lzr}}\unit\pbd,
\end{align}
and in the setting of (iii), may take $\Vrl{\rho}{\lm r}\Tbd\rho{{\lzr}}\unit\lmk \Vrl{\sigma}{\lm r}\rmk\in \fbd$.
In the setting of (iv), for each $\ld\in\CUbk$, fix $\Vrl\rho\ld\in\Vbu\rho\ld$
and $\Vrl\sigma\ld\in\Vbu\sigma\ld$.
Then $\Vrl\rho\ld\Tbd\rho\lzr\unit \lmk \Vrl\sigma\ld\rmk\in \fbd$ and 
\begin{align}
\left. \Ad\lmk \Vrl\rho\ld\Tbd\rho\lzr\unit \lmk \Vrl\sigma\ld\rmk \rmk
\Tbd\rho\lzr\unit\Tbd\sigma\lzr\unit\pbd\right\vert_{\caA_{\ld^c\cap\hu}}
=\pbd\vert_{\caA_{\ld^c\cap\hu}}.
\end{align}
Hence ${\rho\circ_{\lzr} \sigma}\in \Obul$.}
\end{proof}
Let $\rho_i,\sigma_i\in \Orbd$ (resp $\rho_i,\sigma_i\in \OrUbd$), 
$\lm {ri},\lm {ri}'\in \Crbd$, $\Vrl{\rho_i}{\lm{ri}}\in \Vbd{\rho_i}{\lm{ri}}$,
$\Vrl{\sigma_i}{\lm{ri}'}\in \Vbd{\sigma_i}{\lm{ri}'}$
(resp $\Vrl{\rho_i}{\lm{ri}}\in \VUbd{\rho_i}{\lm{ri}}$,
$\Vrl{\sigma_i}{\lm{ri}'}\in \VUbd{\sigma_i}{\lm{ri}'}$), for $i=1,2$.
If $\rho_1\simeq_l\rho_2$ and $\sigma_1\simeq_l\sigma_2$
(resp. $\rho_1\simeq_U\rho_2$ and $\sigma_1\simeq_U\sigma_2$)
then we have
$\Tbdv{\rho_1}{\lm {r1}}\Tbdv{\sigma_1}{\lm {r1}'}\simeq_l \Tbdv{\rho_2}{\lm {2r}}\Tbdv{\sigma_2}{\lm {2r}'}$
(resp. $\Tbdv{\rho_1}{\lm {r1}}\Tbdv{\sigma_1}{\lm {r1}'}\simeq_U\Tbdv{\rho_2}{\lm {2r}}\Tbdv{\sigma_2}{\lm {2r}'}$), by the same proof as Lemma 3.4 of \cite{MTC}.
\begin{lem}\label{lem24}
Consider the setting in subsection \ref{setting2}.
For any $\rho_i,\rho_i'\in \Orbd$, 
$\lm {ri},\lm {ri}'\in \Crbd$, $\Vrl{\rho_i}{\lm{ri}}\in \Vbd{\rho_i}{\lm{ri}}$,
$\Vrl{\rho_i'}{\lm{ri}'}\in \Vbd{\rho_i'}{\lm{ri}'}$
and $R_i\in \lmk\Tbdv{\rho_i}{\lm {ri}}, \Tbdv{\rho_i'}{\lm {ri}'}\rmk_l$
 for $i=1,2$, 
 we have
 \begin{align}
 R_1\otimes R_2:=R_1\Tbdv{\rho_1}{\lm {r1}}\lmk R_2\rmk
 \in \lmk
 \Tbdv{\rho_1}{\lm {r1}}\Tbdv{\rho_2}{\lm {r2}},
  \Tbdv{\rho_1'}{\lm {r1}'}\Tbdv{\rho_2'}{\lm {r2}'}
 \rmk_l,
 \end{align}
 and this defines the tensor product of intertwiners.
 This tensor product satisfies following.
 \begin{description}
 \item[(i)]$(R_1\otimes R_2)^*=R_1^*\otimes R_2^*$.
 \item[(ii)] With 
 $\rho_i''\in \Orbd$, 
$\lm {ri}''\in \Crbd$, 
$\Vrl{\rho_i''}{\lm{ri}''}\in \Vbd{\rho_i''}{\lm{ri}''}$
and $R_i'\in \lmk\Tbdv{\rho_i'}{\lm {ri}'}, \Tbdv{\rho_i''}{\lm {ri}''}\rmk_l$
 for $i=1,2$, 
 we have  
 $R'_1R_1\otimes R_2'R_2=(R_1'\otimes R_2')(R_1\otimes R_2)$.
 \item[(iii)] If $\rho_i,\rho_i'\in \OrUbd$, 
$\Vrl{\rho_i}{\lm{ri}}\in \VUbd{\rho_i}{\lm{ri}}$,
$\Vrl{\rho_i'}{\lm{ri}'}\in \VUbd{\rho_i'}{\lm{ri}'}$
and $R_i\in \lmk\Tbdv{\rho_i}{\lm {r1}}, \Tbdv{\rho_i'}{\lm {r1}'}\rmk_U$,
then $R_1\otimes R_2\in\fbd$.
 \end{description}
\end{lem}
\begin{proof}
The proof is identical to that of Lemma 4.3,4.5 of \cite{MTC}.
\end{proof}

Next we introduce the direct sum to $\Orbdl$.
\begin{lem}\label{lem61}
Consider the setting in subsection \ref{setting2}.
Let $\llz\in \pc$.
Assume Assumption \ref{wakayama}.
Then for any $\rho,\sigma\in \Orbdl$, there exist
$\gamma\in \Orbdl$, isometries
$u\in (\rho,\gamma)_l$, $v\in (\sigma,\gamma)_l$
such that $uu^*+vv^*=\unit$.
In particular, $\gamma=\Ad u \circ\rho+\Ad v \circ \sigma$.
If $\rho,\sigma\in \OrUbdl$, then 
$\gamma$ and $u,v$ can be chosen so that
$\gamma\in \OrUbdl$, 
$u\in (\rho,\gamma)_U$, $v\in (\sigma,\gamma)_U$.
If $\rho,\sigma\in \Obul$, then
$\gamma$ and $u,v$ can be chosen so that
$\gamma\in \Obul$, 
$u\in (\rho,\gamma)_U$, $v\in (\sigma,\gamma)_U$.
\end{lem}
\begin{proof}
For each $\ld\in\CUbk$, with $\Lambda\subset\hu$,
$ \pbd(\caA_{\ld})''$ is properly infinite by Assumption \ref{wakayama}.
Therefore, there exist
 isometries
$u_{\ld},v_{\ld}\in  \pbd(\caA_{{\ld}})''\subset\fbd$ such that 
$u_{\ld}u_{\ld}^*+v_{\ld}v_{\ld}^*=\unit$.

For each $\lm r\in \Crbd$, choose $\lm {\lm r}\in \CUbk$ with $(\lm {\lm r},\lm r)\in \pc$.
Set $\gamma:=\Ad(u_{\lz})\rho+\Ad(v_{\lz})\sigma$, $u:=u_{\lz}$, $v:=v_{\lz}$ and
$W_{{{\lm r}}}:=u_{{{\ld}_{\lm r}}}\Vrl{\rho}{{{\lm r}}}u_{\lz}^*+v_{{{\ld}_{\lm r}}}\Vrl{\sigma}{{{\lm r}}}v_{\lz}^*$,
with $\Vrl{\rho}{{{\lm r}}}\in \Vbd\rho{{{\lm r}}}$, $\Vrl{\sigma}{{{\lm r}}}\in \Vbd\sigma{{{\lm r}}}$.
Then as in the proof of 5.7 \cite{MTC}, we have
$\gamma\in \Orbdl$, $u\in (\rho,\gamma)_U$, $v\in (\sigma,\gamma)_U$
 and $W_{\lm r}\in \Vbd\rho{\lm r}$.
 
 Suppose $\rho,\sigma\in \Obul$.
 For any $\ld\in\CUbk$, choose $\tilde\Lambda_\ld\in \CUbk$ with
 ${\tilde\Lambda_\ld}\subset\ld\cap\hu$.
 Set 
 $X_\ld:=u_{\tilde\Lambda_\ld}\Vrl\rho\ld u_{\lz}^*+v_{\tilde\Lambda_\ld}\Vrl\sigma\ld v_{\lz}^*$
 with $\Vrl\rho\ld\in \Vbu\rho\ld$, $\Vrl\sigma\ld\in \Vbu\sigma\ld$.
 Then we have $X_\ld\in\Vbu\gamma\ld$ and 
 $\gamma\in\Obul$.
\end{proof}
This defines the direct sum.
Note that it is unique up to $\simeq_l$
(resp. $\simeq_U$) because
for other $\tilde \gamma\in \Orbdl$
(resp. $\tilde \gamma\in \OrUbdl$)
$\tilde u\in (\rho,\tilde \gamma)_l$, $\tilde v\in (\sigma,\tilde \gamma)_l$,
(resp. $\tilde u\in (\rho,\tilde \gamma)_U$, $\tilde v\in (\sigma,\tilde \gamma)_U$), we have
$\tilde u u^*+\tilde v v^*\in (\gamma,\tilde\gamma)_l$
(resp. $\tilde u u^*+\tilde v v^*\in (\gamma,\tilde\gamma)_U$).

Next we prove the existence of subobject.
First we prepare several Lemmas.
\begin{lem}\label{lem65}
Consider the setting in subsection \ref{setting2}.
Then for any $\lm 1,\lm 2\subset \hu$,
$\pbd\lmk \caA_{\lm 1}\rmk ''\cap\pbd\lmk \caA_{(\lm 2)^c\cap\hu}\rmk'$
is a factor.
\end{lem}
\begin{proof}
Note that $\caM:=\pbd\lmk \caA_{\lm 1}\rmk ''\cap\pbd\lmk \caA_{(\lm 2)^c\cap\hu }\rmk'$
includes $\pbd\lmk\caA_{\lm 1\cap\lm 2}\rmk''$.
Therefore, we have
\begin{align}
\caM\cap\caM'
\subset 
\pbd\lmk \caA_{(\lm 2)^c\cap\hu}\rmk '\cap\pbd\lmk \caA_{(\lm 1)^c\cap\hu }\rmk'\cap \pbd\lmk\caA_{\lm 2\cap\lm 1}\rmk'
=\pbd(\abd)'
=\bbC\unit,
\end{align}
because $\pbd$ is irreducible.
\end{proof}

\begin{lem}\label{lem66}
Consider the setting in subsection \ref{setting2} and assume Assumption \ref{wakayama}.
Let $\lm{r1},\lm{r2}\in\Crbd$, $D\in \CUbk$ with
$D\cap\lm {r1}=\emptyset$ and $D,\lm {r1}\subset\lm{r2}$.
Then for any projection $p\in \pbd\lmk \caA_{(\lm {r1})^c\cap\hu}\rmk'\cap\bl$,
there exists an isometry
$w\in \pbd\lmk \caA_{(\lm {r2})^c\cap\hu}\rmk'\cap\bl$
such that $ww^*=p$.
\end{lem}
\begin{proof}
We apply Lemma 5.10 of \cite{MTC}.
Let $\delta>0$ be the number given in  Lemma 5.10 of \cite{MTC}.
We also use Lemma \ref{lem63}, Lemma \ref{lem64} and Lemma \ref{lem14}.
Consider the number $\delta_3(\delta_2(\delta))$ given for 
 $\delta>0$ with the functions
 $\delta_2$, $\delta_3$ given in Lemma \ref{lem63}, Lemma \ref{lem64}.

Let $p\in \pbd\lmk \caA_{(\lm {r1})^c\cap\hu}\rmk'\cap\bl$ be a projection.
Then, because $p\in\bl$, there exists $\lm l\in \Clbd$, 
and self-adjoint $x\in \pbd(\caA_{\lm l})''$ such that
$\lV
p-x
\rV\le \delta_3(\delta_2(\delta))$.
Taking larger cone if necessary, we may assume
$(\lm{r1})^c\cap \hu\subset\lm l$ and that  there exists $\Lambda\in\CUbk$
such that \begin{align}\label{shizuoka}
\Lambda\subset\lm l\cap\lm {r1}\subset \lm l \cap\lm {r2}.
\end{align}
Applying Lemma \ref{lem14} with replacing
$\lm 1$, $\lm 2$, $\Gamma$, $(\caH,\pi)$, $x$, $y$, $\varepsilon$
by
$\lm {r1}$, $\lm l$, $\hu$, $(\hbd,\pbd)$, $p$, $x$, $\delta_3(\delta_2(\delta))$
respectively, we obtain self-adjoint
$z\in \pbd\lmk\caA_{(\lm {r1})^c\cap\hu}\rmk'\cap\pbd(\caA_{\lm l})''$
such that $\lV p-z\rV\le \delta_3(\delta_2(\delta))$.

Now apply Lemma \ref{lem64} with
$\caH$, $\caA$, $p$, $x$
replaced by $\hbd$, $\pbd\lmk\caA_{(\lm {r1})^c\cap\hu}\rmk'\cap\pbd(\caA_{\lm l})''$,
$p$, $z$ respectively.
Then from Lemma \ref{lem64}, we obtain
a projection $q\in \pbd\lmk\caA_{(\lm {r1})^c\cap\hu}\rmk'\cap\pbd(\caA_{\lm l})''$
such that $\lV q-p\rV<\delta_2(\delta)$.

Next we apply Lemma \ref{lem63} with
$\caA$, $p$, $q$ replaced by
$ \pbd\lmk\caA_{(\lm {r1})^c\cap\hu}\rmk'\cap\bl$, $p$, $q$ respectively.
Then from Lemma \ref{lem63}, we obtain a unitary
$u\in \pbd\lmk\caA_{(\lm {r1})^c\cap\hu}\rmk'\cap\bl$
such that $p=uqu^*$, $\lV u-\unit\rV<\delta$.

From Lemma \ref{lem65}, $\caM:=\pbd\lmk\caA_{(\lm {r2})^c\cap\hu}\rmk'\cap\pbd(\caA_{\lm l})''$
is a factor.
For the cone $\Lambda\in\CUbk$ satisfying (\ref{shizuoka}),
$\pbd(\caA_{\Lambda})''$ is an infinite factor from Assumption \ref{wakayama}.
Because $\caM$ includes $\pbd(\caA_{\Lambda})''$, it means
$\caM$ is an infinite factor.

Now we apply Lemma 5.10 of \cite{MTC}, with
$\caH$, $\caN$, $\caM$, $p$, $u$
replaced by $\hbd$, $\pbd(\caA_{D})''$, $\caM$, $p$, $u^*$.
Then from Lemma 5.10 of \cite{MTC}, we have
$q\sim \unit$ in $\caM$.
Namely, there exists an isometry $v\in\caM$
 such that
 $vv^*=q$.
 
 Set $w:=uv\in \pbd\lmk\caA_{(\lm {r2})^c\cap\hu}\rmk'\cap\bl$.
 Then this $w$ is isometry with
 $ww^*=uqu^*=p$, proving the claim.
\end{proof}
Now we are ready to introduce subobjects.
\begin{lem}\label{lem67}
Consider the setting in subsection \ref{setting2} and assume Assumption \ref{wakayama}.
Let $\llz\in\pc$.
Then for any $\rho\in\Orbdl$ and a projection $p\in(\rho,\rho)_l$, there exists $\gamma\in\Orbdl$
and an isometry $v\in (\gamma,\rho)_l$ 
such that $vv^*=p$.
In particular, $\gamma=\Ad v^*\rho$.
\end{lem}
\begin{proof}
For any $\lm r\in\Crbd$, fix $\Vrl\rho{\lm r}\in\Vbd\rho{\lm r}$ and set
$p_{\lm r}:=\Ad(\Vrl\rho{\lm r})(p)$.
Then as in the proof of Lemma 5.8 \cite{MTC}, 
$p_{\lm r}$ is a projection in $\pbd\lmk\caA_{\lc{r}}\rmk'\cap\bl$.

For each $\lm r\in\Crbd$, fix $\Gamma_{\lm r}\in\Crbd$, $D_{\lm r}\in \CUbk$
such that $D_{\lm r}\cap \Gamma_{\lm r}=\emptyset$ and
$D_{\lm r},\Gamma_{\lm r}\subset\lm r$.
By Lemma \ref{lem66}, for the projection
$p_{{\Gamma_{\lm r}}}\in \pbd\lmk\caA_{(\Gamma_{\lm r})^c\cap\hu}\rmk'\cap\bl$, there exists
an isometry $w_{\lm r}\in \pbd\lmk\caA_{\lc{r}}\rmk'\cap\bl$
such that $w_{\lm r}w_{\lm r}^*=p_{\Gamma_{\lm r}}$.
We set $\gamma:=\Ad\lmk  w_{\lzr}^* \Vrl\rho{\Gamma_{\lzr}}\rmk\circ\rho$,
$X_{\lm r}:=w_{\lm r}^*\Vrl\rho{\Gamma_{\lm r}}\Vrl{\rho}{\Gamma_{\lzr}}^* w_{\lzr}$,
$v:=\Vrl{\rho}{\Gamma_{\lzr}}^* w_{\lzr}$.
Then, as in the proof of Lemma 5.8 \cite{MTC}, 
we get 
$\gamma\in\Orbdl$ $X_{\lm r}\in \Vbd\gamma{\lm r}$
and $v\in (\gamma,\rho)_l$ is an isometry 
such that $vv^*=p$.
\end{proof}
This defines the subobjects.
Note that it is unique up to $\simeq_l$ because
for other $\tilde \gamma\in\Orbd$,
$\tilde v\in (\gamma,\rho)_l$, $v^*\tilde v\in(\tilde\gamma,\gamma)_l$.
Hence we obtain the following.
\begin{thm}\label{naha}
Consider the setting in subsection \ref{setting2} and assume
Assumption \ref{wakayama}.
Let $\lzr\in\Crbd$.
Let $\Carbdl$ be the category with objects $\Orbdl$ and morphisms 
$(\rho,\sigma)_l$ between each $\rho,\sigma\in\Orbdl$.
Then $\Carbdl$ is the $C^*$-tensor category with respect to the
tensor structure $\otimes$, 
\begin{align}\label{osaka}
\begin{split}
&\rho\otimes\sigma:=\rho\circ_{\lm {r0}}\sigma,\quad \rho,\sigma\in \Orbdl,\\
&R\otimes S:=R\Tbd\rho{\lzr}\unit(S),\quad
R\in (\rho,\rho')_l,\quad S\in (\sigma,\sigma')_l,\quad
\rho,\rho',\sigma,\sigma'\in\Orbdl.
\end{split}
\end{align}
The representation $\pbd\in \Orbdl$ is the tensor unit,
with $\unit$ as the associativity isomorphisms as well as 
left and right multiplication by the tensor unit $\pbd$.
\end{thm}\kakunin{
\begin{proof}
Recall the definition of the $C^*$-tensor category given in Definition 2.1.1 \cite{NT}.
Clearly, for any $\rho,\sigma\in \Orbdl$, 
$(\rho,\sigma)_l$ is a Banach space
and composition as operators on $\hbd$
\begin{align}
(\sigma,\gamma)_l\times(\rho,\sigma)_l\ni (S,T)\mapsto ST\in(\rho,\gamma)_l
\end{align}
is a well defined bilinear map
satisfying $\lV ST\rV\le \lV S\rV\lV T\rV$.
The adjoint as operators on $\hbd$ gives a map
\begin{align}
(\rho,\sigma)_l\ni T\mapsto T^*\in (\sigma,\rho)_l
\end{align}
satisfying $T^{**}=T$, $\lV T^*T\rV=\lV T\rV^2$, $T^*T\ge 0$.
The formula (\ref{osaka}) gives 
$\rho\otimes\sigma\in\Orbdl$ by Lemma \ref{lem22}, and
$R\otimes S\in (\rho\otimes\sigma,\rho'\otimes \sigma')$
by Lemma \ref{lem24}.
By Lemma \ref{lem22}, 
$\unit : (\rho\otimes\sigma)\otimes\gamma\mapsto \rho\otimes(\sigma\otimes\gamma $ 
is an isomorphism
satisfying the naturality
\begin{align}
\begin{split}
(R\otimes S)\otimes T
=R\Tbd\rho{\lzr}\unit(S)\Tbd{\rho\otimes \sigma}{\lzr} \unit(T)
=R\Tbd\rho{\lzr}\unit(S)\Tbd{\rho}{\lzr} \unit\Tbd{\sigma}{\lzr}(T)\\
=R\Tbd\rho{\lzr}\unit\lmk S\Tbd{\sigma}{\lzr}(T)\rmk
=R\otimes\lmk S\otimes T\rmk
\end{split}
\end{align}
for $R\in (\rho,\rho')_l$, $S\in(\sigma,\sigma')_l$, $T\in (\gamma,\gamma')_l$,
using Lemma \ref{lem22}.
The unit $\unit$ on $\hbd$ also
gives natural isomorphisms
$\unit : \pbd\otimes\rho\to\rho$, $\unit :\rho\otimes\pbd\to \rho$
corresponding to the left and right multiplication by $\pbd$.

By Lemma \ref{lem24}, we have $(S\otimes T)^*=S^*\otimes T^*$.
By Lemma \ref{lem61}, the direct sums exist and by Lemma \ref{lem67},
subobjects exist.
Because $\pbd$ is irreducible, $\pbd$ is simple.
\end{proof}}
Analogous to Lemma \ref{lem66} and Lemma \ref{lem67}, we have the followings.
\begin{lem}\label{lem77}
Consider the setting in subsection \ref{setting2} and assume Assumption \ref{wakayama}.
Let $\lm{1},\lm{2}\in\CUbk$, $D\in \CUbk$ with
$D\cap\lm {1}=\emptyset$ and $D,\lm {1}\subset\lm{2}\cap\hu $.
Then for any projection $p\in \pbd\lmk \caA_{(\lm {1})^c\cap\hu}\rmk'\cap\fbd$,
there exists an isometry
$w\in \pbd\lmk \caA_{(\lm {2})^c\cap\hu}\rmk'\cap\fbd$
such that $ww^*=p$.
\end{lem}
\kakunin{Remove at next version
\begin{proof}
We apply Lemma 5.10 of \cite{MTC}.
Let $\delta>0$ be the number given in  Lemma 5.10 of \cite{MTC}.
We also use Lemma \ref{lem63}, Lemma \ref{lem64}.
Consider the number $\delta_3(\delta_2(\delta))$ given for 
 $\delta>0$(the number given in  Lemma 5.10 of \cite{MTC}) with the functions
 $\delta_2$, $\delta_3$ given in Lemma \ref{lem63}, Lemma \ref{lem64}.

Let $p\in \pbd\lmk \caA_{(\lm {1})^c\cap\hu}\rmk'\cap\fbd$ be a projection.
Then, because $p\in\fbd$, there exists $\lm {}\in \CUbk$, 
and self-adjoint $x\in \pbd(\caA_{\Lambda\cap\hu})''$ such that
$\lV
p-x
\rV\le \delta_3(\delta_2(\delta))$.
We may assume $D, \lm 2\subset\Lambda$.

Apply Lemma \ref{lem14} with
$\lm 1$, $\lm 2$, $\Gamma$, $(\caH,\pi)$,
$x,y,\varepsilon$
replaced by 
$\lm {1}$, $\ld\cap\hu$, $\hu$, $(\hbd,\pbd)$
$p,x,\delta_3(\delta_2(\delta))$.
Then we obtain a selfadjoint
$z\in \pbd\lmk\caA_{(\lm {1})^c\cap\hu}\rmk'\cap\pbd(\caA_{\lm {}\cap\hu })''$
such that $\lV p-z\rV\le \delta_3(\delta_2(\delta))$.

Now apply Lemma \ref{lem64} with
$\caH$, $\caA$, $p$, $x$
replaced by $\hbd$, $\pbd\lmk\caA_{(\lm {1})^c\cap\hu}\rmk'\cap\pbd(\caA_{\lm {}\cap\hu})''$,
$p$, $z$ respectively.
Then from Lemma \ref{lem64}, we obtain
a projection $q\in \pbd\lmk\caA_{(\lm {1})^c\cap\hu}\rmk'\cap\pbd(\caA_{\lm {}\cap\hu})''$
such that $\lV q-p\rV<\delta_2(\delta)$.

Note that $p, q\in \pbd\lmk\caA_{(\lm {1})^c\cap\hu}\rmk'\cap\fbd$.
Next we apply Lemma \ref{lem63} with
$\caA$, $p$, $q$ replaced by
$ \pbd\lmk\caA_{(\lm {1})^c\cap\hu}\rmk'\cap\fbd$, $p$, $q$ respectively.
Then from Lemma \ref{lem63}, we obtain a unitary
$u\in \pbd\lmk\caA_{(\lm {1})^c\cap\hu}\rmk'\cap\fbd$
such that $p=uqu^*$, $\lV u-\unit\rV<\delta$.

From Lemma \ref{lem65}, $\caM:=\pbd\lmk\caA_{(\lm {2})^c\cap\hu}\rmk'\cap\pbd(\caA_{\lm {}\cap\hu })''$
is a factor.
Because $\caM$ includes $\pbd(\caA_{\Lambda_2\cap\hu})''$, which is an infinite factor from Assumption \ref{wakayama},
$\caM$ is an infinite factor.

Now we apply Lemma 5.10 of \cite{MTC}, with
$\caH$, $\caN$, $\caM$, $p$, $u$
replaced by $\hbd$, $\pbd(\caA_{D})''$, $\caM$, $p$, $u^*$.
Then from Lemma 5.10 of \cite{MTC}, we have
$q\sim \unit$ in $\caM$.
Namely, there exists an isometry $v\in\caM$
 such that
 $vv^*=q$.
 
 Set $w:=uv\in \pbd\lmk\caA_{(\lm {2})^c\cap\hu}\rmk'\cap\fbd$.
 Then this $w$ is isometry with 
 $ww^*=uqu^*=p$, proving the claim.
\end{proof}
}
\begin{lem}\label{lem70}
Consider the setting in subsection \ref{setting2} and assume Assumption \ref{wakayama}.
Let $\llz\in\pc$.
Then for any $\rho\in\Obul$ and a projection $p\in(\rho,\rho)_U$, there exists $\gamma\in\Obul$
and an isometry $v\in (\gamma,\rho)_U$ 
such that $vv^*=p$.
In particular, $\gamma=\Ad v^*\rho$.
\end{lem}\kakunin{
\begin{proof}
For each $\ld \in\CUbk$, fix $\Vrl\rho{\ld}\in \Vbu\rho{\ld}$ and set
$p_{\ld }:=\Ad(\Vrl\rho{\ld})(p)$.
Then as in the proof of Lemma 5.8 \cite{MTC}, 
$p_{\ld}$ is a projection in $\pbd\lmk\caA_{\ld^c\cap\hu }\rmk'\cap\fbd$.

For each $\ld\in \CUbk$, fix some $\Gamma_\ld,D_\ld\in \CUbk$
such that $\Gamma_\ld,D_\ld\subset \ld\cap \hu$ and $\Gamma_\ld\cap D_\ld=\emptyset$.
Applying Lemma \ref{lem77}, with 
$\lm{1},\lm{2}$, $D$, $p$ with
replaced by $\Gamma_\ld$, $\ld$, $D_\ld$, $p_{\Gamma_\ld}$,
there exists an isometry
$w_\Lambda\in \pbd\lmk \caA_{(\ld)^c\cap\hu}\rmk'\cap\fbd$
such that $w_\ld w_\ld^*=p_{\Gamma_\ld}$.

We set $\gamma:=\Ad\lmk  w_{\lz}^* \Vrl\rho{\Gamma_{\lz}}\rmk\circ\rho$,
$X_{\ld}:=w_{\ld}^*\Vrl\rho{\Gamma_{\ld}}\Vrl{\rho}{\Gamma_{\lz}}^* w_{\lz}$,
$v:=\Vrl{\rho}{\Gamma_{\lz}}^* w_{\lz}$.
Then, as in the proof of Lemma 5.8 \cite{MTC}, 
we get 
$\gamma\in\Obul$, $X_{\ld}\in \Vbu\gamma{\ld}$
and $v\in (\gamma,\rho)_U$ is an isometry 
such that $vv^*=p$.
\end{proof}
}

\begin{thm}\label{okinawa}
Consider the setting in subsection \ref{setting2} and assume
Assumption \ref{wakayama}.
Let $\llz\in\pc$.
Let $\Cabul$ be the category with objects $\Obul$ and morphisms 
$(\rho,\sigma)_U$ between each $\rho,\sigma\in\Obul$.
Then $\Cabul$ is a $C^*$-tensor subcategory of $\Carbdl$ in the following sense:
\begin{description}
\item[(i)] $\Cabul$ is closed under $*$, composition of morphisms, and tensor.
\item[(ii)] Direct sum $\gamma$, $u\in (\rho,\gamma)_l$, 
$v\in (\sigma,\gamma)_l$ of $\rho,\sigma\in \Obul$ can be taken so that
$\gamma\in \Obul$, $u\in (\rho,\gamma)_U$, 
$v\in (\sigma,\gamma)_U$.
\item[(iii)] For any $\rho\in \Obul$ and projection $p\in (\rho,\rho)_U$,
subobject $\gamma$, $v\in (\gamma,\rho)_l$ can be taken so that $\gamma\in \Obul$
and $v\in (\gamma,\rho)_U$ .

\end{description}
\end{thm}
\begin{proof}
For any $\rho,\sigma\in \Obul$, $(\rho,\sigma)_U\subset (\rho,\sigma)_l$, and (i) is 
(iv) of Lemma \ref{lem22}.
(ii) is proven in Lemma \ref{lem61}.(iii) is Lemma \ref{lem70}.
\end{proof}

The same proof gives the following.
\begin{thm}\label{nagoya}
Consider the setting in subsection \ref{setting2} and assume
Assumption \ref{wakayama}.
Then the same statement as Theorem \ref{okinawa}  holds 
when we replace $\Cabul$ by $\CarUbdl$,
the category with objects $\OrUbdl$ and morphisms
$(\rho,\sigma)_U$ for each $\rho,\sigma\in \OrUbdl$.
\end{thm}

\section{Braiding on $\Obul$}
In this section, we show that for $\Obul$, we may assign braiding.
We use the following notation.
\begin{nota}
For cones $\{\lm {1\alpha}\}\subset\CUbk$ 
(resp. $\{\lm {l2\alpha}\}\subset\Clbd$) and $\{\lm {r2\beta}\}\subset \Crbd$,
we write $\{\lm {1\alpha}\}\leftarrow_r\{\lm {r2\beta}\}$ (resp $\{\lm {l2\alpha}\}\leftarrow_r\{\lm {r2\beta}\}$) if
there exists an $\varepsilon>0$ such that
$\arg(\lm {1\alpha})_\varepsilon\cap \arg(\lm {r2\beta})_\varepsilon=\emptyset$
(resp. $\arg(\lm {l2\alpha})_\varepsilon\cap \arg(\lm {r2\beta})_\varepsilon=\emptyset$
) for all $\alpha,\beta$.
We also write $\lm 1\leftarrow_r\lm {r2}$ if $\{\lm 1\}\leftarrow_r\{\lm {r2}\}$.
\change{See Figure \ref{arrow} for an example of $\ld_l\leftarrow \ld_r$.}

For $\lm {1}=\Lambda_{\bm a,\theta,\varphi}\in\CUbk$,
$\lm {r2}\in \Crbd$ and $t\in \bbR$, we set
$\lm 1(t):=\lm 1+t\bm e_\theta$,
$\lm {r2}(t):=\lm{r2}+t\bm e_{0}$.

\end{nota}
\begin{figure}[htbp]
  \begin{minipage}[b]{0.5\linewidth}
    \centering
    \includegraphics[width=5cm]{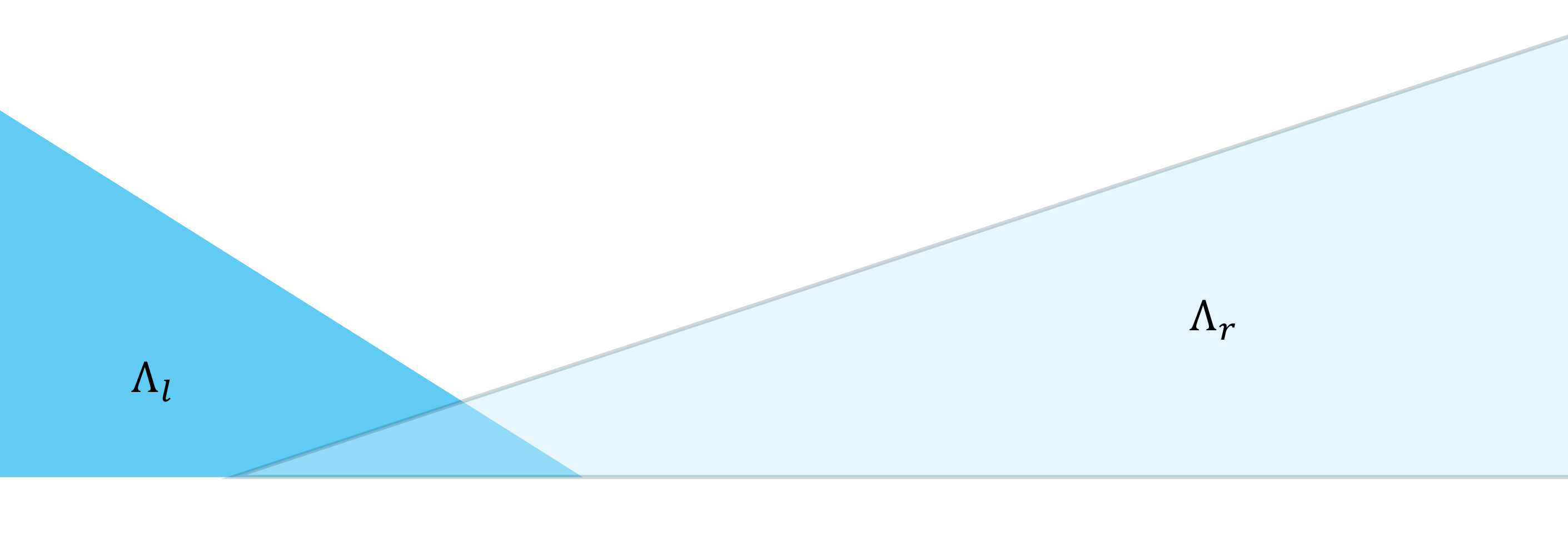}\\
\caption{}   \label{arrow}
   \end{minipage}
  \begin{minipage}[b]{0.5\linewidth}
    \centering
    \includegraphics[width=3cm]{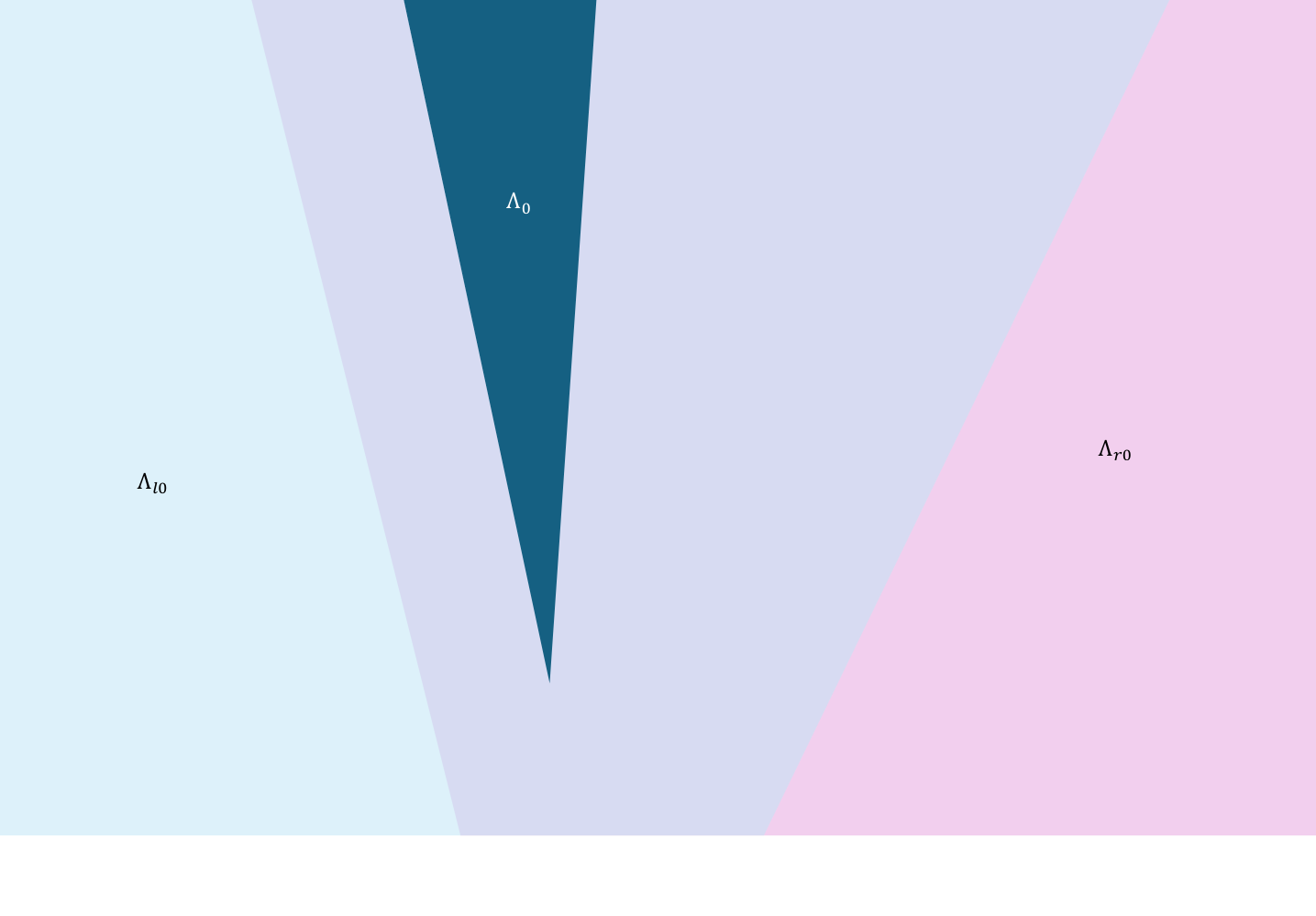}
    \caption{}\label{32}
  \end{minipage}
\end{figure}

We first investigate the asymptotic behaviors of our endomorphisms.

\begin{lem}\label{lem35}
Consider the setting in subsection \ref{setting2}. Assume Assumption \ref{assum80l}.
Let $\lm {l1}\in\Clbd$ and $\lm {r2}\in \Crbd$,
with $\lm {l1}\leftarrow_r\lm {r2}$.
Then for any $\sigma\in \Orbd$ and $\Vrl\sigma{\lm {r2}(t)}\in \Vbd\sigma{\lm {r2}(t)}$, $t\ge 0$, we have
\begin{align}
\lim_{t\to\infty}
\lV\left.
\lmk
\Tbdv\sigma{\lm {r2}(t)}-\id
\rmk\right\vert_{\pbd\lmk\caA_{(\lm {l1})^c\cap\hu} \rmk'}
\rV=0,
\end{align}
and
\begin{align}
\lim_{t\to\infty}
\sup_{\sigma\in O^{r}_{\mathop{\mathrm{bd},\lm {r2}(t)}}}
\lV\left.
\lmk
\Tbd\sigma{\lm {r2}(t)}\unit-\id
\rmk\right\vert_{\pbd\lmk\caA_{(\lm {l1})^c\cap\hu} \rmk'}
\rV=0.
\end{align}
\end{lem}
\change{
\begin{rem}
If the Haag duality holds,
then we have
\begin{align}
\pbd\lmk \caA_{(\lm {l1})^c\cap\hu}\rmk'= \pbd\lmk \caA_{\lm {l1}}\rmk''.
\end{align}
Hence if $\ld_{r2}(t)\cap\ld_{l1}=\emptyset$ (namely, if $t$ is large enough)
then we have
\begin{align}
\begin{split}
\Tbdv\sigma{\lm {r2}(t)}\vert_{\pbd\lmk\caA_{(\lm {l1})^c\cap\hu} \rmk'}
=
\id\vert_{\pbd\lmk\caA_{(\lm {l1})^c\cap\hu} \rmk'},
\end{split}
\end{align}
by (\ref{soto}).
The Lemma above is a ``tailed version" of this.
\end{rem}}
\begin{proof}
Let $\lm {l1}=\lef a \varphi\in \Clbd$.
Fix $\varepsilon>0$ such that $\arg(\lm {l1})_{2\varepsilon}\cap \arg(\lm {r2})_{2\varepsilon}
=\emptyset$.
Then there exists a $t_1\in\bbR_+$ such that
$\lmk \lef {a+s} {(\varphi+2\varepsilon)} \rmk\cap \lm {r2}(t)=\emptyset$ for all $t,s\in\bbR_+$ with
$t-s\ge t_1$.

By Assumption \ref{assum80l}, for any $s\ge R^{(l)\mopbd}_{\varphi,\varepsilon}$,
we have
\begin{align}
{\pbd\lmk\caA_{(\lm {l1})^c\cap\hu} \rmk'}\subset_{2 f^{(l)\mopbd}_{\varphi,\varepsilon,\varepsilon}(s)} \pbd\lmk 
 \caA_{
  \lef{a+s}{(\varphi+2\varepsilon)}
 }\rmk''.
\end{align}
(See Appendix \ref{basic} for the notation $\subset_\varepsilon$).
By Lemma \ref{lem20}, if $\lmk \lef {a+s} {(\varphi+2\varepsilon)} \rmk\cap \lm {r2}(t)=\emptyset$ ,
(hence if $t-s\ge t_1$)
we have 
\begin{align}\left.
\lmk
\Tbdv\sigma{\lm {r2}(t)}-\id
\rmk
\right\vert_{\pbd\lmk\caA_{ \lef{a+s}{(\varphi+2\varepsilon)} }\rmk''}=0.
\end{align}
Therefore, if $t-s\ge t_1$, we have
\begin{align}
\lV
\left.
\lmk
\Tbdv\sigma{\lm {r2}(t)}-\id
\rmk
\right\vert_{{\pbd\lmk\caA_{(\lm {l1})^c\cap\hu} \rmk'}}\rV
\le 4f^{(l)\mopbd}_{\varphi,\varepsilon,\varepsilon}(s).
\end{align}
This completes the proof of the first equality.
The proof of the second one is the same.
\end{proof}

The following Lemma gives an information on the support of $\Tbd\rho{\lzr}{\Vrl\rho\lz}$
for $\rho\in\Obu$
\begin{lem}\label{lem81}
Consider the setting in subsection \ref{setting2}.
Let $\lz\in\CUbk$, $\lzr\in\Crbd$, $\lm {l0}\in \Clbd$ satisfying
$\hu\cap\lz\subset \lzr\cap \lm {l0}$.
If $\rho\in\Obu$ and $\Vrl\rho\lz\in \Vbu\rho\lz\subset \VUbd\rho{\lzr}$ then 
\begin{align}\label{fukushima}\left.
\Tbd\rho{\lzr}{\Vrl\rho\lz}\right\vert_{\pbd\lmk \caA_{(\lm {l0})^c\cap\hu}\rmk''\cap\fbd}=\id,
\end{align}
and
\begin{align}\label{tochigi}
\left.
\Tbd\rho{\lzr}{\Vrl\rho\lz}\pbd\right\vert_{ \caA_{\lmk (\lm {l0})^c\cup (\lm {r0})^c\rmk\cap\hu}}
=\pbd.
\end{align}
\end{lem}
\change{
\begin{rem}
See Figure \ref{32}. The pink part is $\ld_{r0}$, and the light-blue part is $\ld_{l0}$.
The cone $\ld_0$ is the dark-blus part in the intersection of them.
Note that $\Tbd\rho{\lzr}{\Vrl\rho\lz}$ is not defined on the whole
$\pbd\lmk \caA_{(\lm {l0})^c\cap\hu}\rmk''$, but it is defined on
${\pbd\lmk \caA_{(\lm {l0})^c\cap\hu}\rmk''\cap\fbd}$.
\end{rem}}
\change{
\begin{rem}
We cannot replace $\fbd$ by $\bl$.
In fact this difference between $\bl$ and $\fbd$ is the main point of
Assumption \ref{raichi}.
\end{rem}
}
\begin{proof}
By interchanging the role of left and right in Lemma \ref{lem20},
we can define an endomorphism ${T_\rho^{\mathrm{(r)} \lm {l0}\Vrl{{\rho}}{\lz}}}$ of 
$\caB_r:=\overline{\cup_{\Lambda_r\in \Crbd}\pbd\lmk\caA_{\Lambda_r}\rmk''}^n
=\overline{\cup_{\lm l\in\Clbd}\pbd\lmk\caA_{\lm l^c\cap\hu}\rmk''}$
which is $\sigma$-weak continuous on any
$\pbd(\caA_{{(\lm {l})^c\cap \hu}})''$, $\lm{l}\in\Clbd$, and
\begin{align}\label{yamagata}
{T_\rho^{\mathrm{(r)} \lm {l0}\Vrl{{\rho}}{\lz}}}\pbd=\Ad\lmk \Vrl\rho\lz\rmk\rho=\Tbd\rho{\lzr}{\Vrl\rho\lz}\pbd
\end{align}
and 
\begin{align}\label{miyagi}
\left.{T_\rho^{\mathrm{(r)} \lm {l0}\Vrl{{\rho}}{\lz}}}\right\vert_{\pbd\lmk \caA_{(\lm {l0})^c\cap\hu}\rmk''}
=\id_{\pbd\lmk \caA_{(\lm {l0})^c\cap\hu}\rmk''}.
\end{align}By the $\sigma$-weak continuity and
(\ref{yamagata}), we have
${T_\rho^{\mathrm{(r)} \lm {l0}\Vrl{{\rho}}{\lz}}}\vert_{\fbd}=\Tbd\rho{\lzr}{\Vrl\rho\lz}\vert_{\fbd}$.
Combining this with (\ref{miyagi}),
we obtain (\ref{fukushima}).

By Lemma \ref{lem20}, we have $
\Tbd\rho{\lzr}{\Vrl\rho\lz}\vert_{\pbd\lmk \caA_{\lmk (\lm {r0})^c\rmk\cap\hu}\rmk}=\id
$. We also have $
\Tbd\rho{\lzr}{\Vrl\rho\lz}\vert_{\pbd\lmk \caA_{\lmk (\lm {l0})^c\rmk\cap\hu}\rmk}=\id
$  from (\ref{fukushima}).
Hence we obtain (\ref{tochigi}).
\end{proof}

\begin{lem}\label{lem82}
Consider the setting in subsection \ref{setting2}, and assume Assumption \ref{assum80}.
Let $(\lm 1,\lm {r1})\in \pc$, $\lm {r2}\in \Crbd$
with $\lm 1\leftarrow_r \lm {r2}$.
Let $\rho\in \Obu$.
Then the following hold.
\begin{description}
\item[(i)]
For any $\Vrl \rho{\lm 1}\in \Vbu\rho{\lm 1}$,
\begin{align}
\lim_{t\to\infty }
\lV\left.
\lmk \Tbd\rho{\lm {r1}}{\Vrl\rho{\lm 1}}-\id\rmk
\right\vert_{\pbd\lmk \caA_{\lmk \lm {r2}(t)\rmk^c\cap\hu}\rmk'\cap\fbd}
\rV=0.
\end{align}
\item[(ii)] 
For any 
$\Vrl\rho{\lm 1(t)}\in \Vbu\rho{\lm 1(t)}$, $t\ge 0$
we have
\begin{align}
\lim_{t\to\infty}
\lV
\left.
\lmk \Tbd\rho{\lm {r1}}{\Vrl\rho{\lm 1(t)}}-\id\rmk
\right\vert_{{\pbd\lmk \caA_{\lmk \lm {r2}\rmk^c\cap\hu}\rmk'\cap\fbd}}
\rV=0.
\end{align}
\end{description}
\end{lem}
\change{
\begin{rem}
If the Haag duality holds, then we 
have
\begin{align}
\pbd\lmk \caA_{\lmk \lm {r2}\rmk^c\cap\hu}\rmk'=\pbd\lmk \caA_{\lm {r2}\cap\hu}\rmk''.
\end{align}
In this case, if ${\lm 1(t)}\cap \lm {r2}=\emptyset$ (i.e., if $t$ is large enough),
then we have
\begin{align}
\Tbd\rho{\lm {r1}}{\Vrl\rho{\lm 1(t)}}\vert_{{\pbd\lmk \caA_{\lmk \lm {r2}\rmk^c\cap\hu}\rmk'\cap\fbd}}=
\id\vert_{{\pbd\lmk \caA_{\lmk \lm {r2}\rmk^c\cap\hu}\rmk'\cap\fbd}}
\end{align}
from Lemma \ref{lem81}.
The above is its ``tailed version''.
\end{rem}
}
\begin{proof}
The cone $\lm {r2}$ is of the form $\lm {r2}=\lr a\varphi$.
Because $\lm 1\leftarrow_r \lm {r2}$, there exists $\varepsilon>0$ such that 
$\arg \lmk \lm 1\rmk_{2\varepsilon}\cap \arg\lmk \lm{r2}\rmk_{2\varepsilon}=\emptyset$.
\\
(i)
Fix $b\in \bbR$ large enough so that
$\lm l:=\lmk\lr  {b} {(\varphi+\varepsilon)}\rmk^c\cap \hu\in\Clbd$ satisfies
$\Lambda_1\subset \lm l\cap\lm {r1}$.
We apply Lemma \ref{lem81} replacing
$\lz$, $\lm {r0}$, $\lm {l0}$, $\rho$, $\Vrl\rho{\lz}$ by
$\lm 1$, $\lm {r1}$, $\lm l$, $\rho$,
$\Vrl\rho{\lm 1}\in\Vbu\rho{\lm 1}$ respectively.
Then we obtain 
\begin{align}\label{taiwan}
\left.\Tbd\rho{\lm {r1}}{\Vrl\rho{\lm1}}\right\vert_{\pbd\lmk \caA_{(\lm {l})^c\cap\hu}\rmk''\cap\fbd}=\id.
\end{align}

By Assumption \ref{assum80}, for any $t\ge 2R^{(r)\mathop{\mathrm {bd}}}_{\varphi,\frac \varepsilon 2}$,
there exists $W_t\in \caU(\fbd)$ such that
\begin{align}\label{ibaraki}
&\pbd\lmk \caA_{\lmk \lm {r2}(t)\rmk^c\cap\hu}\rmk'
\subset
\Ad\lmk W_t\rmk\lmk
\pbd\lmk
\caA_{\lr {a+\frac t2}{\varphi+\varepsilon}}
\rmk''
\rmk,\\
&\lV
W_t-\unit
\rV\le {f_{\varphi, \frac \varepsilon 2,\frac \varepsilon 2 }^{(r)\mopbd}\lmk \frac t2\rmk}.
\end{align}
If $t\ge \max\{2(b-a),2R^{(r)\mathop{\mathrm {bd}}}_{\varphi,\frac \varepsilon 2}\}$, then 
$\lr {(a+\frac t2)}{\varphi+\varepsilon}\subset \lr {b}{\varphi+\varepsilon}=(\lm l)^c\cap\hu$,
and combining with (\ref{ibaraki}), we have
\begin{align}
\begin{split}
\pbd\lmk \caA_{\lmk \lm {r2}(t)\rmk^c\cap\hu}\rmk'
\subset
\Ad\lmk W_t\rmk
\lmk
\pbd\lmk
\caA_{(\lm l)^c\cap\hu}
\rmk''
\rmk.
\end{split}
\end{align}
Now for any $t\ge \max\{2(b-a),2R^{(r)\mathop{\mathrm {bd}}}_{\varphi,\frac \varepsilon 2}\}$
and  $x\in {\pbd\lmk \caA_{\lmk \lm {r2}(t)\rmk^c\cap\hu}\rmk'\cap\fbd}$,
we have $\Ad(W_t^*)(x)\in \pbd\lmk
\caA_{(\lm l)^c\cap\hu}
\rmk''\cap\fbd
$ and 
\begin{align}\label{gunma}
\begin{split}
&\lV\lmk \Tbd\rho{\lm {r1}}{\Vrl\rho{\lm 1}}-\id\rmk(x)\rV
\le
\lV\lmk \Tbd\rho{\lm {r1}}{\Vrl\rho{\lm 1}}-\id\rmk\lmk \Ad(W_t^*)(x)\rmk\rV
+2\lV x-\Ad(W_t^*)(x)\rV<4f_{\varphi, \frac \varepsilon 2,\frac \varepsilon 2 }^{(r)\mopbd}\lmk \frac t2\rmk\lV x\rV,
\end{split}
\end{align}
using (\ref{taiwan})
This proves (i).
\\
(ii)For each $s\ge R^{(r)\mathop{\mathrm {bd}}}_{\varphi,\frac \varepsilon 2}$,
by Assumption \ref{assum80}, there exists $\tilde W_s\in\caU(\fbd)$ such that 
\begin{align}
&\pbd\lmk \caA_{\lmk \lm {r2}\rmk^c\cap\hu}\rmk'
\subset
\Ad\lmk\tilde W_s\rmk
\lmk
\pbd\lmk
\caA_{\lr {a-s}{\varphi+\varepsilon}}
\rmk''
\rmk,\\
&\lV
\tilde W_s-\unit
\rV\le {f^{(r)\mopbd}_{\varphi, \frac \varepsilon 2,\frac \varepsilon 2 }(s)}
\end{align}
For each $s\ge 0$, there exists a $t_0(s)\ge 0$ such that
$\lm 1(t)\cap {\lr {a-s}{\varphi+\varepsilon}}=\emptyset$ for all $t\ge t_0(s)$.

Applying Lemma \ref{lem81} replacing
$\lz$, $\lm {r0}$, $\lm {l0}$, $\rho$, $\Vrl\rho{\lz}$ by
$\lm 1(t)$, $\lm {r1}$, $\lmk \lr {a-s}{\varphi+\varepsilon}\rmk^c\cap\hu$, $\rho$,
$\Vrl\rho{\lm 1(t)}\in\Vbu\rho{\lm 1(t)}$ respectively, for $t\ge t_0(s)$,
we obtain 
\begin{align}
\begin{split}
\left.\Tbd{\rho}{\lm {r1}}{\Vrl \rho{\lm 1(t)}}\right\vert_{\pbd\lmk
\caA_{{\lr {a-s}{\varphi+\varepsilon}}}\rmk''\cap\fbd}
=\id.
\end{split}
\end{align}
Now for any $x\in {{\pbd\lmk \caA_{\lmk \lm {r2}\rmk^c\cap\hu}\rmk'\cap\fbd}}
$,
 $s\ge R^{(r)\mathop{\mathrm {bd}}}_{\varphi,\frac \varepsilon 2}$ and 
 $t\ge t_0(s)$
and  we have $\Ad(\tilde W_s^*)(x)\in \pbd\lmk
\caA_{\lr {a-s}{\varphi+\varepsilon}}
\rmk''\cap\fbd
$ and 
\begin{align}\label{nagasaki}
\begin{split}
&\lV \lmk \Tbd\rho{\lm {r1}}{\Vrl\rho{\lm 1(t)}}-\id\rmk(x)\rV\\
&\le
\lV \lmk \Tbd\rho{\lm {r1}}{\Vrl\rho{\lm 1(t)}}-\id\rmk
\lmk \Ad(\tilde W_s^*)(x)\rmk\rV
+2\lV \Ad(\tilde W_s^*)(x)-x\rV
\le {4f^{(r)\mopbd}_{\varphi, \frac \varepsilon 2,\frac \varepsilon 2 }(s)}\lV x\rV.
\end{split}
\end{align}
Hence for any $s\ge R^{(r)\mathop{\mathrm {bd}}}_{\varphi,\frac \varepsilon 2}$ and 
 $t\ge t_0(s)$,
 we have
$\lV \lmk \Tbd\rho{\lm {r1}}{\Vrl\rho{\lm 1(t)}}-\id\rmk\vert_{{{\pbd\lmk \caA_{\lmk \lm {r2}\rmk^c\cap\hu}\rmk'\cap\fbd}}}\rV
\le {4f^{(r)\mopbd}_{\varphi, \frac \varepsilon 2,\frac \varepsilon 2 }(s)}$,
and this proves (ii).

\end{proof}
Using the Lemmas we obtained, we obtain the following.
\begin{lem}\label{lem38}
Consider the setting in subsection \ref{setting2}.
Assume Assumption \ref{assum80} and Assumption \ref{assum80l}.
Let $\rho,\rho'\in \Obu$, $\sigma,\sigma'\in \OrUbd$.
Let $\lm 1,\lm 1'\in \CUbk$, $\lm {r2},\lm {r2}'\in\Crbd$
with $\{\lm 1,\lm 1'\}\leftarrow_r \{\lm {r2},\lm {r2}'\}$.
Let $\lm {r1}, \lm {r1}'\in\Crbd$ with
$(\lm {1}, \lm {r1})\in \pc$, $(\lm {1}', \lm {r1}')\in \pc$.
Let $\Vrl\rho{\lm 1(t_1)}\in \Vbu\rho{\lm 1(t_1)}\subset \Vbd\rho{\lm {r1}}$,
$\Vrl{\rho'}{{\lm 1}'(t_1')}\in \Vbu{\rho'}{{\lm 1}'(t_1')}\subset \Vbd\rho{\lm {r1}'}$,
$\Vrl\sigma{\lm {r2}(t_2)}\in \VUbd\sigma{\lm {r2}(t_2)}$,
$\Vrl{\sigma'}{\lm {r2}'(t_2')}\in \VUbd{\sigma'}{\lm {r2}'(t_2')}$,
for $t_1,t_1',t_2,t_2'\ge 0$.
Then for any
\begin{align}
\begin{split}
S_1^{(t_1,t_1')}\in \lmk
\Tbd\rho{\lm {r1}}{\Vrl\rho{\lm 1(t_1)}},
\Tbd{\rho'}{\lm {r1}'}{\Vrl{\rho'}{\lm 1'(t_1')}}
\rmk_U,\\
S_2^{(t_2,t_2')}\in \lmk
\Tbd\sigma{\lm {r2}}{\Vrl\sigma{\lm {r2}(t_2)}},
\Tbd{\sigma'}{\lm {r2}'}{\Vrl{\sigma'}{\lm {r2}'(t_2')}}
\rmk_U
\end{split}
\end{align}
with $\lV S_1^{(t_1,t_1')}\rV, \lV S_2^{(t_2,t_2')}\rV\le 1$,
we have
\begin{align}
\lV
S_1^{(t_1,t_1')}\otimes S_2^{(t_2,t_2')}-
S_2^{(t_2,t_2')}\otimes S_1^{(t_1,t_1')}
\rV\to 0,\quad t_1,t_2,t_1',t_2'\to\infty.
\end{align}
\end{lem}
\begin{rem}
Assume the Haag duality to hold. Suppose that there are cones 
$\ld\in\CUbk$
$\lm l\in \Clbd$ and $\lm r\in \Crbd$ such that $\lm 1,\lm 1'\subset \ld\subset \lm l$,
$\lm {r2},\lm{r2}'\subset \lm r$ and $\lm l\cap\lm r=\emptyset$.
If $\rho\in \Obun{\lm1}$, $\rho'\in \Obun{\lm1'}$, 
$\sigma\in \OrUbdn{\lm{r2}}$, $\sigma'\in \OrUbdn{\lm{r2}'}$  
then for any 
\begin{align}
\begin{split}
S_1\in (\rho,\rho')_U\subset \pbd(\caA_{\lm l^c})'=\pbd(\caA_{\lm l})'',\\
S_2\in (\sigma,\sigma')_U\subset \pbd(\caA_{\lm r^c})'\cap \fbd
=\pbd(\caA_{\lm r})''\cap \fbd,
\end{split}
\end{align}
we have
\begin{align}
S_1\otimes S_2-S_2\otimes S_1=S_1\Tbd\rho{\lm {r1}} \unit(S_2)
-S_2\Tbd\sigma{\lm {r2}} \unit(S_1)=S_1S_2-S_2S_1=0,
\end{align}
\end{rem}
by Lemma \ref{lem81}, Lemma \ref{lem20}.
This Lemma is a ``tailed version '' of this.
\begin{proof}
By the assumption $\{\lm 1,\lm 1'\}\leftarrow_r \{\lm {r2},\lm {r2}'\}$,
there exist an $\varepsilon>0$ and ${\hlm {1}}\in\CUbk$, ${\hlm {r2}}\in \Crbd$
such that
\begin{align}
\begin{split}
 &\arg\lmk\lm 1\rmk_\varepsilon \cup \arg\lmk\lm 1'\rmk_\varepsilon\subset \arg{\hlm 1},\quad
 \arg\lmk\lm {r2}\rmk_\varepsilon\cup \arg\lmk\lm {r2}'\rmk_\varepsilon\subset \arg{\hlm {r2}},\\
 &(\hlm 1)_{\varepsilon}\cap(\hlm {r2})_\varepsilon=\emptyset,\quad
\arg (\hlm 1)_{\varepsilon}\cap\arg (\hlm {r2})_\varepsilon=\emptyset, \\
&{\hlm 1}\subset \hu,\quad |\arg\hlm 1|+4\varepsilon<2\pi.
\end{split}
\end{align}
By Lemma A.2 of \cite{MTC}, for any $t\ge 0$ fixed, there exists $t_0(t)\ge 0$
such that
\begin{align}\label{hiei}
\lm 1(t_1), \lm 1'(t_1')\subset \hlm 1(t),\quad
t_1,t_1'\ge t_0(t).
\end{align}
Similarly, for $t\ge 0$ fixed, there exists $\tilde t_0(t)\ge 0$ such that
\begin{align}\label{atago}
\lm {r2}(t_2), \lm {r2}'({t_2'})\subset \hlm{r2}(t),\quad 
t_2,t_2'\ge\tilde t_0(t).
\end{align}
By Lemma \ref{lem81} (\ref{tochigi}), for any $t\ge 0$, $t_1,t_1'\ge t_0(t)$,
$\Tbd\rho{\lm {r1}}{\Vrl\rho{\lm 1(t_1)}}\pbd$
and $\Tbd{\rho'}{\lm {r1}'}{\Vrl{\rho'}{\lm 1'(t_1')}}\pbd$
 are $\pbd$ on $\caA_{(\hlm 1(t))^c\cap \hu}$.
 Therefore,
as in Lemma 4.2 of \cite{MTC}, we have
\begin{align}\label{chiba}
\begin{split}
S_1^{(t_1,t_1')}\in \lmk
\Tbd\rho{\lm {r1}}{\Vrl\rho{\lm 1(t_1)}},
\Tbd{\rho'}{\lm {r1}'}{\Vrl{\rho'}{\lm 1'(t_1')}}
\rmk_U
\subset 
\pbd\lmk\caA_{(\hlm 1(t))^c\cap \hu}\rmk'\cap\fbd
\end{split}
\end{align}
for any $t\ge 0$, $t_1,t_1'\ge t_0(t)$.
Similarly, we have
\begin{align}\label{kanagawa}
S_2^{(t_2,t_2')}\in \lmk
\Tbd\sigma{\lm {r2}}{\Vrl\sigma{\lm {r2}(t_2)}},
\Tbd{\sigma'}{\lm {r2}'}{\Vrl{\sigma'}{\lm {r2}'(t_2')}}
\rmk_U\subset \pbd\lmk\caA_{(\hlm {r2}(t))^c\cap \hu}\rmk'\cap\fbd,
\end{align}
for any $t\ge 0$, $ t_2,t_2'\ge\tilde t_0(t)$.
As in the proof of Lemma \ref{lem82} (\ref{ibaraki}),
by Assumption \ref{assum80}
there exists $T_0\ge 0$ such that for any $t\ge T_0$
we have
\begin{align}
\begin{split}
&\pbd\lmk \caA_{\lmk \hlm {r2}(t)\rmk^c\cap\hu}\rmk'
\subset
\Ad\lmk W_t\rmk\lmk
\pbd\lmk
\caA_{\lmk \lmk \hlm{r2}\rmk_\varepsilon\cap\hu\rmk+\frac t 2\bm e_1}
\rmk''
\rmk,
\end{split}
\end{align}
with some $W_t\in\caU(\fbd)$ such that $\lim_{t\to\infty}\lV
W_t-\unit
\rV=0$.
Note that 
 $\lmk \lmk \hlm{r2}\rmk_\varepsilon\cap\hu\rmk+\frac t 2\bm e_1\subset (\hlm 1(t))^c\cap \hu$
 for any $t\ge 0$.
 Therefore, from (\ref{chiba}), we have
 \begin{align}
 \begin{split}
& \lV
 \left[S_1^{(t_1,t_1')}, S_2^{(t_2,t_2')}\right]
 \rV
\le 4\lV W_t-\unit\rV 
 \end{split}
 \end{align}
 for any $t\ge 0$, $t_1,t_1'\ge t_0(t)$ and $ t_2,t_2'\ge\tilde t_0(t)$.
 By Lemma \ref{lem35} and (\ref{chiba}), we have
 \begin{align}
 \begin{split}
 \lim_{t_1,t_1',t_2\to\infty} \lV \lmk\Tbdv\sigma{\lm {r2}(t_2)}-\id\rmk\lmk S_1^{(t_1,t_1')}\rmk\rV
 =0.
 \end{split}
 \end{align}
By Lemma \ref{lem82} and (\ref{kanagawa}), we have
\begin{align}
\begin{split}\lim_{t_2,t_2',t_1\to\infty}
\lV
\lmk \Tbd{\rho}{\lm {r1}}{\Vrl \rho{\lm 1(t_1)}}-\id\rmk
\lmk S_2^{(t_2,t_2')}\rmk
\rV=0.
\end{split}
\end{align}
Combining them we obtain
\begin{align}
\begin{split}
&\lim_{t_1,t_1',t_2,t_2'\to\infty}
\lV
S_1^{(t_1,t_1')}\otimes S_2^{(t_2,t_2')}-
S_2^{(t_2,t_2')}\otimes S_1^{(t_1,t_1')}
\rV\\
&=\lim_{t_1,t_1',t_2,t_2'\to\infty}
\lV S_1^{(t_1,t_1')} 
 \Tbd{\rho}{\lm {r1}}{\Vrl \rho{\lm 1(t_1)}}\lmk S_2^{(t_2,t_2')}\rmk
 -S_2^{(t_2,t_2')}\Tbdv\sigma{\lm {r2}(t_2)}\lmk S_1^{(t_1,t_1')} \rmk
\rV\\
&=\lim_{t_1,t_1',t_2,t_2'\to\infty} \lV
 \left[S_1^{(t_1,t_1')}, S_2^{(t_2,t_2')}\right]
 \rV=0.
\end{split}
\end{align}
\end{proof}
%

\change{The condition of cones in the previous Lemma can be relaxed.}
\begin{lem}\label{lem39}
The statement of Lemma \ref{lem38} holds even when
we replace the condition $\{\lm 1,\lm 1'\}\leftarrow_r \{\lm {r2},\lm {r2}'\}$,
by the condition $\lm 1\leftarrow_r \lm {r2}$, $\lm 1'\leftarrow_r \lm {r2}'$.
\end{lem}
\begin{proof}
Note that analogous to Lemma 4.9 of \cite{MTC}, there are sequences of cones $\{\lm 1^{(i)}\}_{i=0}^3\subset\CUbk$,
$\{\lm {r2}^{(i)}\}_{i=0}^3\subset \Crbd$ such that
$\{\lm 1^{(i)},\lm 1^{(i+1)}\}\leftarrow_r \{ \lm {r2}^{(i)}, \lm {r2}^{(i+1)}\}$, $i=0,\ldots,2$
and $\lm 1^{(0)}=\lm 1$, $\lm {r2}^{(0)}=\lm {2r}$,
$\lm 1^{(3)}=\Lambda_{\bm 0, \frac\pi 2,\frac\pi{64}}$,
$\lm {r2}^{(3)}=\lr 0 {\frac\pi{64}}$.
Then the rest is the same as Lemma 4.10 of \cite{MTC}.
\kakunin{
Fix some $\tlm {r1}\in \Crbd$ including all $\lm 1^{(i)}$.
With $\Vrl{\rho}{\lm 1^{(i)}(t_1^{(i)})}\in \Vbu{\rho}{\lm 1^{(i)}(t_1^{(i)})}$,
$\Vrl{\sigma}{\lm {r2}^{(i)}(t_2^{(i)})}\in \VUbd \sigma {\lm {r2}^{(i)}}$,
we set
\begin{align}
W_{\rho}^{(i)}\lmk t_1^{(i)}, t_1^{(i+1)}\rmk:=
\Vrl{\rho}{\lm 1^{(i+1)}(t_1^{(i+1)})}\lmk \Vrl{\rho}{\lm 1^{(i)}(t_1^{(i)})}\rmk^*
\in \lmk
\Tbd\rho {\tlm {r1}}{\Vrl{\rho}{\lm 1^{(i)}(t_1^{(i)})}},
\Tbd\rho {\tlm {r1}}{\Vrl{\rho}{\lm 1^{(i+1)}(t_1^{(i+1)})}}
\rmk_U,
\end{align}
and
\begin{align}
\tilde W_{\sigma}^{(i)}\lmk t_2^{(i)}, t_2^{(i+1)}\rmk:=
\Vrl{\sigma}{\lm {r2}^{(i+1)}(t_2^{(i+1)})}\lmk \Vrl{\sigma}{\lm {r2}^{(i)}(t_2^{(i)})}\rmk^*
\in \lmk
\Tbdv\sigma {{\lm {r2}^{(i)}(t_2^{(i)})}},
\Tbdv\sigma {{\lm {r2}^{(i+1)}(t_2^{(i+1)})}}
\rmk_U.
\end{align}
Then by Lemma \ref{lem38}, we have
\begin{align}\label{thailand}
\begin{split}
\lim_{t_1^{(i)}, t_2^{(i)}, t_1^{(i+1)}, t_2^{(i+1)}\to\infty}
\lV
W_{\rho}^{(i)}\lmk t_1^{(i)}, t_1^{(i+1)}\rmk\otimes
\tilde W_{\sigma}^{(i)}\lmk t_2^{(i)}, t_2^{(i+1)}\rmk
-\tilde W_{\sigma}^{(i)}\lmk t_2^{(i)}, t_2^{(i+1)}\rmk\otimes W_{\rho}^{(i)}\lmk t_1^{(i)}, t_1^{(i+1)}\rmk
\rV=0.
\end{split}
\end{align}
We carry out exactly the same argument to $\lm 1'$, $\lm {r2}'$, and obtain the $'$-version of the objects above.

We set 
\begin{align}\label{tokyo}
\begin{split}
&R_{\bm t}:=
{W'}_{\rho'}^{(2)}\lmk {t_1'}^{(2)}, {t_1'}^{(3)}\rmk\cdots{W'}_{\rho'}^{(0)}\lmk {t_1'}^{(0)}, {t_1'}^{(1)}\rmk
S_1^{(t_1^{(0)},{t_1'}^{(0)})}
\lmk W_{\rho}^{(0)}\lmk t_1^{(0)}, t_1^{(1)}\rmk \rmk^*\cdots \lmk W_{\rho}^{(2)}\lmk t_1^{(2)}, t_1^{(3)}\rmk\rmk^*\\
&\in
\lmk \Tbd\rho{\tlm {r1}}{\Vrl \rho {\Lambda_{\bm 0, \frac\pi 2,\frac\pi{64}}(t_1^{(3)})}}, \Tbd{\rho'}{\tlm {r1}'}{\Vrl {\rho'} {\Lambda_{\bm 0, \frac\pi 2,\frac\pi{64}}({t_1'}^{(3)})}}\rmk,\\
&\tilde R_{\bm t}:=
{\tilde {W'}}_{\sigma'}^{(2)}\lmk {t_2'}^{(2)}, {t_2'}^{(3)}\rmk\cdots{\tilde{W'}}_{\sigma'}^{(0)}\lmk {t_2'}^{(0)}, {t_2'}^{(1)}\rmk
S_2^{(t_2^{(0)},{t_2'}^{(0)})}
\lmk W_{\sigma}^{(0)}\lmk t_2^{(0)}, t_2^{(1)}\rmk \rmk^*\cdots \lmk W_{\sigma}^{(2)}\lmk t_2^{(2)}, t_2^{(3)}\rmk\rmk^*\\
&\in
\lmk \Tbdv\sigma{{\lr 0 {\frac\pi{64}}(t_2^{(3)})}}, \Tbdv{\sigma'}{{\lr 0 {\frac\pi{64}}({t_2'}^{(3)})}}\rmk_U
.
\end{split}
\end{align}
Then by Lemma \ref{lem38}, we have
\begin{align}
\lim_{\bm t\to\infty}
\lV
R_{\bm t}\otimes \tilde R_{\bm t}-\tilde R_{\bm t}\otimes R_{\bm t}
\rV=0.
\end{align}
As in Lemma 4.10 of \cite{MTC}, combining this and (\ref{thailand}),
we obtain $\lV
S_1^{(t_1,t_1')}\otimes S_2^{(t_2,t_2')}-
S_2^{(t_2,t_2')}\otimes S_1^{(t_1,t_1')}
\rV\to 0$.
}
\end{proof}
Noe we are ready to introduce a half-braiding.
\begin{prop}\label{lem40}
Consider the setting in subsection \ref{setting2} and assume Assumption \ref{assum80}, \ref{assum80l}.
Let $(\lz,\lzr)\in \pc$, and  $\rho\in \Obul$, $\sigma\in \OrUbdl$.
Then the norm limit
\begin{align}\label{halfdef}
\begin{split}
\iota^{(\lz,\lzr)}\lmk\rho: \sigma\rmk
:=\lim_{t_1,t_2\to\infty} 
\lmk \Vrl\sigma{\lm {r2}(t_2)}\otimes \Vrl\rho{\lm 1(t_1)}\rmk^*
 \lmk \Vrl\rho{\lm 1(t_1)}\otimes  \Vrl\sigma{\lm {r2}(t_2)}\rmk\in\fbd
\end{split}
\end{align}
exists for any 
$\lm 1\in \CUbk$, $\lm {r2}\in \Crbd$
with $\lm 1\leftarrow_r \lm{r2}$, 
and $\Vrl\rho{\lm 1(t_1)}\in \Vbu\rho{\lm 1(t_1)}$, $t_1\ge 0$,
$\Vrl \sigma{\lm {r2}(t_2)}\in \VUbd\sigma{\lm {r2}(t_2)}$, $t_2\ge 0$,
and its value is independent of the choice of
$\lm 1\in \CUbk$, $\lm {r2}\in \Crbd$
with $\lm 1\leftarrow_r \lm{r2}$, 
and $\Vrl\rho{\lm 1(t_1)}\in \Vbu\rho{\lm 1(t_1)}$, $t_1\ge 0$,
$\Vrl \sigma{\lm {r2}(t_2)}\in \VUbd\sigma{\lm {r2}(t_2)}$, $t_2\ge 0$.
Furthermore, we have
\begin{align}\label{niigata}
\iota^{(\lz,\lzr)}\lmk\rho: \sigma\rmk=
\lim_{t_2\to\infty}\lmk \Vrl \sigma{\lm {r2}(t_2)}\rmk^*\Tbd\rho{\lzr}\unit
\lmk\Vrl \sigma{\lm {r2}(t_2)}\rmk.
\end{align}
for any 
$\lm {r2}\in \Crbd$ with $\lz\leftarrow_r \lm {r2}$
and
$\Vrl \sigma{\lm {r2}(t_2)}\in \VUbd\sigma{\lm {r2}(t_2)}$, $t_2\ge 0$.
\end{prop}
\begin{rem}\label{mike}
From (\ref{niigata}), what we are doing here can be understood as follows.
Recall that $\Vrl \sigma{\lm {r2}(t_2)}$ sends $\sigma$ particle from
$\lzr$ to ``right infinity " $\lm {r2}(t_2)$, as $t_2\to \infty$.
(\change{For example, recall in Toric code \cite{wa} (Example \ref{toric}),
a transporter $V_{\pi_{\gamma'}^Z\ld_{r2}(t_2)}$ is given
as a weak-limit of operators acting along the paths $\gamma_n$ depicted in Figure
\ref{36}.})
When we do so, our $\sigma$ particle may cross the string creating the $\rho$-particle.
What we are measuring with $\iota^{(\lz,\lzr)}\lmk\rho: \sigma\rmk$
is the braiding which is caused at the cross.
\change{Note that the limit was not needed in \cite{wa} to define the half-braiding.
This is because in Toric code, the Haag duality holds, and
$\lmk \Vrl \sigma{\lm {r2}(t_2)}\rmk^*\Tbd\rho{\lzr}\unit
\lmk\Vrl \sigma{\lm {r2}(t_2)}\rmk$ does not depend on $t_2$ if it is large enough.
}

\end{rem}
\begin{proof}
\kakunin{
For any $\lm 1,\lm 1'\in \CUbk$, $\lm {r2},\lm {r2}'\in \Crbd$
with $\lm 1\leftarrow_r \lm{r2}$,  $\lm 1'\leftarrow_r \lm{r2}'$.
Let $(\lm 1,\lm{r1}), (\lm 1',\lm{r1}')\in\pc$.
Let $\Vrl\rho{\lm 1(t_1)}\in \Vbu\rho{\lm 1(t_1)}$,
$\Vrl \sigma{\lm {r2}(t_2)}\in \VUbd\sigma{\lm {r2}(t_2)}$.
Then we have
\begin{align}\begin{split}
&\lmk \Vrl\sigma{\lm {r2}'(t_2')}\otimes \Vrl\rho{\lm 1'(t_1')}\rmk^*
 \lmk \Vrl\rho{\lm 1'(t_1')}\otimes  \Vrl\sigma{\lm {r2}'(t_2')}\rmk\\
& =\lmk \Vrl\sigma{\lm {r2}(t_2)}\otimes \Vrl\rho{\lm 1(t_1)}\rmk^*\\
&\lmk
 \Vrl\sigma{\lm {r2}'(t_2')}\lmk \Vrl\sigma{\lm {r2}(t_2)}\rmk^*
\otimes 
 \Vrl\rho{\lm 1'(t_1')}\Vrl\rho{\lm 1(t_1)}^*
\rmk^*\\
& \lmk
  \Vrl\rho{\lm 1'(t_1')}\lmk \Vrl\rho{\lm 1(t_1)}\rmk^*\otimes
  \Vrl\sigma{\lm {r2}'(t_2')}\lmk \Vrl\sigma{\lm {r2}(t_2)}\rmk^*\rmk
 \\
& \lmk \Vrl\rho{\lm 1(t_1)}\otimes  \Vrl\sigma{\lm {r2}(t_2)}\rmk
 \end{split}
\end{align}
By Lemma \ref{lem39}, we have
\begin{align}
\begin{split}
&\lV
\lmk
 \Vrl\sigma{\lm {r2}'(t_2')}\lmk \Vrl\sigma{\lm {r2}(t_2)}\rmk^*
\otimes 
 \Vrl\rho{\lm 1'(t_1')}\Vrl\rho{\lm 1(t_1)}^*
\rmk^*\lmk
  \Vrl\rho{\lm 1'(t_1')}\lmk \Vrl\rho{\lm 1(t_1)}\rmk^*\otimes
  \Vrl\sigma{\lm {r2}'(t_2')}\lmk \Vrl\sigma{\lm {r2}(t_2)}\rmk^*\rmk
-\unit
\rV\\
&\quad\to 0,\quad\quad\quad t_1,t_1',t_2,t_2'\to\infty.
\end{split}
\end{align}
Hence the norm limit (\ref{halfdef}) exists and it is independent of the choice of
$\lm 1\in \CUbk$ etc.}
The first part is identical to the proof of Lemma 4.11 \cite{MTC}.
Because $ \Vrl\sigma{\lm {r2}(t_2)}$ etc belong to $\fbd$, by Lemma \ref{lem20},
$\iota^{(\lz,\lzr)}\lmk\rho: \sigma\rmk$ belongs to $\fbd$.

For any 
$\lm {r2}\in \Crbd$  with $\lz\leftarrow_r \lm {r2}$
and
$\Vrl \sigma{\lm {r2}(t_2)}\in \VUbd\sigma{\lm {r2}(t_2)}$, $t_2\ge 0$,
set $\lm 1:=\lz \in \CUbk$, and take any $\Vrl\rho{\lm 1(t_1)}\in \Vbu\rho{\lm 1(t_1)}$, $t_1\ge 0$.
Then we have
\begin{align}
\begin{split}
&\lmk \Vrl\sigma{\lm {r2}(t_2)}\otimes \Vrl\rho{\lm 1(t_1)}\rmk^*
 \lmk \Vrl\rho{\lm 1(t_1)}\otimes  \Vrl\sigma{\lm {r2}(t_2)}\rmk\\
& =\lmk \Vrl\sigma{\lm {r2}(t_2)}\rmk^*\Tbdv \sigma {\lm {r2}(t_2)}\lmk \Vrl\rho{\lm 1(t_1)}^*\rmk
\Vrl\rho{\lm 1(t_1)}\Tbd\rho{\lz}\unit\lmk \Vrl\sigma{\lm {r2}(t_2)}\rmk.
\end{split}
\end{align}
Because $\lm 1:=\lz $, we have
$ \Vrl\rho{\lm 1(t_1)}\in
  \pbd\lmk \caA_{\lz^c\cap\hu }\rmk'\cap \fbd$.
 Because of $\lz\leftarrow_r \lm {r2}$, from Lemma \ref{lem35}, we have
 $\Tbdv \sigma {\lm {r2}(t_2)}\lmk \Vrl\rho{\lm 1(t_1)}^*\rmk
\Vrl\rho{\lm 1(t_1)}\to \unit$ , $t_2\to 0$ in norm, uniformly in $t_1\ge 0$.
Hence we obtain (\ref{niigata}).
\end{proof}
\begin{figure}[htbp]
  \begin{minipage}[b]{0.5\linewidth}
    \centering
    \includegraphics[width=5cm]{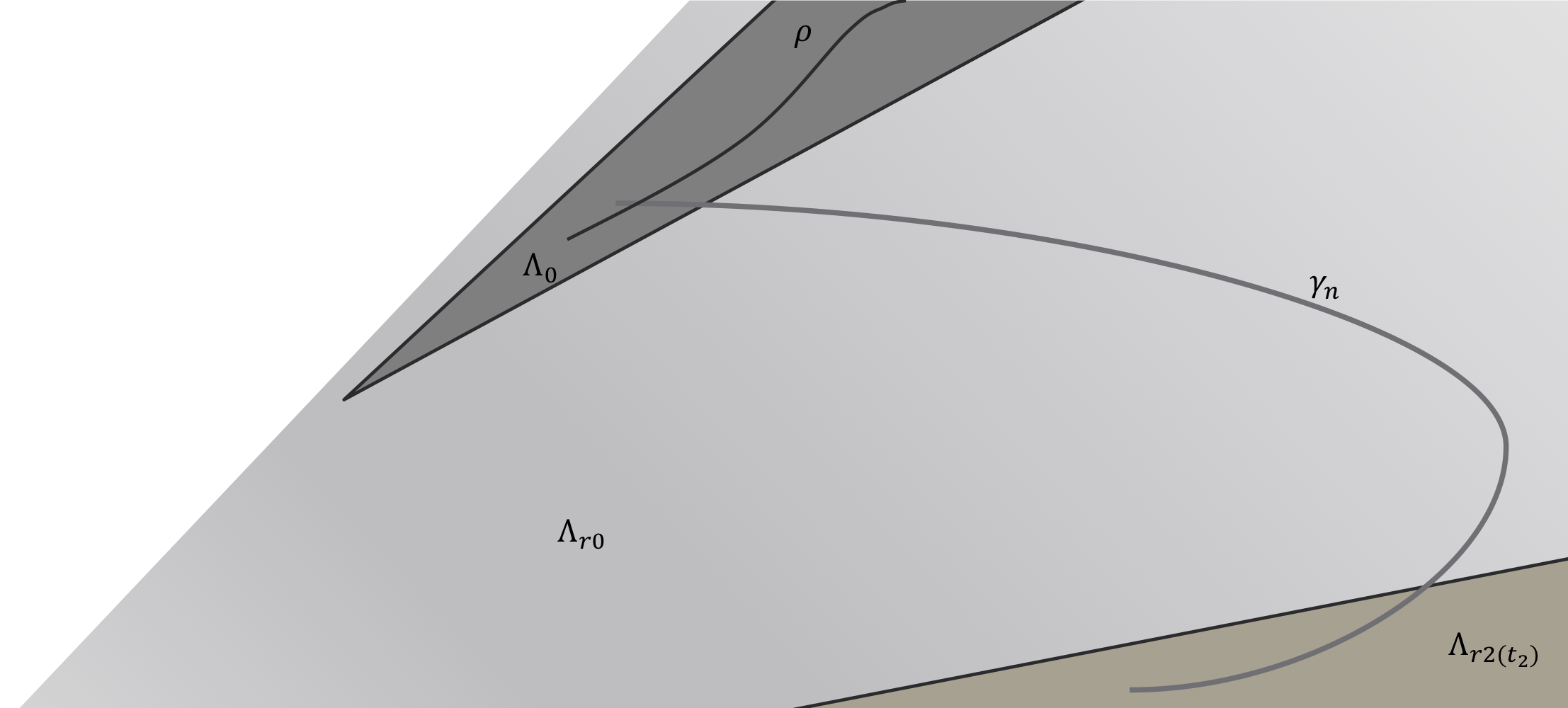}\\
\caption{}   \label{36}
   \end{minipage}
  \begin{minipage}[b]{0.5\linewidth}
    \centering
    \includegraphics[width=3cm]{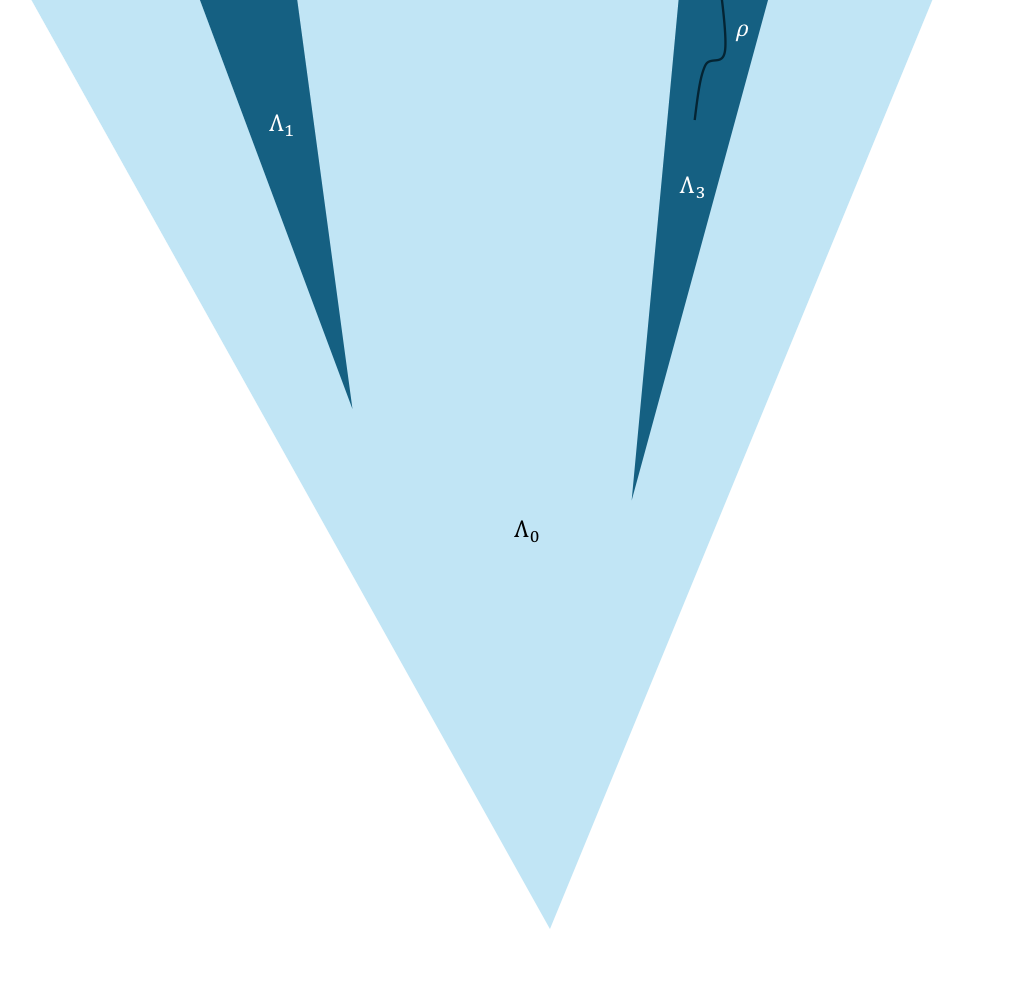}
    \caption{}\label{317}
  \end{minipage}
\end{figure}

\change{
The proof of the following Lemma, stating that $\iota^{(\lz,\lzr)}\lmk\rho: \sigma\rmk$  is a $\fbd$-valued
morphism from $\rho\otimes \sigma$ to $\sigma\otimes\rho$}
is identical to that of Lemma 4.13 \cite{MTC}.
\begin{lem}\label{lem41}
Consider the setting in subsection \ref{setting2} and assume Assumption \ref{assum80}, \ref{assum80l}.
Let $\rho\in \Obul$, $\sigma\in \OrUbdl$ and $(\lz,\lzr)\in \pc$.
Then we have$$\iota^{(\lz,\lzr)}\lmk\rho: \sigma\rmk\in \lmk \Tbd \rho{\lzr}{\unit}\Tbd\sigma{\lzr}\unit, \Tbd\sigma{\lzr}\unit \Tbd \rho{\lzr}{\unit}\rmk_U.$$
\end{lem}\kakunin{
\begin{proof}
Let $(\lm 1,\lm{r1})\in\pc$.
For any $A\in \aloch$, we have
\begin{align}
\begin{split}
&\iota^{(\lz,\lzr)}\lmk\rho: \sigma\rmk \Tbd \rho{\lzr}{\unit}\Tbd\sigma{\lzr}\unit\pbd(A)\\
&=\lim_{t_1,t_2\to\infty} 
\lmk \Vrl\sigma{\lm {r2}(t_2)}\otimes \Vrl\rho{\lm 1(t_1)}\rmk^*
 \Tbd \rho{\lm {r1}}{\lm 1(t_1)}\Tbdv\sigma{\lm {r2}(t_2)}\pbd(A)
 \lmk \Vrl\rho{\lm 1(t_1)}\otimes  \Vrl\sigma{\lm {r2}(t_2)}\rmk\\
 &=\lim_{t_1,t_2\to\infty} 
\lmk \Vrl\sigma{\lm {r2}(t_2)}\otimes \Vrl\rho{\lm 1(t_1)}\rmk^*
\pbd(A)
 \lmk \Vrl\rho{\lm 1(t_1)}\otimes  \Vrl\sigma{\lm {r2}(t_2)}\rmk\\
&= \lim_{t_1,t_2\to\infty} 
\lmk \Vrl\sigma{\lm {r2}(t_2)}\otimes \Vrl\rho{\lm 1(t_1)}\rmk^*
 \Tbdv\sigma{\lm {r2}(t_2)}\Tbd \rho{\lm {r1}}{\lm 1(t_1)}\pbd(A)
 \lmk \Vrl\rho{\lm 1(t_1)}\otimes  \Vrl\sigma{\lm {r2}(t_2)}\rmk\\
 &=\Tbd\sigma{\lzr}\unit \Tbd \rho{\lzr}{\unit}\pbd(A)\iota^{(\lz,\lzr)}\lmk\rho: \sigma\rmk,
\end{split}
\end{align}
because  $A\in\caA_{(\lm 1(t_1)^c\cap (\lm {r2}(t_2))^c}$ eventually.
By the $\sigma$-weak continuity of $\Tbd \rho{\lzr}{\unit}\Tbd\sigma{\lzr}\unit=\Tbd {\rho\circ_{\lzr}\sigma}{\lzr}{\unit}$, etc
(Lemma \ref{lem20}, \ref{lem22}), this proves the claim.
\end{proof}}
The proof of the following Lemma \change{about the naturality of $\iota^{(\lz,\lzr)}$ in the second variable}
 is identical to that of Lemma 4.14 \cite{MTC},
using Lemma \ref{lem39}.

\begin{lem}\label{lem42}
Consider the setting in subsection \ref{setting2} and assume Assumption \ref{assum80}, \ref{assum80l}.
Let $(\lz,\lzr)\in \pc$ and $\rho\in \Obul$, $\sigma,\sigma'\in \OrUbdl$.
Then for any 
$S\in \lmk\Tbd\sigma{\lzr}\unit, \Tbd{\sigma'}{\lzr}\unit\rmk_U $,
we have
\begin{align}
\iota^{(\lz,\lzr)}\lmk\rho: \sigma'\rmk\lmk \unit_\rho\otimes S\rmk
=\lmk S\otimes \unit_\rho \rmk\iota^{(\lz,\lzr)}\lmk\rho: \sigma\rmk.
\end{align}
\end{lem}
\begin{rem}
Note that $S$ is required to be in $\fbd$, not $\bl$.
\end{rem}\kakunin{
\begin{proof}
For any 
$\lm 1\in \CUbk$, $\lm {r2}\in \Crbd$
with $\lm 1\leftarrow_r \lm{r2}$, 
\begin{align}
\begin{split}
&\lV
\iota^{(\lz,\lzr)}\lmk\rho: \sigma'\rmk\lmk \unit_\rho\otimes S\rmk
-\lmk S\otimes \unit_\rho\rmk\iota^{(\lz,\lzr)}\lmk\rho: \sigma\rmk
\rV\\
&=\lim_{t_1,t_2\to\infty}
\lV
\begin{gathered}
\lmk \Vrl{\sigma'}{\lm {r2}(t_2)}\otimes \Vrl{\rho}{\lm 1(t_1)}\rmk^*
 \lmk \Vrl{\rho}{\lm 1(t_1)}\otimes  \Vrl{\sigma'}{\lm {r2}(t_2)}\rmk
\lmk \unit_\rho\otimes S\rmk
-\lmk S\otimes \unit_\rho\rmk
\lmk \Vrl\sigma{\lm {r2}(t_2)}\otimes \Vrl\rho{\lm 1(t_1)}\rmk^*
 \lmk \Vrl\rho{\lm 1(t_1)}\otimes  \Vrl\sigma{\lm {r2}(t_2)}\rmk
\end{gathered}
\rV\\
&=\lim_{t_1,t_2\to\infty}
\lV 
\begin{gathered}
\unit_{\Tbd\rho{\lm {r1}}{\lm 1(t_1)}}
\otimes
\Vrl{\sigma'}{\lm {r2}(t_2)} S \lmk \Vrl\sigma{\lm {r2}(t_2)}\rmk^*
-
 \Vrl{\sigma'}{\lm {r2}(t_2)}S\lmk \Vrl\sigma{\lm {r2}(t_2)}\rmk^*
 \otimes
 \unit_{\Tbd\rho{\lm {r1}}{\lm 1(t_1)}}
\end{gathered}
\rV\\
&=0,
\end{split}
\end{align}
by Lemma \ref{lem39}.
\end{proof}
}
\change{
For the first variable, the naturality holds for a wider class, i.e., $\bl$.}
\begin{lem}\label{inu}
Consider the setting in subsection \ref{setting2} and assume Assumption \ref{assum80}, \ref{assum80l}.
Let $(\lz,\lzr)\in \pc$ and $\rho,\rho'\in \Obul$, $\sigma\in \OrUbdl$.
Then for any $R\in \lmk \Tbd\rho{\lzr}\unit, \Tbd{\rho'}{\lzr}\unit\rmk_l$,
we have
\begin{align}
\iota^{(\lz,\lzr)}\lmk\rho': \sigma\rmk\lmk R\otimes \unit_\sigma\rmk
=\lmk \unit_\sigma \otimes R\rmk\iota^{(\lz,\lzr)}\lmk\rho: \sigma\rmk.
\end{align}
\end{lem}
\begin{proof}
For any 
$\lm 1\in \CUbk$, $\lm {r2}\in \Crbd$
with $\lm 1\leftarrow_r \lm{r2}$, 
\begin{align}
\begin{split}
&\lV
\iota^{(\lz,\lzr)}\lmk\rho': \sigma\rmk\lmk R\otimes \unit_\sigma\rmk
-\lmk \unit_\sigma\otimes R\rmk\iota^{(\lz,\lzr)}\lmk\rho: \sigma\rmk
\rV\\
&=\lim_{t_1,t_2\to\infty}
\lV
\begin{gathered}
\lmk \Vrl{\sigma}{\lm {r2}(t_2)}\otimes \Vrl{\rho'}{\lm 1(t_1)}\rmk^*
 \lmk \Vrl{\rho'}{\lm 1(t_1)}\otimes  \Vrl{\sigma}{\lm {r2}(t_2)}\rmk
\lmk R\otimes \unit_\sigma\rmk\\
-\lmk \unit_\sigma\otimes R\rmk
\lmk \Vrl\sigma{\lm {r2}(t_2)}\otimes \Vrl\rho{\lm 1(t_1)}\rmk^*
 \lmk \Vrl\rho{\lm 1(t_1)}\otimes  \Vrl\sigma{\lm {r2}(t_2)}\rmk
\end{gathered}
\rV\\
&=\lim_{t_1,t_2\to\infty}
\lV 
\begin{gathered}
\Vrl{\rho'}{\lm 1(t_1)} R\lmk \Vrl\rho{\lm 1(t_1)}\rmk^*
\otimes
\Vrl{\sigma}{\lm {r2}(t_2)} \unit_\sigma \lmk \Vrl\sigma{\lm {r2}(t_2)}\rmk^*\\
-
 \Vrl{\sigma}{\lm {r2}(t_2)}\unit_\sigma\lmk \Vrl\sigma{\lm {r2}(t_2)}\rmk^*
 \otimes
 \Vrl{\rho'}{\lm 1(t_1)} R\Vrl\rho{\lm 1(t_1)}^*
\end{gathered}
\rV\\
&=
\lim_{t_1,t_2\to\infty}
\lV 
\begin{gathered}
\Vrl{\rho'}{\lm 1(t_1)} R\lmk \Vrl\rho{\lm 1(t_1)}\rmk^*
-
  \Tbdv{\sigma}{\lm {r2}(t_2)}\lmk
  \Vrl{\rho'}{\lm 1(t_1)} R\Vrl\rho{\lm 1(t_1)}^*
  \rmk
\end{gathered}
\rV\\
&=0.
\end{split}
\end{align}
At the last equality, we used the fact
\begin{align}
\begin{split}
 \Vrl{\rho'}{\lm 1(t_1)} R\Vrl\rho{\lm 1(t_1)}^*\in 
 \pbd\lmk
 \caA_{{{\lm 1}(t_1)}^c\cap \hu}
 \rmk'\cap\bl
\end{split}
\end{align}
and Lemma \ref{lem35}.

\end{proof}

Furthermore, we have the following.
\begin{lem}\label{lem43}
Consider the setting in subsection \ref{setting2} and assume Assumption \ref{assum80}, \ref{assum80l}.
Let $(\lz,\lzr)\in \pc$ and $\rho\in \Obul$, $\sigma,\gamma\in \OrUbdl$.
Then we have
\begin{align}
\begin{split}
\hb{\rho}{\sigma\otimes\gamma}
=\lmk \unit_{\sigma}\otimes \hb\rho\gamma\rmk
\lmk \hb\rho\sigma\otimes\unit_\gamma\rmk.
\end{split}
\end{align}
\end{lem}
\begin{proof}
The proof is the same as \cite{MTC} Lemma 4.17.
\kakunin{
Fix some $(\lm 1,\lm {r1})\in \pc$, $\lm {r2}\in \Crbd$
so that $\lm 1\leftarrow_r\lzr$,
$\lz\leftarrow_r\lm{r2}$, $\lm{r2}\subset \lzr$.
Then we have $\lm 1\leftarrow_r \lm{r2}$, and from
Proposition \ref{lem40},
\begin{align}\label{ishikawa}
\begin{split}
&\lmk \unit_{\sigma}\otimes \hb\rho\gamma\rmk
\lmk \hb\rho\sigma\otimes\unit_\gamma\rmk\\
&=\lim_{t_1,t_2\to\infty}
\Tbd{\sigma}{\lzr}\unit \lmk
\Tbd{\gamma}{\lzr}\unit \lmk\Vrl\rho{\lm 1(t_1)}^*\rmk
\Vrl\gamma{\lm {r2}(t_2)}^*\Vrl\rho{\lm 1(t_1)}
\Tbd\rho{\lzr}\unit\lmk \Vrl\gamma{\lm {r2}(t_2)}\rmk
\rmk\\
&\quad \quad \Tbd{\sigma}{\lzr}\unit \lmk\Vrl\rho{\lm 1(t_1)}^*\rmk
\Vrl\sigma{\lm {r2}(t_2)}^*\Vrl\rho{\lm 1(t_1)}
\Tbd\rho{\lzr}\unit\lmk \Vrl\sigma{\lm {r2}(t_2)}\rmk\\
&=\lim_{t_1,t_2\to\infty}
\Tbd{\sigma}{\lzr}\unit \lmk
\Tbd{\gamma}{\lzr}\unit \lmk\Vrl\rho{\lm 1(t_1)}^*\rmk
\Vrl\gamma{\lm {r2}(t_2)}^*
\Vrl\rho{\lm 1(t_1)}
\Tbd\rho{\lzr}\unit\lmk \Vrl\gamma{\lm {r2}(t_2)}\rmk
\Vrl\rho{\lm 1(t_1)}^*\rmk
\Vrl\sigma{\lm {r2}(t_2)}^*\Vrl\rho{\lm 1(t_1)}
\Tbd\rho{\lzr}\unit\lmk \Vrl\sigma{\lm {r2}(t_2)}\rmk\\
&=\lim_{t_1,t_2\to\infty}
\Tbd{\sigma}{\lzr}\unit \lmk
\Tbd{\gamma}{\lzr}\unit \lmk\Vrl\rho{\lm 1(t_1)}^*\rmk
\Vrl\gamma{\lm {r2}(t_2)}^*
\Tbd\rho{\lm{r1}}{\Vrl\rho{\lm 1(t_1)}}\lmk \Vrl\gamma{\lm {r2}(t_2)}\rmk
\rmk
\Vrl\sigma{\lm {r2}(t_2)}^*\Vrl\rho{\lm 1(t_1)}
\Tbd\rho{\lzr}\unit\lmk \Vrl\sigma{\lm {r2}(t_2)}\rmk\\
&=\lim_{t_1,t_2\to\infty}
\Tbd{\sigma}{\lzr}\unit \lmk
\Tbd{\gamma}{\lzr}\unit \lmk\Vrl\rho{\lm 1(t_1)}^*\rmk
\Vrl\gamma{\lm {r2}(t_2)}^*
\rmk
\Vrl\sigma{\lm {r2}(t_2)}^*\\
&\quad\quad\quad\Tbdv\sigma{\lm {r2}(t_2)}\Tbd\rho{\lm {r1}}{\Vrl\rho{\lm 1(t_1)}}\lmk \Vrl\gamma{\lm {r2}(t_2)}\rmk
\Vrl\rho{\lm 1(t_1)}
\Tbd\rho{\lzr}\unit\lmk \Vrl\sigma{\lm {r2}(t_2)}\rmk.
\end{split}
\end{align}
Here, because $\lm{r2}\subset \lzr$, we have 
\begin{align}
\begin{split}
\Vrl\gamma{\lm {r2}(t_2)},\;
\Tbdv\sigma{\lm {r2}(t_2)}\lmk \Vrl\gamma{\lm {r2}(t_2)}\rmk\in 
\pbd\lmk
\caA_{\lrhuz}
\rmk'\cap\fbd.
\end{split}
\end{align}
Therefore, by Lemma \ref{lem82}, we have
\begin{align}
\begin{split}
\lim_{t_1,t_2\to\infty}
\lV
\lmk
\Tbdv\sigma{\lm {r2}(t_2)}\Tbd\rho{\lm {r1}}{\Vrl\rho{\lm 1(t_1)}}
-\Tbd\rho{\lm {r1}}{\Vrl\rho{\lm 1(t_1)}}\Tbdv\sigma{\lm {r2}(t_2)}
\rmk
\lmk \Vrl\gamma{\lm {r2}(t_2)}\rmk
\rV=0.
\end{split}
\end{align}
Substituting this to (\ref{ishikawa}), we obtain
\begin{align}
\begin{split}
&\lmk \unit_{\sigma}\otimes \hb\rho\gamma\rmk
\lmk \hb\rho\sigma\otimes\unit_\gamma\rmk\\
&=\lim_{t_1,t_2\to\infty}
\Tbd{\sigma}{\lzr}\unit \lmk
\Tbd{\gamma}{\lzr}\unit \lmk\Vrl\rho{\lm 1(t_1)}^*\rmk
\Vrl\gamma{\lm {r2}(t_2)}^*
\rmk
\Vrl\sigma{\lm {r2}(t_2)}^*\\
&\quad\quad\quad\Tbd\rho{\lm{r1}}{\Vrl\rho{\lm 1(t_1)}}\Tbdv\sigma{\lm {r2}(t_2)}\lmk \Vrl\gamma{\lm {r2}(t_2)}\rmk
\Vrl\rho{\lm 1(t_1)}
\Tbd\rho{\lzr}\unit\lmk \Vrl\sigma{\lm {r2}(t_2)}\rmk\\
&=\lim_{t_1,t_2\to\infty}
\Tbd{\sigma}{\lzr}\unit \lmk
\Tbd{\gamma}{\lzr}\unit \lmk\Vrl\rho{\lm 1(t_1)}^*\rmk
\rmk
\Tbd{\sigma}{\lzr}\unit \lmk
\Vrl\gamma{\lm {r2}(t_2)}^*
\rmk
\Vrl\sigma{\lm {r2}(t_2)}^*\\
&\quad\quad\quad 
\Vrl\rho{\lm 1(t_1)}
\Tbd\rho{\lzr}\unit
\lmk
\Tbdv\sigma{\lm {r2}(t_2)}\lmk \Vrl\gamma{\lm {r2}(t_2)}\rmk
\lmk \Vrl\sigma{\lm {r2}(t_2)}\rmk
\rmk\\
&=\lim_{t_1,t_2\to\infty}
\Tbd{\sigma}{\lzr}\unit \lmk
\Tbd{\gamma}{\lzr}\unit \lmk\Vrl\rho{\lm 1(t_1)}^*\rmk
\rmk
\lmk
\lmk \Vrl\sigma{\lm {r2}(t_2)}\rmk
\Tbd\sigma{\lzr}\unit \lmk \Vrl\gamma{\lm {r2}(t_2)}\rmk
\rmk^*\\
&\quad\quad\quad 
\Vrl\rho{\lm 1(t_1)}
\Tbd\rho{\lzr}\unit
\lmk\lmk \Vrl\sigma{\lm {r2}(t_2)}\rmk
\Tbd\sigma{\lzr}\unit \lmk \Vrl\gamma{\lm {r2}(t_2)}\rmk
\rmk\\
&=\lim_{t_1,t_2\to\infty}
\lmk\lmk  \lmk \Vrl\sigma{\lm {r2}(t_2)}\rmk
\Tbd\sigma{\lzr}\unit \lmk \Vrl\gamma{\lm {r2}(t_2)}\rmk\rmk
\otimes 
\lmk\Vrl\rho{\lm 1(t_1)}
\rmk\rmk^*
\\
&\quad\quad\quad 
\Vrl\rho{\lm 1(t_1)}
\Tbd\rho{\lzr}\unit
\lmk\lmk \Vrl\sigma{\lm {r2}(t_2)}\rmk
\otimes  \lmk \Vrl\gamma{\lm {r2}(t_2)}\rmk
\rmk\\
&=\hb{\rho}{\sigma\circ_{\lzr}\gamma}.
\end{split}
\end{align}}
\end{proof}

\begin{lem}\label{panama}
Let $\llz\in\pc$.
Consider the setting in subsection \ref{setting2}, and 
assume Assumption \ref{assum80},Assumption \ref{assum80l}.
Then for any $\rho,\sigma\in \Obul$ and $\gamma\in \OrUbdl$, we have
\begin{align}
\begin{split}
&\hb{\rho\otimes\sigma}{\gamma}\\
&=\lmk
\hb { \rho}{\gamma}\otimes 
\unit_{\Tbd{\sigma}\lzr\unit}
\rmk
\lmk
\unit_{\Tbd{\rho}\lzr\unit}
\otimes 
\hb { \sigma}{\gamma}
\rmk
\end{split}
\end{align}
\end{lem}
\begin{proof}
For any $\lm {r2}\in \Crbd$ with $\lz\leftarrow_r \lm{r2}$ and  $\Vrl\gamma{\lm {r2}(t)}\in \VUbd\gamma{{\lm {r2}(t)}}$,
$t\ge 0$
from Proposition \ref{lem40},
we have
\begin{align}
\begin{split}
&\lmk
\hb {\rho}{\gamma}\otimes 
\unit_{\Tbd{\sigma}\lzr\unit}
\rmk
\lmk
\unit_{\Tbd{\rho}\lzr\unit}
\otimes 
\hb {\sigma}{\gamma}
\rmk\\
&=
\lim_{t\to\infty} \Vrl\gamma{\lm {r2}(t)}^*
\Tbd{\rho}{\lzr}\unit \lmk \Vrl\gamma{\lm {r2}(t)}\rmk
{\Tbd{\rho}\lzr\unit}
\lmk
 \Vrl\gamma{\lm {r2}(t)}^*
\Tbd{\sigma}{\lzr}\unit \lmk \Vrl\gamma{\lm {r2}(t)}\rmk
\rmk\\
&=\hb{\rho\otimes\sigma}{\gamma}.
\end{split}
\end{align}
\end{proof}
Hence if we restrict our $\iota^{\llz}$ to $\Obul$, then it becomes a braiding, and
we obtain the following.
\begin{thm}\label{matsuyama}
Consider the setting in subsection \ref{setting2} and assume
Assumption \ref{assum80}, Assumption \ref{assum80l},
Assumption \ref{wakayama}.
Let $\llz\in\pc$.
Then the $C^*$-tensor category $\Cabul$ is braided with braiding
\begin{align}
\begin{split}
\epsilon_{\Cabul}(\rho:\sigma)
:=\hb\rho\sigma,\quad\rho,\sigma\in \Obul.
\end{split}
\end{align}
\end{thm}

Our $\hb-\rho$ satisfies the following asymptotic property.
\begin{lem}\label{neko}
Consider the setting in subsection \ref{setting2} and assume Assumption \ref{assum80}, \ref{assum80l}.
Let $(\lz,\lzr)\in \pc$ and 
$\lm 1,\lm 3\in \CUbk$ with $\lm 1\leftarrow_r \lm 3$, $\lm 1,\lm3\subset\lz$.
Then for any $\rho\in \Obul$ and $\Vrl\rho{\lm3}\in \Vbu\rho{\lm 3}$
we have
\begin{align}\label{mikan1}
\begin{split}
\lim_{t\to\infty}\sup_{\sigma\in \Obun{\lm 1(t)}}
\lV
\hb\sigma \rho -\Vrl\rho{\lm3}^* \Tbd\sigma{\lzr}\unit\lmk \Vrl\rho{\lm3}\rmk
\rV=0.
\end{split}
\end{align}
\end{lem}
\begin{rem}
Note in particular 
\begin{align}
\begin{split}
\lim_{t\to\infty}\sup_{\sigma\in \Obun{\lm 1(t)}}
\lV
{\Vrl\rho{\lm3}'}^* \Tbd\sigma{\lzr}\unit\lmk {\Vrl\rho{\lm3}'}\rmk-\Vrl\rho{\lm3}^* \Tbd\sigma{\lzr}\unit\lmk \Vrl\rho{\lm3}\rmk
\rV
\end{split}
\end{align}
for any other $\Vrl\rho{\lm3}'\in \Vbu\rho{\lm 3}$.
\end{rem}
\change{
\begin{rem}\label{murasaki}See Figure \ref{317}.
If $\rho\in \Obun{\lm3}$, we may take $\Vrl\rho{\lm3}=\unit$
hence
$\Vrl\rho{\lm3}^* \Tbd\sigma{\lzr}\unit\lmk \Vrl\rho{\lm3}\rmk=\unit$,
and (\ref{mikan1}) says
we have \[
\lim_{t\to\infty}\sup_{\sigma\in \Obun{\lm 1(t)}}
\lV
\hb\sigma \rho -\unit 
\rV=0
\]
In our definition of $\hb\sigma \rho$, (roughly) we are moving the $\rho$ particle
to right infinity, and measuring the braiding at the cross with $\sigma$'s string operator (recall Remark \ref{mike}).
If $\rho$ is on the right of $\sigma$ from the beginning (for example , if $\rho\in \Obun{\lm3}$ and
$\lm 3$ is on the right of $\lm 1$), then there will be no
crossing. Therefore $\hb\sigma \rho\approx \unit$ (tail occurs because of the approximate Haag duality). That is what this Lemma is telling us.
\end{rem} 
}
\begin{proof}
For $\rho\in \Obul$ and $\Vrl\rho{\lm3}\in \Vbu\rho{\lm 3}$, we set
\begin{align}
\tilde\rho:=\Ad\lmk \Vrl\rho{\lm3}\rmk\rho\in \Obun{\lm3}\subset\Obul.
\end{align}
Then, because $\Vrl\rho{\lm3}\in (\rho,\tilde\rho)_U$,
by Lemma \ref{lem42},
we have
\begin{align}\label{usagi}
\begin{split}
\hb \sigma{\tilde\rho}\lmk \unit_\sigma\otimes \Vrl\rho{\lm3} \rmk
=\lmk \Vrl\rho{\lm3}\otimes \unit_\sigma\rmk
\hb \sigma\rho,
\end{split}
\end{align}
for each $\sigma\in \Obun{\lm 1(t)}$.
Let $\lm{r2},\lm{r4}\in \Crbd$ with
$\lz\leftarrow_r \lm{r2}$,
$\lm 1\leftarrow \lm{r4}$,
$\lm 3\subset \lm{r4}$.
Then for any $\Vrl{\tilde\rho}{\lm {r2}(t_2)}\in \VUbd{\tilde\rho}{\lm {r2}(t_2)}$,
$t_2\ge 0$, we have
$\Vrl{\tilde\rho}{\lm {r2}(t_2)}\in \pbd\lmk \caA_{\lm {r4}^c\cap\hu}\rmk'\cap\fbd$
for $t_2\ge 0$ large enough.
Hence we  have
\begin{align}
\begin{split}
&\sup_{\sigma\in \Obun{\lm 1(t)}}
\lV 
\hb \sigma {\tilde\rho}-\unit
\rV
=
\sup_{\sigma\in \Obun{\lm 1(t)}}
\lim_{t_2\to \infty}
\lV
\Vrl{\tilde\rho}{\lm {r2}(t_2)}^*
\lmk
\Tbd\sigma{\lm {r0}}\unit-\id
\rmk\lmk \Vrl{\tilde\rho}{\lm {r2}(t_2)}\rmk
\rV\\
&\le
\sup_{\sigma\in \Obun{\lm 1(t)}}
\lV
\lmk\left.
\Tbd\sigma{\lm {r0}}\unit-\id
\rmk\right\vert_{\pbd\lmk \caA_{\lm {r4}^c\cap\hu}\rmk'\cap\fbd}
\rV.
\end{split}
\end{align}
By the same argument as Lemma \ref{lem82}, the last line
goes to $0$ as $t\to\infty$.
Combining this and (\ref{usagi}), we obtain the result.
\end{proof}
\change{This property of $\hb\sigma \rho$ will be relevant later.}
\section{A bulk-to-boundary map}
In this section, we construct a map from bulk to boundary under the Assumption \ref{assump7}.
Recall from Lemma 2.13 of \cite{MTC}, for any
$\theta\in\bbR$, $\varphi\in (0,\pi)$, $\rho\in \Obk$, $\lz\in C_{(\theta,\varphi)}$,
$\Vrl{\rho}{\lz}\in \Vbk{\rho}{\lz}$, 
we may define an endomorphism $T_{\rho}^{(\theta,\varphi)\lz \Vrl{\rho}{\lz}}$
on $\caB_{(\theta,\varphi)}$.
We use this symbol for the rest of this paper.
Assumption \ref{assump7} gives a $*$-isomorphism between $\fbk$ and $\fbd$.
\begin{lem}\label{lem10}
Consider the setting in subsection \ref{setting2}. Assume Assumption \ref{assump7}.
Then there exists a $*$-isomorphism 
$\tau: \fbk\to \fbd$ such that
\begin{align}
\left.\tau\pbk\right\vert_{\abd}=\pbd.
\end{align}
\end{lem}
\begin{proof}
By Assumption \ref{assump7}, for each $\Lambda\in\CUbk$, there exists
a $*$-isomorphism $\tau_\Lambda : \pbk\lmk {\caA_{\lhu}}\rmk''\to \pbd\lmk {\caA_{\lhu}}\rmk''$
such that 
\begin{align}\label{hokkaido}
\tau_\Lambda\pbk(A)=\pbd(A),\quad\text{ for all} \;A\in {\caA_{\lhu}}.
\end{align}
From (\ref{hokkaido}) and $\sigma$-weak continuity of $\tau_\Lambda$s, for any $\Lambda_1,\Lambda_2\in\CUbk$
with $\Lambda_1\subset\Lambda_2$, we have 
\[
\tau_{\Lambda_2}\vert_{\pbk\lmk {\caA_{\Lambda_1\cap\hu}}\rmk''}
=\tau_{\Lambda_1}.
\]
 By the upward-filtering property of $\CUbk$, this means
that we can define an isometric $*$-homomorphism 
$\tau_0$ from $\caF_{\mopbk}^{U(0)}$ onto $\caF_{\mopbd}^{U(0)}$
by
\begin{align}
\tau_0(x)=\tau_\Lambda(x), \quad \text{if}\; \;x\in \pbk\lmk {\caA_{\lhu}}\rmk'',\quad \Lambda\in\CUbk.
\end{align}
Because $\tau_0$ is an isometry, we can extend it to an isometric $*$-homomorphism
$\tau: \fbk\to \fbd$.
It is onto because $\tau(\fbk)\supset \tau_0(\caF_{\mopbk}^{U(0)})=\caF_{\mopbd}^{U(0)}$
and $\tau(\fbk)$ is closed.
Hence $\tau$ is a $*$-isomorphism.
For any $A\in \aloch$, there exists $\Lambda\in \CUbk$ such that $A\in \caA_{\lhu}$.
Hence we have
\begin{align}
\tau\pbk(A)=\tau_0\pbk\lmk A\rmk=\pbd(A).
\end{align}
As this holds for any $A\in \aloch$, the same is true for any $\abd$.
\end{proof}
\change{In a word, this $\tau$ will give a copy of the bulk theory at the boundary.
In order to carry it out, we would like to restrict our bulk endomorphisms to
$\fbk$. The following Lemma will be used for that purpose.}
\begin{lem}\label{lem13}
Consider the setting in subsection \ref{setting2}.
We have
\begin{align}\label{hyogo}
\gu\cap \pbk\lmk \caA_{\hd}\rmk'\subset\fbk.
\end{align}
\end{lem}
\begin{proof}
For any $x\in \gu\cap \pbk\lmk \caA_{\hd}\rmk'$, by the definition of $\gu$,
there exists a sequence of $\Lambda_n\in \CUbk$ and $x_n\in \pbk\lmk\caA_{\Lambda_n}\rmk''$, $n\in\bbN$
such that $\lim_{n\to\infty}\lV x-x_n\rV=0$.
Set $S_n:=\Lambda_n\cap\hd$. Note that $S_n$ is a finite set, because $\Lambda_n\in \CUbk$.
We apply Lemma \ref{lem12} replacing
$\Gamma$, $S$, $(\caH,\pi)$ with $\Lambda_n$, $S_n$, $(\hbk,\pbk\vert_{\caA_{\Lambda_n}})$.
Then we obtain projections of norm $1$ $\caE_{S_n}: \caB(\hbk)\to \pbk(\caA_{S_n})'$
such that $\caE_{S_n}\lmk \pbk\lmk {\caA_{\Lambda_n}} \rmk''\rmk=
\pbk \lmk {\caA_{\Lambda_n\setminus S_n}} \rmk''=\pbk \lmk {\caA_{\Lambda_n\cap H_U}} \rmk''
\subset \fbk$.
From this, for $x_n\in \pbk\lmk\caA_{\Lambda_n}\rmk''$, we have
$\caE_{S_n}(x_n)\in \fbk$.
Because $x\in  \pbk\lmk \caA_{\hd}\rmk'\subset \pbk \lmk \caA_{S_n} \rmk'$,
we have $\caE_{S_n}(x)=x$.
Then, 
\begin{align}
\lV
x-\caE_{S_n}(x_n)
\rV=\lV \caE_{S_n}\lmk x-x_n\rmk\rV\le \lV x-x_n\rV\to 0.
\end{align}
Because $\caE_{S_n}(x_n)\in \fbk$, this means $x\in \fbk$.
\end{proof}
\change{This Lemma guarantees that the restriction of $ \Tbkv\rho{\Lambda_0}$ below to $\fbk$
gives an endomorphism on $\fbk$.}
\begin{lem}\label{lem15}
 Consider the setting in subsection \ref{setting2}.
 Assume Assumption \ref{assum3}.
 Then for any $\Lambda_0\in \CUbk$ with $\Lambda_0\subset H_U$,
 $\rho\in \Obk$, $\Vrl\rho{\Lambda_0}\in\Vbk \rho {\Lambda_0}$,
 we have 
 \begin{align}
 \Tbkv\rho{\Lambda_0}\lmk\fbk\rmk\subset \fbk.
 \end{align}
\end{lem}
\begin{proof}
We first claim $\Tbkv\rho{\Lambda_0}\lmk\pbk\lmk\caA_{\lhu}\rmk''\rmk\subset \pbk\lmk\caA_{\hd}\rmk'$ 
for any $\Lambda\in \CUbk$.
To see this, note that for any $A\in\caA_{\hd}\subset\caA_{\Lambda_0^c}$, we have
$\Tbkv\rho{\Lambda_0}\lmk \pbk(A)\rmk=\pbk(A)$ by Lemma 2.13 of \cite{MTC}.
Therefore, for any
 $x\in \pbk\lmk\caA_{\lhu}\rmk''$, and $A\in \caA_{\hd}$ we have
 \begin{align}\label{shika}
 \left[
 \pbk(A), \Tbkv\rho{\Lambda_0}(x)
 \right]
 = \Tbkv\rho{\Lambda_0}
 \lmk
\left[
\pbk(A),x
\right]
\rmk=0,
 \end{align}
 proving the claim.
 
 Next, we claim that for any $\Lambda\in\CUbk$, 
$
\Tbkv\rho{\Lambda_0}\lmk \pbk\lmk \caA_{\lhu}\rmk''\rmk\subset\gu$.
To see this, note from the upward-filtering property of $\CUbk$,
for $\Lambda,\Lambda_0\in\CUbk$, there exists a $\Lambda_1\in \CUbk$
such that $\Lambda,\Lambda_0\subset \Lambda_1$.
 Then for any
 $A\in \caA_{(\Lambda_1)^c}\subset \caA_{(\Lambda_0)^c}$,
 we have $\Tbkv\rho{\Lambda_0}\lmk \pbk(A)\rmk=\pbk(A)$,
 and the same argument as in (\ref{shika}) implies 
 $\Tbkv\rho\lz\lmk \pbk(\caA_{\lhu})''\rmk\subset \pbk\lmk\caA_{\Lambda_1^c}\rmk'\subset \gu$, proving the claim.
  Here we used Lemma \ref{lem11}.
  
  Combining these with Lemma \ref{lem13}, for any $\Lambda\in\CUbk$, 
 we obtain
  \begin{align}
  \Tbkv\rho{\Lambda_0}\lmk \pbk\lmk \caA_{\lhu}\rmk''\rmk\subset\gu\cap \pbk\lmk\caA_{\hd}\rmk'
\subset\fbk.
  \end{align}
  By the definition of $\fbk$, we obtain $ \Tbkv\rho{\Lambda_0}\lmk \fbk\rmk\subset\fbk$.
\end{proof}
\change{Intertwiners between such endomorphisms are also in $\fbk$.}
\begin{lem}\label{lem16}
 Consider the setting in subsection \ref{setting2}.
 Assume Assumption \ref{assum3}.Then for any $\rho,\sigma\in \Obk$,
 $\lm1,\lm2\in\CUbk$ with $\lm1,\lm2\subset H_U$, we have
 $\lmk\Tbkv\rho{\lm 1},\Tbkv\sigma{\lm 2}\rmk\subset\fbk$.
\end{lem}
\begin{proof}
Because $\Tbkv\rho{\lm1}$ and $\Tbkv\rho{\lm2}$ are identity on $\pbk\lmk\caA_{(\Lambda_1\cup\Lambda_2)^c}\rmk$
 we have $\lmk\Tbkv\rho{\lm1},\Tbkv\sigma{\lm2}\rmk\subset \pbk\lmk\caA_{(\Lambda_1\cup\Lambda_2)^c}\rmk'$.
From the upward filtering property of $\CUbk$ and Lemma \ref{lem13},
we obtain $\pbk\lmk\caA_{(\Lambda_1\cup\Lambda_2)^c}\rmk'\subset \fbk$.

\end{proof}
\change{
Now, we translate these bulk objects to boundary objects using the $*$-isomorphism
$\tau$.
}
\begin{lem}\label{lem17}
Consider the setting in subsection \ref{setting2}.
 Assume Assumption \ref{assum3} and Assumption \ref{assump7}.
 Let $(\lz,\lzr)\in \pc$, and $\tau : \fbk\to \fbd$ the $*$-isomorphism given in
 Lemma \ref{lem10}.
 Then the following hold.
 \begin{description}
 \item[(i)]
 For any $\rho\in \Obkl$, 
 \begin{align}\label{aomori}
  F_0^{\llz}(\rho)=\tau\;\Tbk\rho\lz\unit\;\tau^{-1}\pbd
 \end{align}
 belongs to ${\Obul}$
 and defines a map $ F_0^{\llz}: \Obkl\to {\Obul}$.
 \item[(ii)]For any $\rho\in \Obkl$, $(\Lambda,\Lambda_r)\in \pc$ and $\Vrl\rho\Lambda\in \Vbk\rho\Lambda$,
 $\tau(\Vrl\rho\Lambda)$ belongs to $\Vbu{F_0^{\llz}(\rho)}{\Lambda}$.
 \item[(iii)]
 For any $\rho,\sigma\in \Obkl$ and $R\in (\rho,\sigma)$,
 \begin{align}\label{kyoto}
 \begin{split}
 F_1^{\llz}(R):=\tau(R)\in \lmk F_0^{\llz}(\rho), F_0^{\llz}(\sigma)\rmk_U
 \end{split}
 \end{align} 
  and defines a map $ F_1^{\llz}: (\rho,\sigma)\to \lmk F_0^{\llz}(\rho), F_0^{\llz}(\sigma)\rmk_U$.
 \item[(iv)]
  For any $\rho\in \Obkl$, $(\lm 1,\lm {r1})\in \pc$ and
  $\Vrl\rho{\lm 1}\in \Vbk\rho{\lm {1}}$,
 \begin{align}\label{akita}
 \left. \Tbd{ F_0^{\llz}(\rho)}{\lm {r1}}{\tau(\Vrl\rho{\lm 1})}\right\vert_{\fbd}
 =\left.\tau\;\Tbkv\rho{\lm 1}\;\tau^{-1}\right\vert_{\fbd}.
 \end{align}

 \end{description}
\end{lem}
\begin{proof}
 For any $\rho\in \Obkl$,
the composition (\ref{aomori}) makes sense because 
$\pbk(\abd)\subset \fbk$, Lemma \ref{lem10}, Lemma \ref{lem15}.
For any $\rho\in \Obkl$, $(\Lambda,\Lambda_r)\in \pc$ and $\Vrl\rho\Lambda\in \Vbk\rho\Lambda$, by Lemma 2.14 of \cite{MTC} and Lemma \ref{lem16} we have $\Vrl\rho\Lambda\in 
\lmk\Tbk\rho\lz\unit,\Tbkv\rho\Lambda\rmk\subset\fbk$.
Hence $\tau(\Vrl\rho\Lambda)\in\fbd$ is well-defined, and
by definition, 
\begin{align}
\begin{split}
&\left.\Ad\lmk \tau(\Vrl\rho\Lambda)\rmk F_0^{\llz}(\rho)\right\vert_{\caA_{\lhuc}}
=\left.\Ad\lmk \tau(\Vrl\rho\Lambda)\rmk  \tau\;\Tbk\rho\lz\unit\;\tau^{-1}\pbd\right\vert_{\caA_{\lhuc}}\\
&=\left.\Ad\lmk \tau(\Vrl\rho\Lambda)\rmk  \tau\;\Tbk\rho\lz\unit\;\pbk\right\vert_{\caA_{\lhuc}}
=\left.\Ad\lmk \tau(\Vrl\rho\Lambda)\rmk  \tau\;\rho\right\vert_{\caA_{\lhuc}}
=\left.\tau\;\pbk\right\vert_{\caA_{\lhuc}}=\left.\pbd\right\vert_{\caA_{\lhuc}}.
\end{split}
\end{align}
Hence $\tau(\Vrl\rho\Lambda)\in \Vbu{F_0^{\llz}(\rho)}{\Lambda}$
for any $\rho\in \Obkl$, $(\Lambda,\Lambda_r)\in \pc$ and $\Vrl\rho\Lambda\in \Vbk\rho\Lambda$.
In particular, if we take $\Lambda=\lz$ and $\Vrl\rho\lz=\unit$,
we see that $\left.F_0^{\llz}(\rho)\right\vert_{\caA_{\lhucz}}=\left.\pbd\right\vert_{\caA_{\lhucz}}$.
This implies  (i) and (ii).
(Note that for any $\ld_1\in \CUbk$, there exists a $(\Lambda,\Lambda_r)\in \pc$ such that
$\ld\subset \ld_1\cap\hu$.)
For any $\rho,\sigma\in \Obkl$ and $R\in (\rho,\sigma)$,
$R\in \lmk\Tbk\rho\lz\unit, \Tbk\sigma\lz\unit\rmk\subset \fbk$ by Lemma \ref{lem16}
and $\tau(R)\in \fbd$ is well-defined.
It is then straight forward to check that $\tau(R)\in \lmk F_0^{\llz}(\rho), F_0^{\llz}(\sigma)\rmk_U$,
proving (iii).
To show (iv), let  $\rho\in \Obkl$, $(\lm 1,\lm {r1})\in \pc$ and
  $\Vrl\rho{\lm 1}\in \Vbk\rho{\lm {1}}$.
Note that $ \left.  \Tbd{ F_0^{\llz}(\rho)}{\lm {r1}}{\tau(\Vrl\rho{\lm 1})}\right\vert_{\fbd}$.
is $\sigma$-weak continuous on $\pbd\lmk\caA_{\lhu}\rmk''$ for any $\Lambda\in\CUbk$,
because  for any $\Lambda\in\CUbk$, there is some $\Lambda_l\in \Clbd$ such that $\lhu\subset \Lambda_l$.
(Recall Lemma \ref{lem20}.)
Therefore, $\tau^{-1}\; \left. \Tbd{ F_0^{\llz}(\rho)}{\lm {r1}}{\tau(\Vrl\rho{\lm 1})}\right\vert_{\fbd}\;\tau$ is $\sigma$-weak continuous on
$\pbk\lmk\caA_{\lhu}\rmk''$ for any $\Lambda\in\CUbk$.
On the other hand, for any $A\in \abd$,
\begin{align}\begin{split}
\tau^{-1}\; \left. \Tbd{ F_0^{\llz}(\rho)}{\lm {r1}}{\tau(\Vrl\rho{\lm 1})}\right\vert_{\fbd}\;\tau\pbk(A)=
\tau^{-1}\; \left. \Tbd{ F_0^{\llz}(\rho)}{\lm {r1}}{\tau(\Vrl\rho{\lm 1})}\right\vert_{\fbd}\;\pbd(A)=
\tau^{-1}\;\Ad\lmk \tau(\Vrl\rho{\lm 1})\rmk 
F_0^{\llz}(\rho)(A)\\
=\tau^{-1}\;\Ad\lmk \tau(\Vrl\rho{\lm 1})\rmk 
\tau\;\Tbk\rho\lz\unit\;\tau^{-1}\pbd(A)
=\;\Ad\lmk \Vrl\rho{\lm 1}\rmk 
\;\Tbk\rho\lz\unit\;\pbk(A)
=
\Tbkv\rho{\lm 1}\;\pbk(A).
\end{split}
\end{align}
This proves (iv), as in the proof of Lemma 2.14 \cite{MTC}, via the $\sigma$weak continuity above.

\end{proof}
\change{Braidings are translated accordingly.}
\begin{lem}\label{lem47}
Consider the setting in subsection \ref{setting2}.
 Assume Assumption \ref{assum3}, Assumption \ref{assum80}, Assumption \ref{assum80l} and Assumption \ref{assump7}.
 Let $(\lz,\lzr)\in\pc$.
 Let $\epsilon_-^{(\lz)}(\rho:\sigma)$ be the braiding defined in Definition 4.12 \cite{MTC} for
 $\rho,\sigma\in \Obkl$.
 Then we have 
 \[
 F_1^{\llz}\lmk\epsilon_-^{(\lz)}(\rho:\sigma)\rmk=\hb{F_0^{\llz}(\rho)}{F_0^{\llz}(\sigma)}.
 \]
\end{lem}
\begin{proof}
Let $\lm 1,\lm 2\in \CUbk=\caC_{(\frac{3\pi}2,\frac \pi 2)}$ with $\lm 1,\lm 2\subset \hu$
such that $\lm 1\leftarrow_{(\frac{3\pi}2,\frac\pi 2)}\lm 2$,
$\lm 1\perp_{(\frac{3\pi}2,\frac\pi 2)}\lm 2$ (see Definition 2.6 Definition 4.8 of \cite{MTC}).
Let $\Vrl\rho{\lm 1(t_1)}\in \Vbk\rho{\lm 1(t_1)}$, $t_1\ge 0$,
  $\Vrl\sigma{\lm 2(t_2)}\in \Vbk\sigma{\lm 2(t_2)}$, $t_2\ge 0$.
  By Lemma \ref{lem16}, they belong to $\fbk$.
By the Definition 4.12 \cite{MTC}, 
we have
\begin{align}
\epsilon_-^{(\lz)}(\rho:\sigma)
={\mathop{\mathrm {norm}}}-\lim_{t_1,t_2\to\infty}
\lmk\Vrl\sigma{\lm 2(t_2)}\Tbk\sigma\lz\unit\lmk \Vrl\rho{\lm 1(t_1)}\rmk  \rmk^*
\Vrl\rho{\lm 1(t_1)} \Tbk\rho\lz\unit\lmk \Vrl\sigma{\lm 2(t_2)}\rmk.
\end{align}
We have $\epsilon_-^{(\lz)}(\rho:\sigma)\in \fbk$ because of Lemma \ref{lem15}.
Applying $\tau$ to this equation, from Lemma \ref{lem17}, we have
\begin{align}\begin{split}
&\tau\lmk \epsilon_-^{(\lz)}(\rho:\sigma)\rmk\\
&={\mathop{\mathrm {norm}}}-\lim_{t_1,t_2\to\infty}
\lmk\tau\lmk \Vrl\sigma{\lm 2(t_2)}\rmk \Tbd{F_0^{\llz}(\sigma)}\lzr\unit
\lmk \tau\lmk \Vrl\rho{\lm 1(t_1)}\rmk \rmk  \rmk^*
\tau\lmk \Vrl\rho{\lm 1(t_1)} \rmk\Tbd{F_0^{\llz}(\rho)}\lzr\unit\lmk \tau\lmk \Vrl\sigma{\lm 2(t_2)}\rmk\rmk.
\end{split}
\end{align}
Let $\lm {r2}\in \Crbd$ with $\lm 1\leftarrow_r \lm {r2}$, $\arg\lm 2\subset \arg\lm{r2}$,
and let $\lm {r1}\in \Crbd$ with $(\lm 1,\lm {r1})\in\pc$.
Then by Lemma A.2 of \cite{MTC}, for any $t\ge 0$, there exists $t_0(t)\ge 0$
such that
$\lm 2(t_2)\subset \lm {r2}(t)$ for all $t_2\ge t_0(t)$, with $\lim_{t\to\infty} t_0(t)=\infty$.
From Lemma \ref{lem17}, for such $t_2\ge t_0(t)$,
\[
\tau\lmk \Vrl\sigma{\lm 2(t_2)}\rmk\in \Vbu{F_0^{\llz}(\sigma)}{\lm {2}(t_2)}
\subset\VUbd{F_0^{\llz}(\sigma)}{\lm{r2}(t)}.\]
Hence we have
\begin{align}\begin{split}
&\tau\lmk \epsilon_-^{(\lz)}(\rho:\sigma)\rmk\\
&={\mathop{\mathrm {norm}}}-\lim_{t_1,t\to\infty}
\lmk\tau\lmk \Vrl\sigma{\lm 2(t_0(t))}\rmk \Tbd{F_0^{\llz}(\sigma)}\lzr\unit
\lmk \tau\lmk \Vrl\rho{\lm 1(t_1)}\rmk \rmk  \rmk^*
\tau\lmk \Vrl\rho{\lm 1(t_1)} \rmk\Tbd{F_0^{\llz}(\rho)}\lzr\unit\lmk \tau\lmk \Vrl\sigma{\lm 2(t_0(t))}\rmk\rmk\\
&=\hb{F_0^{\llz}(\rho)}{F_0^{\llz}(\sigma)}.
\end{split}
\end{align}

\end{proof}

In \cite{MTC}, we showed that $\Obkl$ with their morphisms $(\rho,\sigma)$,
$\rho,\sigma\in \Obkl$
forms a braided $C^*$-tensor category $\Cabkl$.
It is now easy to see we have a functor from this bulk theory to our boundary theory.
\begin{thm}\label{lem74}
Let $\llz\in\pc$.
Consider the setting in subsection \ref{setting2}, and 
assume Assumption \ref{assum3}, Assumption \ref{assum80}, Assumption \ref{assum80l},
Assumption \ref{assump7}, Assumption \ref{aichi}.
(Note from Lemma \ref{daidai}, Assumption \ref{wakayama} automatically holds.)
Then $F^{\llz} : \Cabkl\to \Cabul$ given by
$F^{\llz}_0$ (\ref{aomori}), $F^{\llz}_1$ (\ref{kyoto})
gives a fully faithful unitary braided tensor functor.
\end{thm}
\change{
\begin{ex}\label{ebToric}
Let us consider the map $F^{\llz}_0$
for the Toric code.
Let $\gamma'$ be an infinite path in $\ld_0$ on the lattice and 
consider endomorphism $\rho_{\gamma'}^Z$ in Example \ref{toric}.
Then we have
\begin{align}
\begin{split}
F^{\llz}_0\lmk\pbk\rho_{\gamma'}^Z\rmk
=\tau\;\Tbk{\pbk\rho_{\gamma'}^Z}\lz\unit\;\tau^{-1}\pbd
=\tau\;\Tbk{\pbk\rho_{\gamma'}^Z}\lz\unit\;\pbk
=\tau \pbk\rho_{\gamma'}^Z
=\pbd \rho_{\gamma'}^Z,
\end{split}
\end{align}
matching with \cite{wa}.
\end{ex}
}

\kakunin{
\begin{proof}
The map $F^{\llz}$ is a functor
because any $\rho,\sigma,\gamma\in \OUbkl$, $R\in (\rho,\sigma)$, $S\in (\sigma,\gamma)$,
we clearly have
$F^{\llz}_1(SR)=F^{\llz}_1(S)F^{\llz}_1(R)$ and $F^{\llz}_1(\unit)=\unit$.
It is linear on morphism.

Note that $F^{\llz}_0(\pbk)=\pbd$, because $\Tbk{\pbk}{\lz}\unit=\id$.
Then $\unit : \pbd\to F^{\llz}_0(\pbk)=\pbd$ gives an isomorphism.

We also see that
$\unit : F^{\llz}_0(\rho)\otimes F^{\llz}_0(\sigma)\to F^{\llz}_0(\rho\otimes\sigma)$
give natural isomorphisms.
In fact, note that from (\ref{akita}), Lemma \ref{lem20}, and Lemma 3.6 of \cite{MTC},
\begin{align}\label{hirosaki}
\begin{split}
&F^{\llz}_0(\rho)\otimes F^{\llz}_0(\sigma)
=\Tbd{F^{\llz}_0(\rho)}{\lzr}{\unit}\Tbd{F^{\llz}_0(\sigma)}{\lzr}{\unit}\pbd
=\tau\Tbk\rho{\lz}\unit \Tbk\sigma{\lz}\unit\tau^{-1}\pbd\\
&=\tau\Tbk{\rho\otimes \sigma}{\lz}\unit \tau^{-1}\pbd
=F^{\llz}_0(\rho\otimes\sigma).
\end{split}
\end{align}
Furthermore,  let $\rho,\rho',\sigma,\sigma'\in \OUbkl$
and $R\in (\rho,\rho')$, $S\in (\sigma,\sigma')$.
Then, from (\ref{akita}),
\begin{align}
\begin{split}
&F^{\llz}_1(R)\otimes F^{\llz}_1(S)
=F^{\llz}_1(R)\Tbd{F^{\llz}_0(\rho)}{\lz}\unit \lmk  F^{\llz}_1(S)\rmk
=\tau(R)\Tbd{F^{\llz}_0(\rho)}{\lz}\unit\lmk  \tau (S)\rmk\\
&=\tau(R)\tau\Tbk\rho{\lz}\unit\tau^{-1}\tau(S)
=\tau\lmk R\Tbk\rho{\lz}\unit\lmk S\rmk\rmk
=F^{\llz}_1\lmk
R\otimes S
\rmk,
\end{split}
\end{align}
proving the naturality.
These isomorphisms trivially satisfies the diagrams given in Definition 2.1.3 \cite{NT}.
Hence $F^{\llz}$ is a tensor functor.
Clearly, it is unitary.
It is a central functor from Lemma \ref{lem17}, Lemma \ref{lem42},  Lemma \ref{lem43}, Lemma \ref{lem47},
Lemma \ref{panama}.
\end{proof}}
%
%

\section{A boundary to bulk map}
In this section, we construct an inverse of the map constructed in the previous section.
\begin{lem}\label{lem86}
Consider the setting in subsection \ref{setting2}, and assume Assumption \ref{assump7}, Assumption \ref{assum3},
and Assumption \ref{oo}.
Then for any $\rho\in\OrUbd$, there exists a unique $*$-homomorphism
$G(\rho) : \abk\to \caB(\hbk)$ such that
\begin{align}
G(\rho)(A)
=\tau^{-1}\rho(A),\quad A\in\abd,\label{nagano}\\
G(\rho)(B)=\pbk(B),\quad B\in \caA_{\hd}\label{gifu}.
\end{align}
If furthermore $\rho\in \Obu$, then
$G(\rho)\in \OUbk=\Obk$.
In particular, for $\llz\in\pc$, $\rho\in \Obul$, we have
$G(\rho)\in \OUbkl=\Obkl$, and $F_0^{{\llz}}\lmk G(\rho)\rmk=\rho$.
For $\hat\rho\in \OUbkl=\Obkl$, we have
$G(F_0^{\llz}(\hat\rho))=\hat\rho$.
The map preserves tensor i.e., $G(\rho\otimes\sigma)=G(\rho)\otimes G(\sigma)$
for any $\rho,\sigma\in \Obul$.
\end{lem}
\begin{proof}
Let $\llz\in\pc$ and $\rho\in\OrUbd$, $\Vrl\rho{\lzr}\in \VUbd\rho{\lzr}$.
Note that for any $A\in\abd$,
$$\rho(A)=\Ad\lmk \Vrl{\rho}\lzr^*\rmk\Tbdv\rho{\lzr}\pbd(A)\in\fbd.$$
 Therefore, $\tau^{-1}\rho(A)$ is well-defined.
 Hence there is a representation
 $G_1(\rho)$( resp.$G_2(\rho)$) of $\abd$ (resp $\caA_{\hd}$)
 on $\hbk$, given by the right hand side of
 (\ref{nagano})( resp. (\ref{gifu})). Because $G_1(\rho)(\abd)\subset\fbk$
 and $\fbk$ commutes with $\pbk(\caA_{\hd})$
 by definition, the range of $G_1(\rho)$ and $G_2(\rho)$ commutes.
Because $\abd$, $\caA_{\hd}$ are nuclear,
there exists a unique representation $G(\rho) : \abk \to \caB(\hbk)$
 such that $G(\rho)\vert_{\abd}=G_1(\rho)$, $G(\rho)\vert_{\caA_{\hd}}=G_2(\rho)$.
 (Proposition 4.7 IV of \cite{takesaki}.)
 
Now suppose that $\rho\in \Obu$. 
For any $\lm {}\in \CUbk$, take $\Vrl{\rho}{\lm{}}\in \Vbu\rho{\lm {}}\subset\fbd$.
 Then for any $A\in \caA_{\lhuc}$,
 \begin{align}
 \begin{split}
&\Ad\lmk\tau^{-1}\lmk \Vrl{\rho}{\lm{}}\rmk\rmk G(\rho)(A)
 =\Ad\lmk\tau^{-1}\lmk \Vrl{\rho}{\lm{}}\rmk\rmk\lmk \tau^{-1}\rho(A)\rmk\\
& =\tau^{-1}\lmk \Ad\lmk\lmk \Vrl{\rho}{\lm{}}\rmk\lmk \rho(A)\rmk \rmk\rmk
 =\tau^{-1}\lmk\pbd(A)\rmk=\pbk(A).
 \end{split}
 \end{align}
 For any $B\in \caA_{\hd}$, we have
 \begin{align}
\Ad\lmk\tau^{-1}\lmk \Vrl{\rho}{\lm{}}\rmk\rmk G(\rho)(B)
=\Ad\lmk\tau^{-1}\lmk \Vrl{\rho}{\lm{}}\rmk\rmk\lmk \pbk(B)\rmk
=\pbk(B),
 \end{align}
 because $ \pbk(B)\in{\fbk}'$.
 Combining these, we obtain
 \begin{align}
 \left.\Ad\lmk\tau^{-1}\lmk \Vrl{\rho}{\lm{}}\rmk\rmk G(\rho)\right\vert_{\caA_{\Lambda^c}}
 =\left.\pbk\right\vert_{\caA_{\Lambda^c}}
 \end{align}
for any $\Lambda\in\CUbk$, proving $G(\rho)\in \OUbk$.

If $(\lz,\lzr)\in \pc$, $\rho\in \Obul$, then we may take $ \Vrl{\rho}{\lz}=\unit$ above.
Hence $G(\rho)\in \OUbkl=\Obkl$.
For any $A\in\hu$, we have
\begin{align}
\begin{split}
F_0^{\llz}\lmk G(\rho)\rmk(A)
=\tau \Tbk{G(\rho)}\lz\unit\tau^{-1}\pbd(A)
=\tau G(\rho)(A)
=\rho(A).
\end{split}
\end{align}
Hence we have $F_0^{\llz}\lmk G(\rho)\rmk=\rho$.

For $\hat\rho\in \OUbkl=\Obkl$, we have
\begin{align}
\begin{split}
&G(F_0^{\llz}(\hat\rho))
\lmk
A
\rmk
=\tau^{-1}F_0^{\llz}\lmk \hat\rho\rmk(A)
=\tau^{-1}\tau\Tbk{\hat\rho}\lz\unit \tau^{-1}\pbd(A)\\
&=\Tbk{\hat\rho}\lz\unit \pbk(A)
=\hat\rho(A),\quad A\in \abd,\\
&G(F_0^{\llz}(\hat\rho))
(B)
=\pbk(B)=\hat\rho(B),\quad B\in\caA_{\hd}.
\end{split}
\end{align}
Hence
 we have
$G(F_0^{\llz}(\hat\rho))=\hat\rho$.

For any $\rho,\sigma\in \Obul$,
from above and Theorem \ref{lem74}, we have
\begin{align}
\begin{split}
F_0^{\llz}G(\rho\otimes\sigma)
=\rho\otimes\sigma
=F_0^{\llz}G(\rho)\otimes F_0^{\llz}G(\sigma)
=F_0^{\llz}\lmk G(\rho)\otimes G(\sigma)\rmk.
\end{split}
\end{align}
Applying $G$ to this, we have
\begin{align}
\begin{split}
G(\rho\otimes\sigma)
=GF_0^{\llz}\lmk G(\rho)\otimes G(\sigma)\rmk
=G(\rho)\otimes G(\sigma).
\end{split}
\end{align}
This proves the last statement.

\end{proof}
\kakunin{
The information of $\rho\in\Obul$ is not enough
to recover the full theory of the bulk.
Namely, it cannot guarantee in general if $G(\rho)$ satisfies the superselection criterion
of the bulk.
However, if we combine it with other half boundary information, we can do that.
To do so, note that we can re-do all the things we have done for upper half-plane
to lower half-plane, just reflect things with respect to the $x$-axis.
We denote such lower half-plane version(reflected with respect to the $x$-axis version), by adding $(-)$ symbol.
We prepare $\obdm$ a pure state on $\caA_{\hd}$ and
its GNS representation $(\hbdm,\pbdm)$.
The set of cones $\CUbk$ is reflected to $\CUbkm$
algebra $\fbd$, $\fbk$ is reflected to $\fbdm$, $\fbkm$
the equivalence relation $\simeq_U$
 is reflected to $\simeq_{U(-)}$, and so on.

Assumption \ref{assum80}$(-)$, Assumption \ref{assum80l}$(-)$,
Assumption \ref{wakayama}$(-)$, etc, are
Assumption \ref{assum80}, Assumption \ref{assum80l},
Assumption \ref{wakayama} reflected with respect to
$x$-axis.
And assuming such $(-)$-versions of assumptions, 
we obtain braided $C^*$-tensor category $\Obulm$,
boundary to bulk map $G_{(-)}$, and $\tau_{(-)}$, corresponding to
$\tau$ in Lemma \ref{lem17}.

\begin{defn}\label{singapore}
Consider the setting in subsection \ref{setting2}, and 
assume Assumption \ref{assum3}, Assumption \ref{assum80}, Assumption \ref{assum80l},
Assumption \ref{assump7}, Assumption \ref{aichi}, Assumption \ref{wakayama}, and their $(-)$-versions.
Let $\lz\in\CUbk$, $\lzm\in \CUbkm$.
We denote by
$\Ototl$ the set of all $\rho\in \Obul$
which allows some $\rho_-\in \Obulm$ such that
$G(\rho)\simeq G_{(-)}(\rho_-)$.
\end{defn}
For the rest of this section, we consider the setting of Definition \ref{singapore}
and fix some $\varphi_i\in (0,\frac\pi 2)$ such that
$\lz,\lzm\in C_{(0,\varphi_1)}$.
The $\rho_-$ is determined essentially uniquely.
\begin{lem}
For each $\rho\in\Ototl$, $\rho_-\in \Obulm$
satisfying $G(\rho)\simeq G_{(-)}(\rho_-)$ are determined uniquely up to
$\simeq_{U(-)}$.
\end{lem}
\begin{proof}
If $G(\rho)\simeq G_{(-)}(\rho_-)\simeq G_{(-)}(\rho'_-)$, then 
$(-)$-version of Lemma \ref{lem17}, Lemma \ref{lem87} imply
$\rho_-\simeq_{U(-)}\rho_-'$
\end{proof}
Elements in $\Ototl$ gives bulk elements in $\Obkl$.
\begin{lem}\label{sapporo}
In the setting of Definition \ref{singapore}, for any
$\rho\in \Ototl$, we have $G(\rho)\in\Obkl$. 
\end{lem}
\begin{rem}
From this property, we may regard $\rho_-$
as dual object of $\rho$.
\end{rem}
\begin{proof}
By the definition, there are unitaries
$W_\rho$ on $\hbk$ and
$\rho_-\in \Obulm$
such that 
$
\Ad(W_\rho)G_{(-)}(\rho_-)=G(\rho)
$.
For any cone $\ld$ in $\bbZ^2$,
either (i) $\lm 1\in\CUbk$ with
$\lm 1\subset\ld$
or (ii) $\lm 2\in \CUbkm$ with
$\lm 2\subset\ld$ exists.
If (i) occurs,
choose $\Vrl\rho{\lm 1}\in \Vbu\rho{\lm 1}\subset\fbd$
and we have
\begin{align}
\begin{split}
\left.\Ad\lmk \tau^{-1}\lmk \Vrl\rho{\lm 1}\rmk\rmk G(\rho)\right\vert_{\caA_{\ld^c}}=
\left. \pbk\right\vert_{\caA_{\ld^c}}.
\end{split}
\end{align}
If (ii) occures, choose $\Vrl{\rho_-}{\lm 2}\in \Vbum{\rho_-}{\lm 2}\subset\fbdm$
and we have
\begin{align}
\begin{split}
\left.\Ad\lmk \tau_{(-)}^{-1}\lmk \Vrl{\rho_-}{\lm 2}\rmk W_\rho^*\rmk G(\rho)\right\vert_{\caA_{\ld^c}}=
\left. \pbk\right\vert_{\caA_{\ld^c}}.
\end{split}
\end{align}
Hence we have $G(\rho)\in\Obkl$.
\end{proof}
First, $\Otot$ is closed under tensor.
\begin{lem}
In the setting of Definition \ref{singapore}, for any
$\rho,\sigma\in\Ototl$, we have
$\rho\otimes\sigma\in\Ototl$.
\end{lem}
\begin{proof}
Let $\rho,\sigma\in\Otot$.
By the definition, there are unitaries
$W_\rho,W_\sigma$ on $\hbk$ and
$\rho_-,\sigma_-\in \Obulm$
such that 
\begin{align}\label{morioka}
\begin{split}
\Ad(W_\rho)G_{(-)}(\rho_-)=G(\rho),\quad
\Ad(W_\sigma)G_{(-)}(\sigma_-)=G(\sigma)
\end{split}
\end{align}
Note that $W_\sigma\in \pbk\lmk\caA_{\lmk \lz\cup\lzm\rmk^c}\rmk'\subset \caB_{(0,\varphi_1)}$,
because $G(\sigma)$ and $G_{(-)}(\sigma_-)$ are $\pbk$ on $\caA_{\lmk \lz\cup\lzm\rmk^c}$.

From Lemma \ref{sapporo} and its $(-)$ version, 
$G(\rho), G(\sigma), G_{(-)}(\rho_-), G_{(-)}(\sigma_-)\in\Obkl$.
By Lemma 2.14 of \cite{MTC}
and
${\caB_{\lmk \frac{3\pi+\varphi_1}{2},\frac{\varphi_1+\pi}2\rmk}}
\subset \caB_{(\frac{3\pi}2,\frac\pi 2)}\cap\caB_{(0,\varphi_1)}$,
we have
\begin{align}
\begin{split}
T_{\hat\rho}^{(\frac{3\pi+\varphi_1}{2},\frac{\varphi_1+\pi}2)\lz\unit}\
=\left. T_{\hat\rho}^{(0,\varphi_1)\lz\unit}\right\vert_{\caB_{\lmk \frac{3\pi+\varphi_1}{2},\frac{\varphi_1+\pi}2\rmk}}
=\left.\Tbk{\hat\rho}\lz\unit\right\vert_{\caB_{\lmk \frac{3\pi+\varphi_1}{2},\frac{\varphi_1+\pi}2\rmk}},
\end{split}
\end{align}
for any $\hat\rho\in\Obkl$.
Hence for $G(\rho),G(\sigma)\in \Obkl$, we have
\begin{align}
\begin{split}
G(\rho)\otimes G(\sigma)
:=\Tbk{G(\rho)}\lz\unit\Tbk{G(\sigma)}\lz\unit\pbk
=T_{G(\rho)}^{(\frac{3\pi+\varphi_1}{2},\frac{\varphi_1+\pi}2)\lz\unit}
T_{G(\rho)}^{(\frac{3\pi+\varphi_1}{2},\frac{\varphi_1+\pi}2)\lz\unit}\pbk
=T_{G(\rho)}^{(0,\varphi_1)\lz\unit}T_{G(\sigma)}^{(0,\varphi_1)\lz\unit}\pbk.
\end{split}
\end{align}
Furthermore,
because
\begin{align}
\begin{split}
T_{G(\rho)}^{(0,\varphi_1)\lz\unit}\pbk
=G(\rho)
=\Ad(W_\rho)G_{(-)}(\rho_-)
=\Ad(W_\rho)T_{G_{(-)}(\rho_-)}^{(0,\varphi_1)\lzm\unit}\pbk,
\end{split}
\end{align}
and $T_{G(\rho)}^{(0,\varphi_1)\lz\unit}$,
$\Ad(W_\rho)T_{G_{(-)}(\rho_-)}^{(0,\varphi_1)\lzm\unit}$
are $\sigma$-weak continuous 
on $\pbk(\caA_{\Lambda})''$ for $\Lambda\in C_{(0,\varphi_1)}$,
we have $T_{G(\rho)}^{(0,\varphi_1)\lz\unit}=\Ad(W_\rho)T_{G_{(-)}(\rho_-)}^{(0,\varphi_1)\lzm\unit}$.
Similarly we have  $T_{G(\sigma)}^{(0,\varphi_1)\lz\unit}=\Ad(W_\sigma)T_{G_{(-)}(\sigma_-)}^{(0,\varphi_1)\lzm\unit}$.

From those facts and  Lemma \ref{lem86} and their $(-)$ version, we have
\begin{align}
\begin{split}
&G(\rho\otimes\sigma)=G(\rho)\otimes G(\sigma)
=T_{G(\rho)}^{(0,\varphi_1)\lz\unit}T_{G(\sigma)}^{(0,\varphi_1)\lz\unit}\pbk
=\Ad(W_\rho)T_{G_{(-)}(\rho_-)}^{(0,\varphi_1)\lzm\unit}
\Ad(W_\sigma)T_{G_{(-)}(\sigma_-)}^{(0,\varphi_1)\lzm\unit}\pbk\\
&=
\Ad\lmk W_\rho T_{G_{(-)}(\rho_-)}^{(0,\varphi_1)\lzm\unit}\lmk W_\sigma \rmk \rmk
T_{G_{(-)}(\rho_-)}^{(0,\varphi_1)\lzm\unit}
T_{G_{(-)}(\sigma_-)}^{(0,\varphi_1)\lzm\unit}\pbk
=\Ad\lmk W_\rho T_{G_{(-)}(\rho_-)}^{(0,\varphi_1)\lzm\unit}\lmk W_\sigma \rmk \rmk
G_{(-)}(\rho_-)\otimes G_{(-)}(\sigma_-)\\
&=\Ad\lmk W_\rho T_{G_{(-)}(\rho_-)}^{(0,\varphi_1)\lzm\unit}\lmk W_\sigma \rmk \rmk
G_{(-)}(\rho_-\otimes \sigma_-).
\end{split}
\end{align}
Hence we have $\rho\otimes\sigma\in\Ototl$.
\end{proof}

Next $\Otot$ is closed under direct sum
\begin{lem}
In the setting of Definition \ref{singapore}, for any
$\rho,\sigma\in\Ototl$, 
$\gamma\in\Obul$, $u\in (\rho,\gamma)_U$, $v\in (\sigma,\gamma)_U$
isometry with $uu^*+vv^*=\unit$, we have
$\gamma\in\Ototl$.
\end{lem}
\begin{proof}
Let $W_\rho,W_\sigma$ unitaries on $\hbk$ and 
$\rho_-,\sigma_-\in \Obulm$ satisfying (\ref{morioka}).
For $\rho_-,\sigma_-\in\Obulm$, there are
$\gamma_-\in\Obulm$, $u_-\in (\rho,\gamma)_{U(-1)}$, $v_-\in (\sigma,\gamma_{U(-)}$
isometry with $u_-u_-^*+v_-v_-^*=\unit$.
It is straightforward to check
\begin{align}
\begin{split}
G(\gamma)=\Ad\tau^{-1}(u) G(\rho)+\Ad\tau^{-1}(v) G(\sigma),\quad
G_{(-)}(\gamma_-)=\Ad\tau_{(-)}^{-1}(u_-) G_{(-)}(\rho_-)+\Ad\tau_{(-)}^{-1}(v_-) G_{(-)}(\sigma_-)
\end{split}
\end{align}
Setting
\begin{align}
\begin{split}
W_\gamma:=
\tau^{-1}(u)\cdot  W_\rho\cdot  \tau_{(-)}^{-1}(u_-^*)
+\tau^{-1}(v)\cdot  W_\sigma\cdot  \tau_{(-)}^{-1}(v_-^*)\in\caU(\hbk),
\end{split}
\end{align}
we have
\begin{align}
\begin{split}
G(\gamma)
=\Ad(W_\gamma)G_{(-)}(\gamma_-).
\end{split}
\end{align}
Hence we obtain $\gamma\in \Ototl$.
\end{proof}
Finally we show that $\Ototl$ is closed under taking subobject.
\begin{lem}In the setting of Definition \ref{singapore},
let $\rho\in\Ototl$, $p\in (\rho,\rho)_U$ a projection,
$\gamma\in\Obul$, $v\in (\gamma,\rho)_U$ isometry with $vv^*=p$.
Then we have $\gamma\in \Ototl$
\end{lem}
\begin{proof}
It is straightforward to check $G(\gamma)=\Ad\lmk\tau^{-1}(v^*)\rmk G(\rho)$.
 There are unitaries
$W_\rho$ on $\hbk$ and
$\rho_-\in \Obulm$
such that 
$
\Ad(W_\rho)G_{(-)}(\rho_-)=G(\rho)
$.
We claim $\Ad W_\rho^*\lmk \tau^{-1}(p)\rmk\in G_{(-)}(\rho_-)(\abk)'$.
In fact for any $A\in\abk$, we have
\begin{align}
\begin{split}
&\Ad W_\rho^*\lmk \tau^{-1}(p)\rmk\cdot G_{(-)}(\rho_-)(A)
=W_\rho^*\tau^{-1}(p) W_\rho \cdot G_{(-)}(\rho_-)(A) \cdot W_\rho^* W_\rho
=W_\rho^*\tau^{-1}(p) G(\rho)(A) W_\rho\\
&=W_\rho^* G(\rho)(A)\tau^{-1}(p)  W_\rho
=G_{(-)}(\rho_-)(A) W_\rho^*\tau^{-1}(p)W_\rho.
\end{split}
\end{align}
Because $G_{(-)}(\rho_-)\vert_{\caA_{\lzm^c}}=\pbk\vert_{\caA_{\lzm^c}}$,
we have
\begin{align}
\begin{split}
\Ad W_\rho^*\lmk \tau^{-1}(p)\rmk\in
G_{(-)}(\rho_-)(\abk)'\subset G_{(-)}(\rho_-)\lmk \caA_{\lzm^c}\rmk'
=\pbk\lmk {\caA_{\lzm^c}}\rmk'\subset \fbkm,
\end{split}
\end{align}
and $\tau_{(-)}\lmk \Ad W_\rho^*\lmk \tau^{-1}(p)\rmk\rmk$
is a well-defined projection in $\fbdm$
and $\tau_{(-)}\lmk \Ad W_\rho^*\lmk \tau^{-1}(p)\rmk\rmk\in (\rho_-,\rho_-)_{U(-)}$.
For this projection, because $\rho_-\in\Obulm$, there exist
$\gamma_-\in \Obulm$, isometry $v_-\in (\gamma_-,\rho_-)_{U(-)}$
such that $v_-v_-^*=\tau_{(-)}\lmk \Ad W_\rho^*\lmk \tau^{-1}(p)\rmk\rmk$.

Set
$
W_\gamma:=
\tau^{-1}(v^*) W_\rho \tau_{(-)}^{-1}(v_-)
$.
It is a unitary on $\hbk$, and we have
\begin{align}
\begin{split}
&G(\gamma)
=\Ad\lmk\tau^{-1}(v^*)\rmk G(\rho)
=\Ad\lmk\tau^{-1}(v^*)\rmk\Ad(W_\rho)G_{(-)}(\rho_-)
=\Ad\lmk \tau^{-1}(v^*) W_\rho \tau_{(-)}^{-1}(v_-v_-^*)\rmk G_{(-)}(\rho_-)\\
&=\Ad\lmk \tau^{-1}(v^*) W_\rho \tau_{(-)}^{-1}(v_-)\rmk
\Ad\lmk\tau_{(-)}^{-1}(v_-^*) \rmk G_{(-)}(\rho_-)
=\Ad(W_\gamma) G_{(-)}(\gamma_-).
\end{split}
\end{align}
Hence we obtain $\gamma\in\Ototl$.
\end{proof}}
Hence we obtain the following.
\begin{cor}\label{kanazawa}
Consider the setting in subsection \ref{setting2}, and 
assume Assumption \ref{assum3}, Assumption \ref{assum80}, Assumption \ref{assum80l},
Assumption \ref{assump7}, Assumption \ref{aichi} and Assumption \ref{oo}.
The functor $F^{\llz} $ in Theorem \ref{lem74} is a braided monoidal equivalence
between braided $C^*$-tensor categories.
\end{cor}
\begin{proof}
%
The functor $F^{\llz}$ is also essentially surjective from Lemma \ref{lem86}.
Hence it is an equivalence.
\end{proof}
\change{In a word, there is a copy of the bulk theory at the boundary.}
\change{
\begin{rem}\label{bToric}
From this Corollary and the fact that the vacuum and the $m$-particle are not equivalent
in the bulk, we conclude that the $\pbd$ and $\pi_\gamma^X$ in Example
\ref{toric} are not equivalent with respect to
$\simeq_{\caU}$. 
\end{rem}
}

\kakunin{
Recall the definition of center $\caZ(\caC)$
of a strict tensor category $\caC$ (Definition XIII.4.1 of \cite{Kassel}).
Here, we introduce a category $\caZ(\Cabul)$
which requires the asymptotic behavior (\ref{okayama}).
\begin{defn}\label{hiroshima}
Consider the setting in subsection \ref{setting2}, and 
assume Assumption \ref{assum3}, Assumption \ref{assum80}, Assumption \ref{assum80l},
Assumption \ref{assump7}, Assumption \ref{aichi}, Assumption \ref{wakayama}.
Let $\lz\in\CUbk$.
An object of $\caZ(\Cabul)$ is a pair
$(\rho, \hat\iota(\rho : -))$ where $\rho\in \Obul$ 
and $\hat\iota(\rho : -)$ is a family of natural isomorphisms
\begin{align}
\hat\iota(\rho : \sigma) :
\rho\otimes\sigma \to\sigma\otimes\rho,
\end{align}
defined for all objects $\sigma\in\Obul$ such that
\begin{align}\label{shimane}
\lim_{t\to\infty}
\sup_{\sigma\in O^{rU}_{\mathop{\mathrm{bd},\lm {r2}(t)}}}
\lV
\hat\iota(\rho :\sigma)-\unit
\rV=0,
\end{align}for any $\lm {r2}\in\Crbd$ with $\lz\leftarrow_r\lm {r2}$, $\lm {r2}\subset \lzr$
and for all objects
$\sigma,\gamma\in \Obul$ we have 
\begin{align}\label{kouchi}
\hat\iota\lmk {\rho}: {\sigma\otimes \gamma}\rmk
=\lmk \unit_{\sigma}\otimes\hat\iota\lmk \rho: \gamma\rmk\rmk\cdot
\lmk \hat\iota\lmk \rho :\sigma\rmk\otimes\unit_\gamma\rmk.
\end{align}
The morphism from $(\rho, \hat\iota(\rho : -))$ to $(\sigma, \hat\iota(\sigma : -))$
is a morphism $R: \rho\to\sigma$ in $\Cabul$
such that for each object $\gamma \in \Cabul$, we have
\begin{align}
\begin{split}
\hat\iota(\sigma: \gamma)
\lmk R\otimes \unit_\gamma\rmk
=\lmk \unit_\gamma\otimes R\rmk\hat\iota(\rho:\gamma)
\end{split}
\end{align}
\end{defn}

By Lemma \ref{lem83}, we have $\hat\iota(\rho:\sigma)=\hb{\rho}\sigma$.
From this and the naturality of $\hb{\rho}\sigma$, Lemma \ref{lem42},
the set of all morphisms in  $\caZ(\Cabul)$ from $(\rho, \hat\iota(\rho : -))$ to $(\sigma, \hat\iota(\sigma : -))$
is $(\rho,\sigma)_U$.

As in Theorem XIII.4.2 of \cite{Kassel}, $\caZ(\Cabul)$
is a braided $C^*$-tensor category with
the tensor product
\begin{align}
(\rho, \hat\iota(\rho : -))\otimes(\sigma, \hat\iota(\sigma : -))
:=\lmk
\rho\otimes\sigma, \hat\iota(\rho\otimes\sigma: -)
\rmk.
\end{align}
Following Theorem XIII.4.2 of \cite{Kassel}, we define the morphism $ \hat\iota(\rho\otimes\sigma: -)$
by
\begin{align}\label{yamaguchi}
\begin{split}
 \hat\iota(\rho\otimes\sigma: \gamma)
 :=\lmk\hat\iota(\rho : \gamma)\otimes\unit_\sigma\rmk
 \lmk \unit_\rho\otimes \hat\iota(\sigma,\gamma)\rmk,
\end{split}
\end{align}
but because of Lemma \ref{lem83}, it turns out
this is equal to $\hb{\rho\otimes\sigma}{-}$,
the unique family of morphism
satisfying the condition in Definition \ref{hiroshima}
for $\rho\otimes\sigma$.
For morphism $R$ from
$(\rho, \hat\iota(\rho : -))$ to $(\rho', \hat\iota(\rho' : -))$
and  morphism $S$ from
$(\sigma, \hat\iota(\sigma : -))$ to $(\sigma', \hat\iota(\sigma' : -))$,
the tensor product of $R$ and $S$ in $\caZ(\Cabul)$
by the tensor product $R\otimes S$ in $\Cabul$.
By Lemma \ref{lem83} and Lemma \ref{lem42}, 
this $R\otimes S$ is in fact a morphism
from
$(\rho, \hat\iota(\rho : -))\otimes (\sigma, \hat\iota(\sigma : -))$ to 
$(\rho', \hat\iota(\rho' : -))\otimes (\sigma', \hat\iota(\sigma' : -))$.
With $\hat\iota(\pbd : \sigma)=\unit$, $\sigma\in \Obul$,
$(\pbd,\hat\iota (\pbd,-))$ is the tensor unit 
of $\Cabul$ and 
$\unit :  (\pbd, \hat\iota(\pbd : -))\otimes (\rho, \hat\iota(\rho : -))\to (\rho, \hat\iota(\rho:-))$,
$\unit :  (\rho, \hat\iota(\rho : -))\otimes (\pbd, \hat\iota(\pbd : -))\to (\rho, \hat\iota(\rho:-))$
give left and right multiplication of the unit.
Note that the tensor unit is irreducible.
Furthermore, $\unit : \lmk (\rho, \hat\iota(\rho : -))\otimes (\sigma, \hat\iota(\sigma : -)\rmk
\otimes
(\gamma, \hat\iota(\gamma : -))
\to (\rho, \hat\iota(\rho : -))\otimes 
\lmk (\sigma, \hat\iota(\sigma : -))\otimes
(\gamma, \hat\iota(\gamma : -))\rmk$ is 
the associativity isomorphism.

Next we consider direct sum.
For objects $(\rho, \hat\iota(\rho : -)), (\sigma, \hat\iota(\sigma : -))$
of $\caZ\lmk\Cabul\rmk$, 
let $\gamma\in \Obul$, $u\in (\rho,\gamma)_U$, $v\in (\sigma,\gamma)_U$ isometry
with $uu^*+vv^*=\unit$.
Then, for any $\zeta\in \Obul$, set
\begin{align}
\hat\iota(\gamma:\zeta):=
\Tbd\zeta{\lzr}\unit(u)\hat\iota(\rho,\zeta) u^*
+\Tbd\zeta{\lzr}\unit(v)\hat\iota(\sigma,\zeta) v^*
\in \caU(\fbd).
\end{align}
Then for any $A\in\abd$, we have
\begin{align}
\begin{split}
&\hat\iota(\gamma,\zeta)\cdot \gamma\otimes\zeta(A)
=\lmk
\Tbd\zeta{\lzr}\unit(u)\hat\iota(\rho,\zeta) u^*
+\Tbd\zeta{\lzr}\unit(v)\hat\iota(\sigma,\zeta) v^*
\rmk \Tbd\gamma{\lzr}\unit\Tbd\zeta{\lzr}\unit\pbd(A)\\
&=
\Tbd\zeta{\lzr}\unit(u)\hat\iota(\rho,\zeta)   \lmk \Tbd\rho{\lzr}\unit\Tbd\zeta{\lzr}\unit\pbd(A)\rmk
u^*
+\Tbd\zeta{\lzr}\unit(v)\hat\iota(\sigma,\zeta) \lmk \Tbd\sigma{\lzr}\unit\Tbd\zeta{\lzr}\unit\pbd(A)\rmk
v^*\\
&=\Tbd\zeta{\lzr}\unit(u) \cdot \Tbd\zeta{\lzr}\unit\Tbd\rho{\lzr}\unit\pbd(A)
\hat\iota(\rho,\zeta)  
u^*
+\Tbd\zeta{\lzr}\unit(v)\cdot \Tbd\zeta{\lzr}\unit\Tbd\sigma{\lzr}\unit\pbd(A)
\hat\iota(\sigma,\zeta) 
v^*\\
&=\Tbd\zeta{\lzr}\unit\lmk u\Tbd\rho{\lzr}\unit \pbd(A)\rmk
\hat\iota(\rho,\zeta)  
u^*
+\Tbd\zeta{\lzr}\unit\lmk v\Tbd\sigma{\lzr}\unit \pbd(A)\rmk
\hat\iota(\sigma,\zeta) 
v^*\\
&=\Tbd\zeta{\lzr}\unit\lmk \Tbd\gamma{\lzr}\unit \pbd(A)\cdot u\rmk
\hat\iota(\rho,\zeta)  
u^*
+\Tbd\zeta{\lzr}\unit\lmk \Tbd\gamma{\lzr}\unit \pbd(A)\cdot v\rmk
\hat\iota(\sigma,\zeta) 
v^*\\
&=\lmk \zeta\otimes\gamma(A) \rmk\hat\iota(\gamma,\zeta).
\end{split}
\end{align}
Hence we have $\hat\iota(\gamma,\zeta)\in (\gamma\otimes\zeta,\zeta\otimes\gamma)$.
It is also natural. In fact, for any $\zeta,\zeta'\in \Obul$ and $S\in (\zeta,\zeta')_U$,
we have
\begin{align}
\begin{split}
&(S\otimes\unit_\gamma)\hat\iota(\gamma,\zeta)
=S
\lmk
\Tbd\zeta{\lzr}\unit(u)\hat\iota(\rho,\zeta) u^*
+\Tbd\zeta{\lzr}\unit(v)\hat\iota(\sigma,\zeta) v^*
\rmk\\
&=\Tbd{\zeta'}{\lzr}\unit(u) S\hat\iota(\rho,\zeta) u^*
+\Tbd{\zeta'}{\lzr}\unit(v) S\hat\iota(\sigma,\zeta) v^*
=\Tbd{\zeta'}{\lzr}\unit(u) \hat\iota(\rho,\zeta') \Tbd\rho{\lzr}\unit(S)u^*
+\Tbd{\zeta'}{\lzr}\unit(v) \hat\iota(\sigma,\zeta') \Tbd\sigma{\lzr}\unit(S)v^*\\
&=\Tbd{\zeta'}{\lzr}\unit(u) \hat\iota(\rho,\zeta') u^*\Tbd\gamma{\lzr}\unit(S)
+\Tbd{\zeta'}{\lzr}\unit(v) \hat\iota(\sigma,\zeta')v^*\Tbd\gamma{\lzr}\unit(S)\\
&=\hat\iota(\gamma,\zeta')\Tbd\gamma{\lzr}\unit(S)
=\hat\iota(\gamma,\zeta')\lmk\unit_\gamma\otimes S\rmk.
\end{split}
\end{align}
Furthermore, for any $\zeta,\zeta'\in \Obul$,
\begin{align}
\begin{split}
&\lmk\unit_\zeta\otimes \hat\iota(\gamma,\zeta')\rmk
\lmk\hat\iota(\gamma,\zeta)\otimes\unit_{\zeta'}\rmk\\
&=\Tbd{\zeta}{\lzr}\unit
\lmk
\Tbd{\zeta'}{\lzr}\unit(u)\hat\iota(\rho,\zeta') u^*
+\Tbd{\zeta'}{\lzr}\unit(v)\hat\iota(\sigma,\zeta') v^*
\rmk
\lmk
\Tbd\zeta{\lzr}\unit(u)\hat\iota(\rho,\zeta) u^*
+\Tbd\zeta{\lzr}\unit(v)\hat\iota(\sigma,\zeta) v^*
\rmk\\
&=
\Tbd{\zeta}{\lzr}\unit\lmk \Tbd{\zeta'}{\lzr}\unit(u)\hat\iota(\rho,\zeta') u^* u\rmk\cdot \hat\iota(\rho,\zeta) u^*
+\Tbd{\zeta}{\lzr}\unit\lmk \Tbd{\zeta'}{\lzr}\unit(v)\hat\iota(\sigma,\zeta') v^* v\rmk\cdot \hat\iota(\sigma,\zeta) v^*\\
&=
\Tbd{\zeta}{\lzr}\unit\Tbd{\zeta'}{\lzr}\unit(u)\cdot
\Tbd{\zeta}{\lzr}\unit \lmk \hat\iota(\rho,\zeta')\rmk\cdot \hat\iota(\rho,\zeta) u^*
+\Tbd{\zeta}{\lzr}\unit\Tbd{\zeta'}{\lzr}\unit(v)\cdot
\Tbd{\zeta}{\lzr}\unit\lmk \hat\iota(\sigma,\zeta'\rmk) \cdot \hat\iota(\sigma,\zeta) v^*\\
&=
\Tbd{\zeta}{\lzr}\unit\Tbd{\zeta'}{\lzr}\unit(u)\cdot
\lmk
\unit_{\zeta}\otimes \hat\iota(\rho,\zeta')
\rmk
\lmk  \hat\iota(\rho,\zeta) \otimes \unit_{\zeta'}\rmk\cdot u^*
+\Tbd{\zeta}{\lzr}\unit\Tbd{\zeta'}{\lzr}\unit(v)\cdot
\lmk
\unit_\zeta\otimes \hat\iota(\sigma,\zeta')
\rmk
\lmk
\hat\iota(\sigma,\zeta)\otimes \unit_{\zeta'}
\rmk v^*\\
&=\Tbd{\zeta\otimes\zeta'}{\lzr}\unit (u)
\hat\iota\lmk \rho : \zeta\otimes\zeta'\rmk u^*
+\Tbd{\zeta\otimes\zeta'}{\lzr}\unit(v)
\hat\iota\lmk\sigma:\zeta\otimes\zeta'\rmk v^*\\
&=\hat\iota\lmk \gamma: \zeta\otimes\zeta'\rmk.
\end{split}
\end{align}
To check the asymptotic behavior (\ref{shimane}),
let $\lm {r2}\in\Crbd$ with $\lz\leftarrow_r\lm {r2}$, $\lm {r2}\subset \lzr$.
Note that $u,v\in \pbd\lmk\caA_{\lhucz}\rmk'\cap\fbd$.
Therefore, Lemma \ref{lem35} implies 
\begin{align}
\begin{split}
\lim_{t\to \infty}\sup_{\sigma\in O^{rU}_{\mathop{\mathrm{bd},\lm {r2}(t)}}}
\lV \Tbd\zeta{\lzr}\unit(u)-u\rV=0,\quad
\lim_{t\to \infty}\sup_{\sigma\in O^{rU}_{\mathop{\mathrm{bd},\lm {r2}(t)}}}
\lV \Tbd\zeta{\lzr}\unit(v)-v\rV=0.
\end{split}
\end{align}
Combining this with 
\begin{align}
\begin{split}
\lim_{t\to\infty}
\sup_{\zeta\in O^{rU}_{\mathop{\mathrm{bd},\lm {r2}(t)}}}
\lV
\hat\iota(\rho :\zeta)-\unit
\rV=
\lim_{t\to\infty}
\sup_{\zeta\in O^{rU}_{\mathop{\mathrm{bd},\lm {r2}(t)}}}
\lV
\hat\iota(\sigma :\zeta)-\unit
\rV
=
0,
\end{split}
\end{align}we obtain
\begin{align}
\begin{split}
\lim_{t\to\infty}
\sup_{\zeta\in O^{rU}_{\mathop{\mathrm{bd},\lm {r2}(t)}}}
\lV
\hat\iota(\gamma :\zeta)-\unit
\rV=0.
\end{split}
\end{align}
Hence we conclude
$(\gamma,\hat\iota(\gamma : -))$ is an object of $\caZ\lmk\Cabul\rmk$,
with a morphism $u$ from $(\rho,\hat\iota(\rho : -))$ to $(\gamma,\hat\iota(\gamma : -))$
and a morphism $v$ from $(\sigma,\hat\iota(\sigma : -))$ to $(\gamma,\hat\iota(\gamma : -))$
such that $uu^*+vv^*=\unit$.

Next we consider subobjects.
Let $(\rho,\hat\iota(\rho : -))$ be an object of $\caZ\lmk\Cabul\rmk$
and $p$ a projection in the endomorphism of $(\rho,\hat\iota(\rho : -))$.
Let $\gamma\in\Obul$ and $v\in (\gamma,\rho)_U$ isometry
such that $vv^*=p$.
We then set $\hat\iota\lmk\gamma : -\rmk$ as
\begin{align}
\hat\iota\lmk\gamma : \zeta\rmk:=
\Tbd\zeta{\lzr}\unit(v^*)\hat\iota(\rho :\zeta) v.
\end{align}
This is in fact a half-braiding.
Using the naturality of $\hat\iota(\rho :\zeta)$, we have
\begin{align}
\begin{split}
&\hat\iota(\gamma:\zeta)\hat\iota(\gamma:\zeta)^*
=\Tbd\zeta{\lzr}\unit(v^*)\hat\iota(\rho :\zeta) v
v^*\hat\iota(\rho :\zeta)^*\Tbd\zeta{\lzr}\unit(v)\\
&=\Tbd\zeta{\lzr}\unit(v^*)\hat\iota(\rho :\zeta) p \hat\iota(\rho :\zeta)^*\Tbd\zeta{\lzr}\unit(v)
=\Tbd\zeta{\lzr}\unit(v^*)\Tbd\zeta\lzr\unit (p)\hat\iota(\rho :\zeta) \hat\iota(\rho :\zeta)^*\Tbd\zeta{\lzr}\unit(v)=\unit
\end{split}
\end{align}
and
\begin{align}
\begin{split}
&\hat\iota(\gamma:\zeta)^*\hat\iota(\gamma:\zeta)
=v^*\hat\iota(\rho :\zeta)^*\Tbd\zeta{\lzr}\unit(v)\Tbd\zeta{\lzr}\unit(v^*)\hat\iota(\rho :\zeta) v\\
&=v^* \hat\iota(\rho :\zeta)^*\Tbd\zeta{\lzr}\unit(p)\hat\iota(\rho :\zeta) v
=v^* \hat\iota(\rho :\zeta)^*\hat\iota(\rho :\zeta) pv=\unit.
\end{split}
\end{align}
Hence we have $\hat\iota\lmk
\gamma : \zeta
\rmk\in\caU(\fbd)$.
For any $A\in\abd$,
we have
\begin{align}
\begin{split}
&\hat\iota\lmk\gamma : \zeta\rmk\lmk\gamma\otimes \zeta\rmk(A)
=\Tbd\zeta{\lzr}\unit(v^*)\hat\iota(\rho :\zeta) v\Tbd\gamma\lzr\unit\Tbd\zeta{\lzr}\unit\pbd(A)\\
&=\Tbd\zeta{\lzr}\unit(v^*)\hat\iota(\rho :\zeta) \Tbd\rho\lzr\unit\Tbd\zeta{\lzr}\unit\pbd(A) v
=\Tbd\zeta{\lzr}\unit(v^*)\Tbd\zeta{\lzr}\unit \Tbd\rho\lzr\unit\pbd(A)\hat\iota(\rho :\zeta)  v\\
&=\Tbd\zeta{\lzr}\unit \lmk \Tbd\gamma\lzr\unit\pbd(A) v^*\rmk\hat\iota(\rho :\zeta)  v\\
&=(\zeta\otimes\gamma)(A) \hat\iota\lmk\gamma : \zeta\rmk.
\end{split}
\end{align}
For any $\zeta,\zeta'\in \Obul$, we have
\begin{align}
\begin{split}
&\lmk\unit_\zeta\otimes\hat\iota(\gamma:\zeta')\rmk
\lmk \hat\iota(\gamma: \zeta)\otimes\unit_{\zeta'}\rmk
=\Tbd\zeta\lzr\unit\lmk  \Tbd{\zeta'}{\lzr}\unit(v^*)\hat\iota(\rho :\zeta') v\rmk
\Tbd\zeta{\lzr}\unit(v^*)\hat\iota(\rho :\zeta) v\\
&=\Tbd\zeta\lzr\unit\lmk  \Tbd{\zeta'}{\lzr}\unit(v^*)\hat\iota(\rho :\zeta') p\rmk
=\Tbd\zeta\lzr\unit\lmk  \Tbd{\zeta'}{\lzr}\unit(v^* p)\hat\iota(\rho :\zeta') \rmk
\hat\iota(\rho :\zeta) v\\
&
=\Tbd{\zeta\otimes\zeta'}\lzr\unit(v^*)
\hat\iota\lmk\rho:\zeta\otimes \zeta'\rmk v
=\hat\iota(\rho:\zeta\otimes\zeta').
\end{split}
\end{align}
Furthermore, as in the direct sum case,
we have $\lim_{t\to\infty}
\sup_{\zeta\in O^{rU}_{\mathop{\mathrm{bd},\lm {r2}(t)}}}
\lV
\hat\iota(\gamma :\zeta)-\unit
\rV=0$.

Setting 
\begin{align}
\begin{split}
\epsilon_{\zc}\lmk
(\rho,\hat\iota(\rho:-)), (\sigma,\hat\iota(\sigma:-))
\rmk
:=\hat\iota(\rho:\sigma),
\end{split}
\end{align}
from (\ref{kouchi}) and (\ref{yamaguchi}),
this defines a braiding of $\caZ(\Cabul)$.
\begin{thm}\label{lem87}
Let $\llz\in\pc$.
Consider the setting in subsection \ref{setting2}, and assume Assumption \ref{assump7},
Assumption \ref{aichi}, Assumption \ref{wakayama}, Assumption \ref{assum80}, Assumption \ref{assum80l}.
Let $\hfc : \CaUbkl\to \caZ\lmk\Cabul\rmk$ be a functor
defined by
\begin{align}
\begin{split}
\hfc(\rho):=
\lmk
\ffc_0(\rho), \hb{ \ffc_0(\rho)} {-}
\rmk
\end{split}
\end{align}
for $\rho\in \OUbkl$,
and
\begin{align}
\begin{split}
\hfc(R):=\ffc_1(R)=\tau(R)\in
\lmk\hfc(\rho),\hfc(\sigma)\rmk
\end{split}
\end{align}
for any $\rho,\sigma\in \OUbkl$ and $R\in (\rho,\sigma)$.
\end{thm}
\begin{proof}
In fact $\rho\in \OUbkl$,
$\hfc(\rho)$ is an object of $\zc$because of
Lemma \ref{lem17}, \ref{lem41},\ref{lem42}, \ref{lem43}, \ref{lem84}, Proposition \ref{lem40}.
For any $\rho,\sigma\in \OUbkl$ and $R\in (\rho,\sigma)$,
$\hfc(R)=\tau(R)\in (\ffc_0(\rho),\ffc_0(\sigma))_U$
is in the morphism of $\zc$ between objects
$\lmk
\ffc_0(\rho), \hb{ \ffc_0(\rho)} {-}
\rmk $ 
and 
$\lmk
\ffc_0(\sigma), \hb{ \ffc_0(\sigma)} {-}
\rmk $.
Clearly $\hfc$ is a functor which is linear on morphisms.
Note that 
\begin{align}
\begin{split}
\unit : (\pbd,\hat\iota(\pbd: -))\to 
&\hfc(\pbk)=\lmk
\ffc_0(\pbk),
\hb{\ffc_0(\pbk)}{-}
\rmk\\
&=\lmk\pbd,\hat\iota(\pbd :-)\rmk
\end{split}
\end{align}
is an isomorphism.
Furthermore, by Lemma \ref{lem74},
\begin{align}
\begin{split}
&\hfc(\rho)\otimes\hfc(\sigma)
=\lmk \ffc_0(\rho)\otimes\ffc_0(\sigma), 
\hi{\ffc_0(\rho)\otimes\ffc_0(\sigma)}{-}
\rmk\\
&=\lmk \ffc_0(\rho\otimes\sigma), 
\hi{\ffc_0(\rho\otimes \sigma)}{-}
\rmk
= \hfc(\rho\otimes\sigma),
\end{split}
\end{align}
and its endomorphisms are $(\ffc_0(\rho)\otimes \ffc(\sigma),\ffc(\rho\otimes\sigma))_U$
and $\unit : \hfc(\rho)\otimes\hfc(\sigma)\to \hfc(\rho\otimes\sigma)$
is natural, from the proof of Lemma \ref{lem74}.
Hence $\hfc$ is a unitary tensor functor.

Next we show that $\hfc$ is equivalence.
In order to do so, it suffices to show that $\hfc$ is fully faithful and surjective.
First we show that $\hfc$ is fully faithful  i.e.,
for any $\rho,\sigma\in\OUbkl$, 
\begin{align}
\begin{split}
\hfc :(\rho,\sigma)\ni R\mapsto \tau(R)\in \lmk \hfc(\rho),\hfc(\sigma)\rmk
=\lmk\ffc_0(\rho), \ffc_0(\sigma)\rmk_U 
\end{split}
\end{align}
is an isomorphism.
Note that it is injective and linear.
To show that it is also surjective, let $\tilde R$
be an arbitrary element from $\lmk\ffc_0(\rho), \ffc_0(\sigma)\rmk_U \subset\fbd$.
Then $\tau^{-1}\lmk \tilde R\rmk$
satisfies 
\begin{align}
\begin{split}
&\tau^{-1}\lmk \tilde R\rmk\rho(A)
=\tau^{-1}\lmk \tilde R \tau\Tbk\rho\lz\unit\tau^{-1}\pbd(A)  \rmk 
=\tau^{-1}\lmk \tilde R\lmk \ffc_0(\rho)\pbd(A)\rmk  \rmk \\
&=\tau^{-1}\lmk\lmk\ffc_0(\sigma)\pbd(A)  \rmk\tilde R \rmk
=\Tbk\sigma\lz\unit\pbk(A)\tau^{-1}(\tilde R)
=\sigma(A)\tau^{-1}(\tilde R),
\end{split}
\end{align}
for any $A\in \abd$.
Hence $\tau^{-1}(\tilde R)\in (\rho,\sigma)$, and we have
\begin{align}
\hfc\lmk \tau^{-1}(\tilde R)\rmk
=\tau\tau^{-1}(\tilde R))=\tilde R.
\end{align}
This proves the surjectivity.

Next we show that $\hfc$ is surjective.
Let  $(\rho,\hi \rho -)$ be an arbitrary object in $\zc$.
By Lemma \ref{lem86}, there exists 
there exists a $G(\rho)\in \OUbk$ such that
$\ffc_0\lmk {G(\rho)}\rmk =\rho$.
Then we have 
$
\hfc(G(\rho))
=\lmk
\rho,\hb\rho -
\rmk
$.

Finally, $\hfc$ is braided because from Lemma \ref{lem47},
\begin{align}
\begin{split}
&\epsilon_{\zc}\lmk\hfc(\rho), \hfc(\sigma)\rmk\\
&=
\epsilon\lmk
\lmk \ffc_0(\rho),\hi{\ffc_0(\rho)}\rmk), \lmk\ffc_0(\sigma),\hi{\ffc_0(\sigma}-\rmk
\rmk\\
&=\ffc_1\lmk \epsilon_+(\rho,\sigma)\rmk
=\hfc\lmk \epsilon_+(\rho,\sigma)\rmk.
\end{split}
\end{align}

\end{proof}
}

\section{Bulk-boundary correspondence}\label{drinfeld}
In this section, we construct a tensor category $\caM$
by adding to $\Cabul$ morphisms from $\bl$.
By taking the idempotent completion of $\caM$,
we obtain a $C^*$-tensor category $\tilde \caM$.
We introduce {\it the Drinfeld center with an asymptotic constraint}
$\caZ_a(\tilde \caM)$ of $\tilde \caM$.
This is the same as the usual Drinfeld center\cite{Kassel}, except that the half-braidings
are required to satisfy the asymptotic property we saw in Lemma \ref{neko}.

Note that in Corollary \ref{kanazawa}, we did not assume
the gap condition on the boundary.
In this section we introduce a condition that we expect to be related to the gap condition:
the  
Assumption \ref{raichi}.
It can be interpreted as a kind of non-triviality of the braiding. 
Under this condition,
$\caZ_a(\tilde \caM)$ gets a braided $C^*$-tensor category that is 
equivalent to the bulk category $\Cabk$.
\begin{defn}
Consider the setting in subsection \ref{setting2}.
Let $\llz\in\pc$.
Let $\caM$ be a category with object $\Obj\caM:=\Obul$, 
and morphisms between objects $\rho,\sigma\in \Obj\caM:=\Obul$
defined by
 $\Mor_{\caM}(\rho,\sigma):=(\rho,\sigma)_l$.
Then $\caM$ is a strict tensor category with respect to the tensor $\otimes$
of $\Cabul$.
\end{defn}
That  $\caM$ is a strict tensor category
 follows from the fact that $\Cabul$ is closed under tensor
 and $\Cabul$ is a subcategory of $\Carbdl$.
 The tensor category $\caM$ is not necessarily a $C^*$-tensor category \change{in the sense that there may not be subobjects}.
 By taking the idempotent completion, we obtain a $C^*$-tensor category
 $\tilde\caM$.
 \change{Idempotent completion is a well-known procedure. See, for example the Appendix of 
 \cite{longo1996theory}.
 We include the proof for the reader's convenience.
 }
 \begin{prop}\label{mtcat}
 Let $\llz\in\pc$.
Consider the setting in subsection \ref{setting2}, and 
assume
Assumption \ref{wakayama}.
Let $\tilde \caM$ be a category with objects
\begin{align}
\begin{split}
\Obj\tilde\caM:=
\left\{
(\rho,p)\mid
\rho\in \Obj\caM=\Obul,\quad p\in (\rho,\rho)_l\;\; \text{projection}
\right\},
\end{split}
\end{align}
and
morphisms 
\begin{align}
\begin{split}
\Mor_{\tilde\caM}\lmk (\rho,p), (\sigma,q)\rmk:=
\left\{
R\in (\rho,\sigma)_l\mid
qR=R=Rp
\right\}\quad
\text{for} \quad (\rho,p),(\sigma,q)\in \Obj\tilde\caM.
\end{split}
\end{align}
The category $\tilde\caM$ is a $C^*$-category with
identity
\begin{align}\label{suika}
\begin{split}
\id_{(\rho,p)}:=p,\quad \text{for each}\quad (\rho,p)\in \Obj\tilde\caM
\end{split}
\end{align}
composition of morphisms 
\begin{align}\label{kaki}
\begin{split}
 \Mor_{\tilde\caM}\lmk (\sigma,q), (\gamma,r)\rmk\times
 \Mor_{\tilde\caM}\lmk (\rho,p), (\sigma,q)\rmk\ni (S,R)
 \mapsto  SR\in  \Mor_{\tilde\caM}\lmk (\rho,p), (\gamma,r)\rmk,
\end{split}
\end{align}
and an antilinear contravariant functor $* : \tilde\caM\to\tilde\caM$
given by the adjoint $*$ of $\caB(\hbd)$
for morphisms
\begin{align}\label{melon}
\begin{split}
\Mor_{\tilde\caM}\lmk (\rho,p), (\sigma,q)\rmk\ni T\to
T^*\in \Mor_{\tilde\caM}\lmk (\sigma,q),(\rho,p)\rmk.
\end{split}
\end{align}
Furthermore, $\tilde\caM$ is a $C^*$-tensor category
with the tensor
\begin{align}\label{ringo}
\begin{split}
(\rho,p)\otimes_{\tilde M}(\sigma,q):=
(\rho\otimes\sigma,p\otimes q),\quad
\text{for}\quad(\rho,p),(\sigma,q)\in \Obj\tilde\caM
\end{split}
\end{align}
and
\begin{align}\label{nashi}
\begin{split}
&R\otimes_{\tilde M}S
:=R\otimes S\in \Mor_{\tilde\caM}\lmk (\rho,p)\otimes_{\tilde M}(\sigma,q), 
(\rho',p')\otimes_{\tilde M}(\sigma',q')\rmk\\
&
\text{for}\quad
R\in \Mor_{\tilde\caM}\lmk (\rho,p),  (\rho',p')\rmk,\quad
 S\in \Mor_{\tilde\caM}\lmk (\sigma,q), (\sigma',q')\rmk.
\end{split}
\end{align}
The tensor unit is $(\pbd,\unit)$.
 \end{prop}
 \begin{proof}
 Clearly, each $\Mor_{\tilde\caM}\lmk (\rho,p), (\sigma,q)\rmk$ is a Banach space
 with the norm inherited from $\caB(\hbd)$.
 In (\ref{kaki}), we have $SR\in  \Mor_{\tilde\caM}\lmk (\rho,p), (\gamma,r)\rmk$
 because $SR\in (\rho,\gamma)_l$ and $rSR=SR=SRp$.
 The map (\ref{kaki}) is bilinear and we have $\lV SR\rV\le \lV S\rV \lV R\rV$.
 The associativity of the composition follows from that of $\caB(\hbd)$.
 In (\ref{melon}), we have $T^*\in \Mor_{\tilde\caM}\lmk (\sigma,q),(\rho,p)\rmk$
 because $T^*\in (\sigma,\rho)_l$ and $pT^*=T^*=T^*q$.
 That $T^{**}=T$, $\lV T^*T\rV=\lV T\rV^2$,
 $T^*T\ge 0$ follows from the corresponding properties of $\caB(\hbd)$.
 The formula (\ref{suika}) gives the identity morphism :
 clearly, for each object $(\rho,p)\in \Obj\tilde\caM$, we have
 $\id_{(\rho,p)}:=p\in \Mor_{\tilde\caM}\lmk (\rho,p),  (\rho,p)\rmk$
 because $p\in (\rho,\rho)_l$ and $pp=p=pp$.
 Furthermore, for any
 $R\in \Mor_{\tilde\caM}\lmk (\rho,p), (\sigma,q)\rmk$
 we have $R\circ\id_{(\rho,p)}=Rp=R=qR=\id_{(\sigma,q)}\circ R$.
 Hence, $\tilde\caM$ is a $C^*$-category.
 
 The formula  in (\ref{ringo})
 gives an object of $\tilde\caM$  because
 $\rho\otimes\sigma\in \Obj\caM=\Obul$
 and $p\otimes q\in (\rho\otimes\sigma,\rho\otimes\sigma)_l$
 with
 \begin{align}
 \begin{split}
 (p\otimes q)^*=p^*\otimes q^*=p\otimes q,\quad
 (p\otimes q)^2=p^2\otimes q^2=p\otimes q,
 \end{split}
 \end{align}
 by Lemma \ref{lem24}.
 The formula in (\ref{nashi}) is well-defined because
 $R\otimes_{\tilde M}S
:=R\otimes S\in (\rho\otimes\sigma,\rho'\otimes \sigma')_l$
and
\begin{align}
\begin{split}
(p'\otimes q')(R\otimes S)
=p'R\otimes q'S=R\otimes S
=Rp\otimes Sq=(R\otimes S)(p\otimes q),
\end{split}
\end{align}
 by Lemma \ref{lem24}.
If we further have
$R'\in \Mor_{\tilde\caM}\lmk (\rho',p'),  (\rho'',p'')\rmk$,
$S'\in \Mor_{\tilde\caM}\lmk (\sigma',q'), (\sigma'',q'')\rmk$,
then we have
\begin{align}\label{aka}
\begin{split}
&\lmk R'\otimes_{\tilde M}S'\rmk\circ\lmk R\otimes_{\tilde M}S\rmk
=\lmk R'\otimes S'\rmk\lmk R\otimes S\rmk
=R'R\otimes S'S
=R'R\otimes_{\tilde M} S'S,\\
&\lmk R\otimes_{\tilde\caM} S\rmk^*
=\lmk R\otimes S\rmk^*=R^*\otimes S^*=R^*\otimes_{\tilde\caM} S^*,
\end{split}
\end{align}
 by Lemma \ref{lem24}.
We also have
\begin{align}\label{ao}
\begin{split}
\id_{(\rho,p)\otimes_{\tilde M}(\sigma,q)}=\id_{(\rho\otimes\sigma,p\otimes q)}
=p\otimes q=
p\otimes_{\tilde \caM} q=\id_{(\rho,p)}\otimes_{\tilde \caM}\id_{(\sigma,\rho)}.
\end{split}
\end{align}
As the associativity morphism, we take
\begin{align}
\begin{split}
&\alpha_{(\rho,p), (\sigma,q),(\gamma,r)}
:=(p\otimes q)\otimes r\\
&\in 
\Mor_{\tilde \caM} \lmk
\lmk (\rho,p)\otimes_{\tilde\caM}(\sigma,q)\rmk\otimes_{\tilde\caM}
(\gamma,r),
(\rho,p)\otimes_{\tilde\caM}
\lmk (\sigma,q)\otimes_{\tilde\caM}
(\gamma,r)\rmk\rmk\\
&\quad=
\Mor_{\tilde \caM} \lmk
\lmk \rho\otimes\sigma\otimes\gamma, p\otimes q\otimes r\rmk, \lmk \rho\otimes\sigma\otimes\gamma, p\otimes q\otimes r\rmk
\rmk.
\end{split}
\end{align}
In fact, we have 
\begin{align}
\begin{split}
\alpha_{(\rho,p), (\sigma,q),(\gamma,r)}
:=p\otimes q\otimes r
\in (\rho\otimes\sigma\otimes\gamma, \rho\otimes\sigma\otimes\gamma)_l
\end{split}
\end{align}
and 
\begin{align}
\begin{split}
&\lmk p\otimes q\otimes r\rmk \alpha_{(\rho,p), (\sigma,q),(\gamma,r)}
=\lmk p\otimes q\otimes r\rmk^2
=\lmk p\otimes q\otimes r\rmk=\alpha_{(\rho,p), (\sigma,q),(\gamma,r)},\\
&\alpha_{(\rho,p), (\sigma,q),(\gamma,r)}\lmk p\otimes q\otimes r\rmk
=\lmk p\otimes q\otimes r\rmk^2
=\lmk p\otimes q\otimes r\rmk=\alpha_{(\rho,p), (\sigma,q),(\gamma,r)}.
\end{split}
\end{align}
 In fact it is an isomorphism because $\id_{\lmk \rho\otimes\sigma\otimes\gamma, p\otimes q\otimes r\rmk}=p\otimes q\otimes r$.
 The naturality of $\alpha$ and the Pentagon axiom follows from
 the definition and Lemma \ref{lem24}.
 The object $(\pbd,\unit)$ is the tensor unit with
 left/right unit constraints
 \begin{align}
 \begin{split}
 &l_{(\rho,p)}:=p \in
 \Mor_{\tilde\caM}\lmk (\pbd,\unit)\otimes_{\tilde\caM} (\rho,p),
  (\rho,p)\rmk= \Mor_{\tilde\caM}\lmk(\rho, p),
  (\rho,p)\rmk
  ,\\
 &r_{(\rho,p)}:=p \in\Mor_{\tilde\caM}\lmk 
 (\rho,p) \otimes_{\tilde\caM} (\pbd,\unit), (\rho, p)\rmk= \Mor_{\tilde\caM}\lmk(\rho, p),
  (\rho,p)\rmk.
 \end{split}
 \end{align}
 Their naturalities follow from the definition of morphisms and identity morphisms.The triangle axiom
 follows by
 \begin{align}
 \begin{split}
 \lmk \id_{(\rho,p)}\otimes_{\tilde\caM}l_{(\sigma,q)}\rmk
 \alpha_{(p,\rho), (\pbd,\unit),(\sigma,q)}
=\lmk p\otimes q\rmk\lmk p\otimes\unit\otimes q\rmk
=\lmk p\otimes q\rmk
=r_{(\rho,p)}\otimes_{\tilde\caM} \id_{(\sigma,q)}.
 \end{split}
 \end{align}
 Now we construct a direct sum for 
 $(\rho,p), (\sigma,q)\in \Obj\tilde\caM$.
 By Lemma \ref {lem61}, there exist $\gamma\in \Obj\caM=\Obul$
 and isometries $u\in (\rho,\gamma)_U$, $v\in (\sigma,\gamma)_U$
 such that $uu^*+vv^*=\unit$.
 We set $r:=upu^*+vqv^*$, which is a projection in $(\gamma,\gamma)_l$.
 Hence we have $(\gamma, r)\in \Obj\tilde\caM$.
%
 We have $up\in \Mor_{\tilde\caM} \lmk (\rho, p), (\gamma,r)\rmk$
 because $up\in (\rho,\gamma)_l$ and $rup=up=upp$.
 It is an isometry because $pu^*up=p=\id_{(\rho,p)}$.
Similarly, we have  $vq\in \Mor_{\tilde\caM} \lmk (\sigma, q), (\gamma,r)\rmk$
and it is an isometry.
We also have $upu^*+vqv^*
=r=\id_{(\gamma,r)}$.
Hence we obtain the direct sum.

Next we introduce a subobject of $(\rho,p)\in \Obj\tilde\caM$
with respect to a projection $q\in \mm\lmk (\rho,p), (\rho,p)\rmk$.
We have $(\rho,q)\in \om$ because $\rho\in \Obj(\caM)$,
$q\in (\rho,\rho)_l$ a projection.
Because  $q\in \mm\lmk (\rho,p), (\rho,p)\rmk$, we have
$pq=q=qp$.
Note also that $q\in \mm\lmk (\rho,q), (\rho,p)\rmk$
because $p q=q=qq$.
It is an isometry because $q^*q=q=\id_{(\rho,q)}$.
We also have $qq^*=q$. Hence $(\rho,q)\in \om$ and 
 $q\in \mm\lmk (\rho,q), (\rho,p)\rmk$
gives the subobject.
Finally, we have $\mm\lmk (\pbd,\unit), (\pbd,\unit)\rmk=(\pbd,\pbd)_l=\bbC\unit$
and complete the proof.
 \end{proof}
 \change{We regard this $\tilde \caM$ our boundary theory.}
 \change{
 \begin{rem}\label{module}
 Note that $\Cabul$ acts on $\tilde \caM$ by
 $(\rho,\unit)\otimes_{\tilde \caM}$.
 By Corollary \ref{kanazawa}, we may identify  $\Cabul$
 and the bulk theory $\Cabkl$.
 Hence the boundary theory $\tilde \caM$ is a left module category
 over the bulk theory $\Cabkl$.
 
 Let $(\rho, p)\in \om$ with a non-trivial projection $p\in\bl$.
 Suppose that no $\gamma\in \Obul$ can ``detect'' the projection $p\in \bl$, i.e.,
 \begin{align}\label{sora}
 \Tbd{\gamma}\lzr\unit(p)=p,\quad \text{for all}\quad \gamma\in \Obul.
 \end{align}
 Then for any $\gamma \in \Obul$, we have
 \begin{align}
 \begin{split}
 (\gamma,\unit)\otimes_{\tilde \caM} (\rho, p)
 =(\gamma\otimes\rho, \unit_\gamma\otimes p)
 =(\gamma\otimes\rho,  p).
 \end{split}
 \end{align}
 Hence, the space of $\om$ given by elements of the form
 $(-,p)\in\om$ is invariant under the action of $\Cabkl$.
 If stable gapped boundary is given by an irreducible left module \cite{kitaev2012models}, this should not happen.
 From this point of view, we expect that for a stable gapped boundary,
 (\ref{sora}) should not happen. Namely, 
 any non-trivial projection $p\in \bl$ should be detected by some anyons in $\Obul$.
 
 \end{rem}
 }
 \change{
 \begin{rem}Physically,
 considering the idempotent completion means 
 including all the condensation descendants of $\caM$.
 See section 3.3 of \cite{kong2020classification}.
 More operator algebraically, $(\rho,p)\in\om$ is a descendant of
 $\rho\in \caM$ in the sense that
 \[
 \tilde\rho(A):=\rho(A) p,\quad A\in \abd
 \]
 is a sub-representation of $\rho$.
 \end{rem}
 
 }

Next we introduce the half-braiding with an asymptotic constraint on $\tilde\caM$.
Note that the only difference between the following and the usual definition 
of half-braiding is the asymptotic constraint (iv).
Recall that this asymptotic property (iv)
is observed in Lemma \ref{neko}.
\begin{defn}\label{saru}
Let $\llz\in\pc$.
Consider the setting in subsection \ref{setting2}, and 
assume 
Assumption \ref{wakayama}.
For each $(\rho, p)\in \om$, we say
$C_{(\rho,p)}$ is a half-braiding for $(\rho,p)$
 with asymptotic constraint if it 
is a map $C_{(\rho,p)}: \om\to \Mor(\tilde\caM)$
satisfying the following.
\begin{description}
\item[(i)]
For any $(\sigma,q)\in \om$, 
$C_{(\rho,p)}\lmk (\sigma,q)\rmk\in \mm\lmk (\sigma,q)\omt(\rho,p),(\rho,p)\omt(\sigma,q)\rmk$.
\item[(ii)]
For any $(\sigma,q),(\sigma',q')\in \om$ and $R\in \mm\lmk (\sigma,q),(\sigma',q')\rmk$, we have
\begin{align}
\begin{split}
C_{(\rho,p)}\lmk (\sigma',q')\rmk \lmk R\omt \unit_{(\rho,p)}\rmk
=\lmk \unit_{(\rho,p)}\omt R\rmk C_{(\rho,p)}\lmk (\sigma,q)\rmk.
\end{split}
\end{align}
\item[(iii)]
For any $(\sigma,q),( \gamma,r)\in \om$, we have
\begin{align}
\begin{split}
C_{(\rho,p)}\lmk(\sigma,q)\omt(\gamma,r)\rmk
=\lmk C_{(\rho,p)}\lmk (\sigma,q)\rmk \omt\id_{(\gamma,r)}\rmk
\lmk
\id_{(\sigma,q)}\omt C_{(\rho,p)}\lmk (\gamma,r)\rmk
\rmk.
\end{split}
\end{align}
\item[(iv)]
For any 
$\lm 1,\lm 3\in \CUbk$ with $\lm 1\leftarrow_r \lm 3$, $\lm 1,\lm3\subset\lz$
and $\Vrl\rho{\lm3}\in \Vbu\rho{\lm 3}$
we have
\begin{align}\label{kasa}
\begin{split}
\lim_{t\to\infty}\sup_{\sigma\in \Obun{\lm 1(t)}
(\sigma,q)\in \om}
\lV
C_{(\rho,p)}\lmk (\sigma,q)\rmk-
p \Tbd\rho{\lzr}\unit (q)\cdot \Vrl\rho{\lm3}^* \Tbd\sigma{\lzr}\unit\lmk \Vrl\rho{\lm3}\rmk
q\Tbd\sigma \lzr\unit (p) 
\rV=0.
\end{split}
\end{align}
\change{We call this condition the asymptotic constraint.}
\item[(v)]
For each $(\sigma,q)\in \om$, the morphism
$C_{(\rho,p)}\lmk (\sigma,q)\rmk$
in (i) is an isomorphism.
\end{description}
\end{defn}
\change{
\begin{rem}
Let us consider the physical meaning of the asymptotic constraint (iv) for the case $p=\unit$.
Let $\lm 1,\lm 3\in \CUbk$ with $\lm 1\leftarrow_r \lm 3$, $\lm 1,\lm3\subset\lz$
(see Figure \ref{317})
and $\Vrl\rho{\lm3}\in \Vbu\rho{\lm 3}$, and set
\begin{align}
\begin{split}
\tilde\rho(A):=\Ad \Vrl\rho{\lm3}\rho(A),\quad A\in \abd.
\end{split}
\end{align}
By this definition, we have
\begin{align}
\begin{split}
\tilde\rho\in \Obun{\lm3}\subset \Obul,
\end{split}
\end{align}
and $(\tilde\rho,\unit)\in \om$.
For this object $(\tilde\rho,\unit)\in \om$,
because we may set $ \Vrl{\tilde\rho}{\lm3}=\unit$, 
the asymptotic constraint (iv) reduces to
\begin{align}\label{megane}
\begin{split}
\lim_{t\to\infty}\sup_{\sigma\in \Obun{\lm 1(t)}
(\sigma,q)\in \om}
\lV
C_{(\tilde \rho,\unit)}\lmk (\sigma,q)\rmk-
(\unit_{\tilde\rho}\otimes q)(q\otimes \unit_{\tilde\rho})
\rV=0.
\end{split}
\end{align}
Let us regard half-braiding $C_{(\tilde \rho,\unit)}(\sigma,q)$
as the crossing between $\tilde\rho$ and $\sigma$,
when we move the particle $\tilde \rho$ to right \cite{konginvitation}.
Then, because $\tilde \rho\in \Obun{\lm3}$ is already at the right of $\sigma\in \Obun{\lm 1(t)}$
for $t$ large enough
(Figure \ref{317}), there is no crossing, when we move $\tilde \rho$ to the right.
This is what (\ref{megane}) tells us.
(Note that $\unit_{\tilde\rho}\otimes q$, $q\otimes \unit_{\tilde\rho}$
are identity morphisms of $(\tilde \rho,\unit)\otimes_{\tilde\caM} (\sigma,q)$, 
$(\sigma,q)\otimes_{\tilde\caM} (\tilde \rho,\unit)$ respectively.
In particular, if $q=\unit$, then $(\unit_{\tilde\rho}\otimes q)(q\otimes \unit_{\tilde\rho})=\unit_{\caH}$.)
If we suppose that $C_{(\rho,\unit)}$ is natural with respect to
$\rho$ for morphisms in $\fbk$, then we have
\begin{align}
 \Vrl{\rho}{\lm3} C_{(\rho,\unit)}\lmk(\sigma,q)\rmk
 =C_{(\tilde \rho,\unit)}\lmk(\sigma,q)\rmk \Tbd\sigma{\lzr}\unit\lmk \Vrl\rho{\lm3}\rmk.
\end{align}
Substituting this to (\ref{megane}), we get (\ref{kasa}) for $p=\unit$.
The general (iv) is a truncation of this.
\end{rem}
}

\begin{lem}\label{panda}
Let $\llz\in\pc$.
Consider the setting in subsection \ref{setting2}, and 
assume Assumption \ref{assum80}, Assumption \ref{assum80l},
Assumption \ref{wakayama}.
For each $(\rho,p), (\sigma,q)\in \om$, set
\begin{align}\label{cup}
\begin{split}
\ti \sigma q\rho p :=\lmk p\otimes q\rmk\cdot \hb \sigma \rho \cdot \lmk q\otimes p\rmk.
\end{split}
\end{align}
Then for each $(\rho,p)\in \om$, $\ti --\rho p$ satisfies
 properties (i), (ii), (iv) of Definition \ref{saru}.
 Conversely, if $C_{(\rho,p)} : \om\to \caB(\hbd)$
 satisfies  (i), (ii), (iv) of Definition \ref{saru} for $(\rho,p)\in \om$,
 then we have $C_{(\rho,p)}=\ti --\rho p$.
\end{lem}
\begin{proof}
(i) By the definition, we have $\ti\sigma q \rho p \in \lmk \sigma\otimes \rho,
\rho\otimes\sigma\rmk_l$,
and 
\begin{align}
\begin{split}
\lmk p\otimes q\rmk \cdot \ti\sigma q \rho p 
=\ti\sigma q \rho p=\ti\sigma q \rho p \cdot \lmk q\otimes p\rmk.
\end{split}
\end{align}
(ii)
For any $(\sigma,q),(\sigma',q')\in \om$ and $R\in \mm\lmk (\sigma,q),(\sigma',q')\rmk$, we have
\begin{align}
\begin{split}
&\ti {\sigma'}{q'}{\rho}p\lmk R\omt \unit_{(\rho,p)}\rmk
=(p\otimes q') \hb{\sigma'}{\rho} (q'\otimes p)(R\otimes p)\\
&=(p\otimes q')\hb{\sigma'}{\rho} (R\otimes \unit_\rho)(q\otimes p)
=(p\otimes q')\lmk\unit_\rho\otimes R \rmk
\hb{\sigma}{\rho} (q\otimes p)\\
&=\lmk \unit_{(\rho,p)}\omt R\rmk \ti \sigma q \rho p,
\end{split}
\end{align}
by Lemma \ref{inu}.\\
(iv) For any 
$\lm 1,\lm 3\in \CUbk$ with $\lm 1\leftarrow_r \lm 3$, $\lm 1,\lm3\subset\lz$
and $\Vrl\rho{\lm3}\in \Vbu\rho{\lm 3}$
we have
\begin{align}
\begin{split}
&\lim_{t\to\infty}\sup_{\sigma\in \Obun{\lm 1(t)}
(\sigma,q)\in \om}
\lV\ti \sigma q \rho p
-
p \Tbd\rho{\lzr}\unit (q)\cdot \Vrl\rho{\lm3}^* \Tbd\sigma{\lzr}\unit\lmk \Vrl\rho{\lm3}\rmk
q\Tbd\sigma \lzr\unit (p) 
\rV\\
&=\lim_{t\to\infty}\sup_{\sigma\in \Obun{\lm 1(t)}
(\sigma,q)\in \om}
\lV
(p\otimes q)
\hb\sigma \rho(q\otimes p)
-(p\otimes q)\cdot \Vrl\rho{\lm3}^* \Tbd\sigma{\lzr}\unit\lmk \Vrl\rho{\lm3}\rmk
(q\otimes p)
\rV\\
&=0,
\end{split}
\end{align}
by Lemma \ref{neko}.\\
Suppose now that  $C_{(\rho,p)} : \om\to \caB(\hbd)$
 satisfies  (i), (ii), (iv) of Definition \ref{saru} for $(\rho,p)\in \om$.
 Consider $\lm 1,\lm 3$ in (iv).
For each $t\ge 0$ and $(\sigma,q)\in \om$,
we fix $\Vrl\sigma{\lm 1(t)}\in \Vbu\sigma{\lm 1(t)}$ and set
\begin{align}
\begin{split}
\sigma(t):=\Ad\Vrl\sigma{\lm 1(t)}\circ\sigma\in \Obun{\lm 1(t)}\subset \Obul=\Obj\caM,\quad
q(t):=\Ad\Vrl\sigma{\lm 1(t)}(q).
\end{split}
\end{align}
Then we have $(\sigma(t),q(t))\in \om$, and
$q(t) \Vrl\sigma{\lm 1(t)} q\in \mm\lmk(\sigma,q), (\sigma(t), q(t)) \rmk$.
Because $C_{(\rho,p)}$ satisfies (ii) of Definition \ref{saru},
we have
\begin{align}
\begin{split}
C_{(\rho,p)}\lmk (\sigma(t),q(t)) \rmk
\lmk q(t) \Vrl\sigma{\lm 1(t)} q\omt\id_{(\rho,p)}\rmk 
=\lmk \id_{(\rho,p)}\omt  q(t) \Vrl\sigma{\lm 1(t)} q\rmk
C_{(\rho,p)}\lmk (\sigma,q) \rmk.
\end{split}
\end{align}
Multiplying $\lmk \id_{(\rho,p)}\omt  q{\Vrl\sigma{\lm 1(t)}}^* q(t)\rmk$
from the left, we have
\begin{align}
\begin{split}
\lmk \id_{(\rho,p)}\omt  q{\Vrl\sigma{\lm 1(t)}}^* q(t)\rmk
C_{(\rho,p)}\lmk (\sigma(t),q(t)) \rmk
\lmk q(t) \Vrl\sigma{\lm 1(t)} q\omt\id_{(\rho,p)}\rmk 
=\lmk p\omt  q\rmk
C_{(\rho,p)}\lmk (\sigma,q) \rmk=C_{(\rho,p)}\lmk (\sigma,q) \rmk.
\end{split}
\end{align}
In the last equation we used (i) of Definition \ref{saru}.
Taking $t\to\infty$ limit and using (iv) of Definition \ref{saru},
we get
\begin{align}\label{budou}
\begin{split}
&C_{(\rho,p)}\lmk (\sigma,q) \rmk\\
&=\lim_{t\to\infty} 
\lmk \id_{(\rho,p)}\omt  q{\Vrl\sigma{\lm 1(t)}}^* q(t)\rmk
p \Tbd\rho{\lzr}\unit (q(t))\cdot \Vrl\rho{\lm3}^* \Tbd{\sigma(t)}{\lzr}\unit\lmk \Vrl\rho{\lm3}\rmk
q(t)\Tbd{\sigma(t)} \lzr\unit (p) 
\lmk q(t) \Vrl\sigma{\lm 1(t)} q\omt\id_{(\rho,p)}\rmk.
\end{split}
\end{align}
In particular, because  $\ti --\rho p$ satisfies
 properties (i), (ii), (iv) of Definition \ref{saru},
 it is given by the right hand side of (\ref{budou}).
 Therefore, we get $C_{(\rho,p)}=\ti --\rho p$.
\end{proof}
\change{
We checked that 
$\ti \sigma q\rho p$ in (\ref{cup}) satisfies (i), (ii), (iv) of  Definition \ref{saru}.
For $\ti \sigma q\rho p$ to satisfy (v) of  Definition \ref{saru},
$p$ need to satisfy some condition.
}
\begin{lem}\label{mikan}
Let $\llz\in\pc$.
Consider the setting in subsection \ref{setting2}, and 
assume Assumption \ref{assum80}, Assumption \ref{assum80l},
Assumption \ref{wakayama}.
Then we have the following
\begin{enumerate}
\item If $(\rho,p)\in \om$ satisfies $p\in \fbd$, then
 $\ti --\rho p$ satisfies (iii) and (v) of Definition \ref{saru}
 and is a half-braiding for $(\rho,p)$
 with asymptotic constraint.
\item Let $(\rho,p)\in \om$.
Suppose that
$\ti --\rho p$ satisfies (v) of Definition \ref{saru}.
 Let $\lm 1,\lm 3\in \CUbk$ with $\lm 1\leftarrow_r \lm 3$, $\lm 1,\lm3\subset\lz$
and $\Vrl\rho{\lm3}\in \Vbu\rho{\lm 3}$.
Set $\tilde p:=\Vrl\rho{\lm3} p{\Vrl\rho{\lm3}}^*$.
Then we have
\begin{align}
\begin{split}
\lim_{t\to\infty}\sup_{\sigma\in \Obun{\lm 1(t)}}
\lV \tilde p -\lmk \Tbd\sigma\lzr\unit(\tilde p)\rmk\rV=0.
\end{split}
\end{align}
\end{enumerate}

\end{lem}
\begin{proof}
{\it 1.} 
Let $(\rho,p)\in \om$.
Then for any $(\sigma,q), (\gamma,r)\in\om$, we have
\begin{align}\label{midori}
\begin{split}
&\lmk \ti \sigma q \rho p \omt\id_{(\gamma,r)}\rmk
\lmk 
\id_{(\sigma,q)}\omt\ti\gamma r \rho p
\rmk\\
&=
\lmk
\lmk (p\otimes q)\hb\sigma\rho(q\otimes p)\rmk\otimes r
\rmk
\lmk
q\otimes \lmk (p\otimes r)\hb \gamma\rho (r\otimes p)\rmk
\rmk\\
&=
\lmk p\otimes q\otimes r\rmk
\lmk
 \lmk \hb\sigma\rho\otimes \unit_\gamma\rmk  
 \lmk \unit_\sigma\otimes p\otimes \unit_\gamma\rmk
 \lmk
\unit_\sigma \otimes \hb \gamma\rho \rmk
\rmk
 \lmk q\otimes r\otimes p\rmk
\end{split}
\end{align}
 If $p\in \fbd$, 
using Lemma \ref{lem42} and Lemma \ref{panama},
we have
\begin{align}
\begin{split}
&\lmk \ti \sigma q \rho p \omt\id_{(\gamma,r)}\rmk
\lmk 
\id_{(\sigma,q)}\omt\ti\gamma r \rho p
\rmk=(\ref{midori})\\
&=\lmk p\otimes q\otimes r\rmk
\lmk
 \lmk \hb\sigma\rho\otimes \unit_\gamma\rmk  
 \lmk
\unit_\sigma \otimes \hb \gamma\rho \rmk
\rmk
\lmk q\otimes r\otimes p\rmk
\\
&=\lmk p\otimes q\otimes r\rmk
\hb {\sigma\otimes \gamma}{\rho}
\lmk q\otimes r\otimes p\rmk
=
\tilde\iota \lmk  (\sigma,q)\omt(\gamma,r), (\rho,p)\rmk.
\end{split}
\end{align}
Hence (iii) holds. Similarly, if $p\in\fbd$,
(v) can be checked by
\begin{align}
\begin{split}
\ti\sigma q \rho p
=\lmk p\otimes q\rmk
\hb\sigma \rho\lmk q\otimes p\rmk
=\lmk p\otimes q\rmk
\hb\sigma \rho
=
\hb\sigma \rho\lmk q\otimes p\rmk.
\end{split}
\end{align}
Here, we used 
Lemma \ref{lem42} and Lemma \ref{inu}.
\\
{\it 2.} 
Consider the setting of {\it 2}. 
\change{If (v) holds for $\ti--\rho p$,
then 
$\lmk \ti \sigma \unit \rho p\rmk^* \ti \sigma \unit \rho p$
should be the identity of $(\sigma,\unit)\omt(\rho,p)$, which is
$\Tbd\sigma\lzr\unit(p)$, for any $\sigma\in\Obj\caM$.
Similarly, 
$\ti \sigma \unit \rho p \lmk \ti \sigma \unit \rho p\rmk^* $
should be the identity of $(\rho,p) \omt (\sigma,\unit)$, which is
$p$, for any $\sigma\in\Obj\caM$.
Therefore,}
for any $\sigma\in\Obj\caM$, we have
\begin{align}\label{chairo}
\begin{split}
&\Tbd\sigma\lzr\unit(p)=
\lmk \ti \sigma \unit \rho p\rmk^* \ti \sigma \unit \rho p\\
&=\Tbd\sigma\lzr\unit(p)\cdot \lmk \hb\sigma\rho\rmk^* \cdot p\cdot \hb\sigma\rho\cdot \Tbd\sigma\lzr\unit(p),\\
& p=\ti \sigma \unit \rho p\cdot \lmk \ti \sigma \unit \rho p\rmk^*
=p\cdot \hb\sigma\rho\cdot \Tbd\sigma\lzr\unit(p)\cdot \lmk \hb\sigma\rho\rmk^* p.
\end{split}
\end{align}

Then from Lemma \ref{neko}, we have
\begin{align}
\begin{split}
&\lim_{t\to\infty}\sup_{\sigma\in \Obun{\lm 1(t)}}
\lV
\Tbd\sigma\lzr\unit(\tilde p)\tilde p \Tbd\sigma\lzr\unit(\tilde p)-\Tbd\sigma\lzr\unit(\tilde p)
\rV\\
&=\lim_{t\to\infty}\sup_{\sigma\in \Obun{\lm 1(t)}}
\lV
\Tbd\sigma\lzr\unit(p) \Tbd\sigma\lzr\unit(\Vrl\rho{\lm 3}^*) \Vrl\rho{\lm 3} p \Vrl\rho{\lm 3}^*
\Tbd\sigma\lzr\unit( \Vrl\rho{\lm 3})\Tbd\sigma\lzr\unit(p)-\Tbd\sigma\lzr\unit(p)
\rV=0,
\end{split}
\end{align}
and 
 \begin{align}
 \begin{split}
& \lim_{t\to\infty}\sup_{\sigma\in \Obun{\lm 1(t)}}
 \lV
 \tilde p\Tbd\sigma\lzr\unit(\tilde p)\tilde p-\tilde p
 \rV\\
 &=\lim_{t\to\infty}\sup_{\sigma\in \Obun{\lm 1(t)}}
 \lV
 p \Vrl\rho{\lm 3}^* \Tbd\sigma\lzr\unit(\Vrl\rho{\lm 3}) \Tbd\sigma\lzr\unit(p) \Tbd\sigma\lzr\unit(\Vrl\rho{\lm 3}^*)
 \Vrl\rho{\lm 3} p-p
 \rV=0.
 \end{split}
 \end{align}
 Therefore, we have 
 \begin{align}
 \begin{split}
  &\lim_{t\to\infty}\sup_{\sigma\in \Obun{\lm 1(t)}}
  \lV   \tilde p\Tbd\sigma\lzr\unit(\tilde p)- \Tbd\sigma\lzr\unit(\tilde p)\rV
  =\lim_{t\to\infty}\sup_{\sigma\in \Obun{\lm 1(t)}}
  \lV   \Tbd\sigma\lzr\unit(\tilde p)\lmk \tilde p-\unit\rmk\Tbd\sigma\lzr\unit(\tilde p)\rV^{\frac 12}=0,\\
   &\lim_{t\to\infty}\sup_{\sigma\in \Obun{\lm 1(t)}} 
   \lV   \tilde p\Tbd\sigma\lzr\unit(\tilde p)-  \tilde p\rV
   =\lim_{t\to\infty}\sup_{\sigma\in \Obun{\lm 1(t)}} 
 \lV   \Tbd\sigma\lzr\unit(\tilde p) \tilde p-  \tilde p\rV\\
&  =\lim_{t\to\infty}\sup_{\sigma\in \Obun{\lm 1(t)}} 
  \lV
  \tilde p
  \lmk \Tbd\sigma\lzr\unit(\tilde p)-\unit\rmk
  \tilde p
  \rV^{\frac 12}=0.
 \end{split}
 \end{align}
 Hence we obtain
 \begin{align}
 \begin{split}
 \lim_{t\to\infty}\sup_{\sigma\in \Obun{\lm 1(t)}}
  \lV   \tilde p- \Tbd\sigma\lzr\unit(\tilde p)\rV=0.
 \end{split}
 \end{align}

 \end{proof}
Next we introduce the Drinfeld center $\ozam$ with the asymptotic constraints.
First we prepare sets corresponding to the objects and morphisms.
\begin{defn}\label{hato}
Let $\llz\in\pc$.
Consider the setting in subsection \ref{setting2}, and 
assume 
Assumption \ref{wakayama}.
We denote by $\ozam$, the set of all pairs $\lmk (\rho,p), C_{(\rho,p)}\rmk$
where 
$(\rho,p)$ is an object of $\tilde\caM$ and
$C_{(\rho,p)}$ is a half-braiding for $(\rho,p)$ with the asymptotic constraint.
For $\lmk (\rho,p), C_{(\rho,p)}\rmk, \lmk (\sigma,q), C_{(\sigma,q)}\rmk\in \ozam$, we denote by
$\mzam\lmk \lmk (\rho,p), C_{(\rho,p)}\rmk, \lmk (\sigma,q), C_{(\sigma,q)}\rmk\rmk $
the set of all $R\in \mm\lmk (\rho,p), (\sigma,q)\rmk$
satisfying
\begin{align}
\begin{split}
\lmk R\omt \unit_{(\gamma,r)}\rmk C_{(\rho,p)}\lmk(\gamma,r)\rmk
=C_{(\sigma,q)}\lmk(\gamma,r)\rmk\lmk \unit_{(\gamma,r)}\omt R\rmk
\end{split}
\end{align}
for all $(\gamma,r)\in \om$.
\end{defn}
Note that this definition is identical to the standard one  \cite{Kassel} except for the asymptotic constraints.
From Lemma \ref{lem42}, Lemma \ref{panda},
we have the following
\begin{lem}\label{drian}
Let $\llz\in\pc$.
Consider the setting in subsection \ref{setting2}, and 
assume Assumption \ref{assum80}, Assumption \ref{assum80l},
Assumption \ref{wakayama}.
Then for any 
$\lmk (\rho,p), C_{(\rho,p)}\rmk, \lmk (\sigma,q), C_{(\sigma,q)}\rmk\in \ozam$,
 we have
\begin{align}
\begin{split}
\mm\lmk (\rho,p), (\sigma,q)\rmk\cap\fbd
\subset 
\mzam\lmk \lmk (\rho,p), C_{(\rho,p)}\rmk, \lmk (\sigma,q), C_{(\sigma,q)}\rmk\rmk
\end{split}
\end{align}
\end{lem}
Now, we introduce an additional assumption.
\begin{assum}\label{raichi}
Let $\llz\in\pc$.
Consider the setting in subsection \ref{setting2}.
For any
$\lm 1,\lm 3\in \CUbk$ with $\lm 1\leftarrow_r \lm 3$, $\lm 1,\lm3\subset\lz$,
we assume the following to hold:
If  $R\in (\rho,\sigma)_l$ for some $\rho,\sigma\in \Obun{\lm 3}$
and 
\begin{align}\label{mizuiro}
\begin{split}
\lim_{t\to\infty}\sup_{\gamma\in \Obun{\lm 1(t)}}
\lV R-\Tbd\gamma{\lzr}\unit(R)
\rV=0,
\end{split}
\end{align}
then 
$R\in \fbd$.
\end{assum}
\begin{rem}\label{dil}
For $\rho,\sigma,R$  and 
$\gamma\in \Obun{\lm 1(t)}$ in Assumption \ref{raichi},
 from Lemma \ref{inu},
we have
\begin{align}
\begin{split}
\hb\sigma\gamma R
=\hb\sigma\gamma\lmk R\otimes\unit_\gamma\rmk
=\lmk \unit_\gamma\otimes R \rmk\hb\rho\gamma
=\Tbd\gamma\lzr\unit(R)\hb\rho\gamma.
\end{split}
\end{align}
From this, we have
\begin{align}
\begin{split}
\lV \hb\sigma\gamma R-R\hb\rho\gamma\rV=
\lV R-\Tbd\gamma{\lzr}\unit(R)
\rV.
\end{split}
\end{align}
In particular, if $R$ is a unitary, 
(\ref{mizuiro}) holds if and only if $\hb\sigma-$
and $\hb\rho-$ are asymptotically unitarily equivalent.
From this point of view, Assumption \ref{raichi} says that 
quasi-particles which have the asymptotically unitarily equivalent 
half-braiding should be equivalent with respect to $\simeq_U$.
Under the assumptions of Corollary \ref{kanazawa}, using the equivalence from Corollary \ref{kanazawa},
it means two quasi-particles in the bulk which have the asymptotically unitarily equivalent braiding
  are equivalent. 
  This is a kind of a non-degeneracy condition of the braiding.
\end{rem}
\change{
\begin{rem}\label{gapmean}
Note that in general, $R\in (\rho,\sigma)_l$ above belongs to
$\pbd(\caA_{(\ld_3)^c})'$ but not necessarily to $\pbd(\caA_{\ld_3})''$.
As was seen for Toric code (Example \ref{toric}), there can be an element
$R\in \pbd(\caA_{(\ld_3)^c})'\setminus \pbd(\caA_{\ld_3})''$ given by a weak limit of string operators 
on the paths $\zeta_n$ depicted in Figure \ref{610},
where $a_n\to-\infty$ as $n\to\infty$.
Recall in Toric code (Example \ref{toric}), $\pi_\gamma^X$ can be created from the vacuum by
a unitary constructed in this manner.
Assumption \ref{raichi} tells us we can detect the fact that
$R$ does not belong to $\pbd(\caA_{\ld_3})''$, 
(i.e., corresponding to a {\it deconfined} particle) by a quasi-particle $\gamma$
sitting at the cone $\ld_1(t)$.

Let us consider the case that $\rho=\sigma$ and $R$ is a projection.
Recall Remark \ref{module}. There, we physically concluded that at a stable gapped boundary,
any nontrivial projection in $\bl$ should be detected by some anyons in $\Obul$.
Suppose our projection $R$ is in $\pbd(\caA_{(\ld_3)^c})'\setminus \pbd(\caA_{\ld_3})''$ and constructed 
as above as a weak limit of operators defined on the string $\zeta_n$.
Suppose, nevertheless, (\ref{mizuiro}) holds.
This (roughly) means that our projection
$R$ cannot be detected by anyons, because crossing at $\ld_1(t)$ does not change $p$.
From Remark \ref{module}, the existence of such projection physically means that
the boundary is not a stable gapped boundary.

For this reason, we would like to regard Assumption \ref{raichi}
as a stable gapped boundary condition.
\end{rem}
}
\change{
\begin{rem}
In Toric code (Example \ref{toric} ), Proposition 12.7 of \cite{wa}
and its proof proves Assumption \ref{raichi}.
\end{rem}
}
\begin{figure}[htbp]
    \includegraphics[width=5cm]{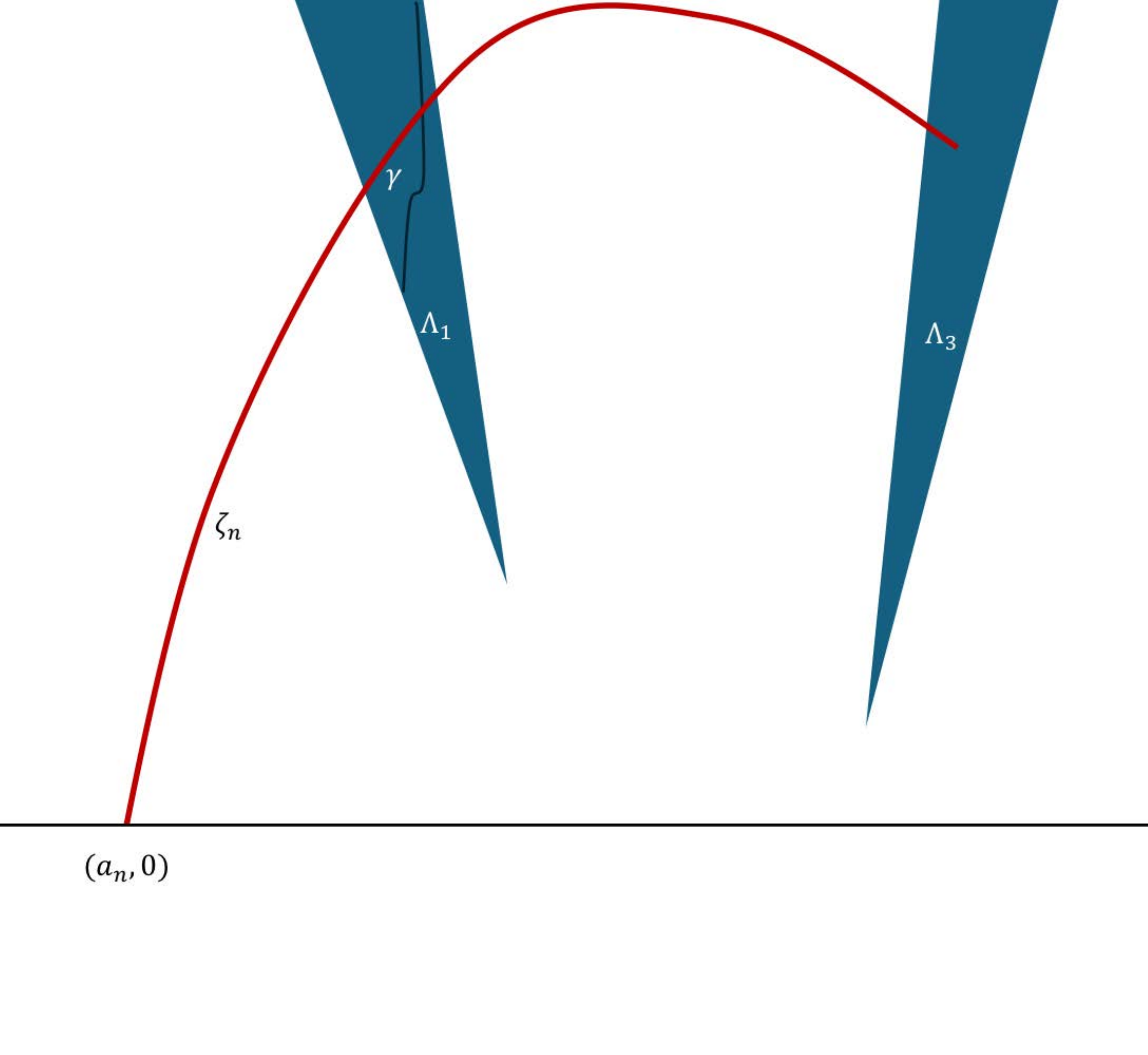}\\
\caption{}   \label{610}
   \end{figure}

Under this Assumption \ref{raichi}, we note that $p$ need to be in $\fbd$
for $(\rho,p)$ to have half-braiding with the asymptotic constraint:
\begin{lem}\label{zabon}
Let $\llz\in\pc$.
Consider the setting in subsection \ref{setting2}, and 
assume Assumption \ref{assum80}, Assumption \ref{assum80l},
Assumption \ref{wakayama}, Assumption \ref{raichi}.
Then 
\begin{enumerate}
\item$p$ belongs to $\fbd$
for any $\lmk (\rho,p), C_{(\rho,p)}\rmk\in \ozam$ and
\item 
for any
$\lmk (\rho,p), C_{(\rho,p)}\rmk, \lmk (\sigma,q), C_{(\sigma,q)}\rmk\in \ozam$,
\begin{align}
\begin{split}
\mm\lmk (\rho,p), (\sigma,q)\rmk\cap\fbd
= 
\mzam\lmk \lmk (\rho,p), C_{(\rho,p)}\rmk, \lmk (\sigma,q), C_{(\sigma,q)}\rmk\rmk.
\end{split}
\end{align}
\end{enumerate}
\end{lem}
\begin{proof}
{\it 1.} For any $\lmk (\rho,p), C_{(\rho,p)}\rmk\in \ozam$,
we have $C_{(\rho,p)}=\ti--\rho p$ by Lemma \ref{panda} satisfying
(v) of Definition \ref{saru}.
Then by Lemma \ref{mikan} {\it 2.}, $\tilde p:=\Vrl\rho{\lm3} p{\Vrl\rho{\lm3}}^*\in (\tilde\rho,\tilde\rho)_l$,
with $\tilde\rho:=\Ad (\Vrl\rho{\lm3})\rho\in \Obun{\lm 3}$ in 
Lemma \ref{mikan} {\it 2.}
satisfy the condition of $R$ in Assumption \ref{raichi} (\ref{mizuiro}).
Hence from Assumption \ref{raichi} we have $\tilde p\in\fbd$ and $p\in \fbd$.\\
{\it 2.} We already know the inclusion $\subset$ from Lemma \ref{drian}.
To see the opposite inclusion, assume that
$R\in \mzam\lmk \lmk (\rho,p), C_{(\rho,p)}\rmk, \lmk (\sigma,q), C_{(\sigma,q)}\rmk\rmk$.
Then by the definition of the morphisms, and Lemma \ref{panda}, we have
\begin{align}
\begin{split}
\lmk R\omt \unit_{(\gamma,r)}\rmk \ti\gamma r \rho p
=
\ti \gamma r \sigma q
\lmk \unit_{(\gamma,r)}\omt R\rmk
\end{split}
\end{align}
for all $(\gamma,r)\in \om$.
In particular, considering $r=\unit$ case, we have
\begin{align}\label{pine}
\begin{split}
\lmk R\otimes \unit_\gamma \rmk\hb
\gamma \rho\lmk \unit_\gamma\otimes p\rmk
-(q\otimes \unit_\gamma)\hb\gamma \sigma (\unit_\gamma\otimes R)=0
\end{split}
\end{align}
for all $\gamma\in \Obj\caM$.
Let $\lm 1,\lm 3\in \CUbk$ be as in Assumption \ref{raichi} and fix
$\Vrl\sigma{\lm 3}\in \Vbu\sigma{\lm 3}$, $\Vrl\rho{\lm 3}\in \Vbu\rho{\lm 3}$, and
set $\tilde R:=\Vrl\sigma{\lm 3} R{\Vrl\rho{\lm 3}}^*$,
$\tilde p:=\Ad\lmk \Vrl\rho{\lm 3}\rmk(p)\in \pbd(\caA_{\lm 3^c})'$,
$\tilde q:=\Ad\lmk \Vrl\sigma{\lm 3}\rmk(q)\in \pbd(\caA_{\lm 3^c})'$.
Then we obtain 
\begin{align}\label{momo}
\begin{split}
&\lim_{t\to\infty}\sup_{{\gamma\in \Obun{\lm 1(t)}}}
\lV \tilde R \Tbd\gamma{\lzr}\unit(\tilde p)-\tilde q \Tbd\gamma{\lzr}\unit(\tilde R)\rV\\
&=
\lim_{t\to\infty}\sup_{{\gamma\in \Obun{\lm 1(t)}}}
\lV R
\Vrl\rho{\lm3}^* \Tbd\gamma{\lzr}\unit\lmk \Vrl\rho{\lm3}\rmk
\Tbd\gamma{\lzr}\unit(p)
-q \Vrl\sigma{\lm3}^* \Tbd\gamma{\lzr}\unit\lmk \Vrl\sigma{\lm3}\rmk
\Tbd\gamma{\lzr}\unit(R)
\rV
=0.
\end{split}
\end{align}
from Lemma \ref{neko}, using (\ref{pine}).
Because of {\it 1}, we already know that $p,q\in\fbd$.
Therefore, by the same proof as
Lemma \ref{lem82} we have
\begin{align}
\begin{split}
\lim_{t\to\infty}\sup_{{\gamma\in \Obun{\lm 1(t)}}}\lV\Tbd\gamma{\lzr}\unit(\tilde p)-\tilde p\rV,\quad
\lim_{t\to\infty}\sup_{{\gamma\in \Obun{\lm 1(t)}}}\lV \Tbd\gamma{\lzr}\unit(\tilde q)-\tilde q\rV\to 0,\quad t\to\infty.
\end{split}
\end{align}
Substituting this to (\ref{momo}), we obtain
\begin{align}
\begin{split}
\lim_{t\to\infty}\sup_{{\gamma\in \Obun{\lm 1(t)}}}
\lV \tilde R-\Tbd\gamma{\lzr}\unit(\tilde R)\rV\to 0.
\end{split}
\end{align}
Because $\tilde R\in \lmk\Ad \Vrl\rho{\lm 3}\circ\rho,
\Ad \Vrl\sigma{\lm 3}\circ\sigma
\rmk_l$
and $\Ad \Vrl\rho{\lm 3}\circ\rho, \Ad \Vrl\sigma{\lm 3}\circ\sigma
\in \Obun{\lm 3}$,
from the assumption \ref{raichi}, we get $\tilde R\in \fbd$.
This means
$R\in \fbd$.

\end{proof}

\change{Under the assumptions we consider, our $\zam$
becomes a braided $C^*$-tensor category.
The construction is pararell to the standard one \cite{Kassel}.}
\begin{prop}
Let $\llz\in\pc$.
Consider the setting in subsection \ref{setting2}, and 
assume Assumption \ref{assum80}, Assumption \ref{assum80l},
Assumption \ref{wakayama}, Assumption \ref{raichi}.
Then $\zam$ is a braided $C^*$-tensor category
with objects $\ozam$, and morphisms $\Mor\lmk \zam\rmk$
and
\begin{description}
\item[(i)] identity morphism
$\id_{\lmk (\rho,p), C_{(\rho,p)}\rmk}:=p$, for  $\lmk (\rho,p), C_{(\rho,p)}\rmk\in \ozam$,
\item[(ii)] 
composition of morphisms and an anti-linear contravariant functor $* : \zam\to\zam$
inherited from that of $\Mor(\tilde \caM)$ (\ref{kaki}), (\ref{melon}) 
\item[(iii)] the tensor product of objects
$\rpc\rho p, \rpc \sigma q\in \ozam$
\begin{align}\label{udon}
\begin{split}
\rpc\rho p\ozt\rpc\sigma q
:=\lmk
(\rho,p)\omt (\sigma,q), C_{(\rho,p)}\hat\otimes C_{(\sigma,q)}
\rmk
\end{split}
\end{align}
where
\begin{align}
\begin{split}
\lmk C_{(\rho,p)}\hat\otimes C_{(\sigma,q)}\rmk\lmk
(\gamma,r)\rmk
:=\lmk \id_{(\rho,p)}\omt C_{(\sigma,q)}\lmk(\gamma,r) \rmk\rmk\cdot
\lmk C_{(\rho,p)}\lmk(\gamma,r)\rmk\omt \id_{(\sigma,q)}\rmk,
\end{split}
\end{align}
for $(\gamma,r)\in \om$
with the
tensor unit 
$\lmk (\pbd, \unit), \ti --\pbd \unit\rmk$,
and the tensor of morphisms inherited from that 
of $\tilde\caM$, and
\item[(iv)]
the braiding
\begin{align}\label{pasta}
\begin{split}
\epsilon_{\zam}\lmk \rpc\rho p, \rpc \sigma q\rmk:=
C_{(\sigma,q)}\lmk (\rho,p)\rmk
\end{split}
\end{align}
for $\rpc \rho p, \rpc \sigma q\in \ozam$.

\end{description}
\end{prop}
\begin{proof}
Clearly, each $\mzam\lmk \lmk (\rho,p), C_{(\rho,p)}\rmk, \lmk (\sigma,q), C_{(\sigma,q)}\rmk\rmk $
is a Banach space by the definition.
Because of Lemma \ref{zabon} {\it 1}, {\it 2}, the element in (i) gives the identity morphism and the composition 
and $*$, inherited from
$\tilde\caM$ are well-defined.
For example, for $T\in \mzam\lmk\rpc\rho p, \rpc \sigma q \rmk$,
we have 
\begin{align}
\begin{split}
T^*\in\mm\lmk (\sigma, q), (\rho,p)\rmk \cap\fbd
=\mzam\lmk \rpc \sigma q, \rpc \rho p\rmk.
\end{split}
\end{align}
With these operations, $\zam$ becomes a $C^*$-category.

Next we consider the
tensor product of objects
$\rpc\rho p, \rpc \sigma q\in \zam$.
By Lemma \ref{panda}, we have
$C_{(\rho, p)}=\ti --\rho p$, $C_{(\sigma, q)}=\ti --\sigma q$.
Therefore,  any $\lmk  \gamma, r\rmk \in \om$, we have 
\begin{align}\label{somen}
\begin{split}
&\lmk C_{(\rho,p)}\hat\otimes C_{(\sigma,q)}\rmk \lmk (\gamma, r)\rmk
:=\lmk \id_{(\rho,p)}\omt \ti \gamma r \sigma q \rmk\cdot
\lmk \ti \gamma r \rho p \omt \id_{(\sigma,q)}\rmk\\
&=\lmk p\otimes q\otimes r \rmk
\lmk \unit_\rho \otimes \hb \gamma\sigma \rmk
\lmk \unit_\rho\otimes r\otimes\unit_\sigma \rmk
\lmk \hb \gamma \rho \otimes \unit_\sigma\rmk
\lmk r\otimes p \otimes q \rmk\\
&=\lmk p\otimes q\otimes r \rmk
\lmk \unit_\rho \otimes \hb \gamma\sigma \rmk
\lmk \hb \gamma \rho \otimes \unit_\sigma\rmk
\lmk r\otimes p \otimes q \rmk\\
&=\lmk p\otimes q\otimes r \rmk
 \hb \gamma{\rho\otimes \sigma} 
\lmk r\otimes p \otimes q \rmk
\\
&=\tilde\iota\lmk(\gamma, r), (\rho,p)\omt (\sigma,q)\rmk
\end{split}
\end{align}
using Lemma \ref{inu} and Lemma \ref{lem43}.
Hence we have $ C_{(\rho,p)}\hat\otimes C_{(\sigma,q)}=\tilde\iota\lmk(-, -), (\rho,p)\omt (\sigma,q)\rmk$.
By Lemma \ref{zabon} {\it 1.}, 
$p,q$ belong to $\fbd$ and so is $p\otimes q$.
Therefore,  $ C_{(\rho,p)}\hat\otimes C_{(\sigma,q)}=\tilde\iota\lmk(-, -), (\rho,p)\omt (\sigma,q)\rmk$ is
a half-braiding for $(\rho,p)\omt (\sigma,q)$ with asymptotic constraint by
Lemma \ref{mikan}.
Hence (\ref{udon}) defines an element in $\ozam$.
The tensor of morphisms inherited from that 
of $\tilde\caM$ gives a well-defined tensor product on $\mzam$:
For any $R\in \mzam\lmk \rpc \rho p, \rpc {\rho'}{p'}\rmk$ and
$S\in \mzam\lmk \rpc \sigma q\rpc{\sigma'}{q'}\rmk$ belong to $\fbd$,
we have 
\begin{align}
\begin{split}
R\ozt S:=R\omt S\in \mm\lmk (\rho,p)\omt (\sigma,q),  (\rho',p')\omt (\sigma',q')\rmk\cap \fbd.
\end{split}
\end{align}
From Lemma \ref{zabon} we have
\begin{align}
\begin{split}
&R\ozt S\\
&\in 
\mzam\lmk \rpc\rho p\ozt\rpc\sigma q,\rpc{\rho'} {p'}\ozt \rpc{\sigma'} {q'}
\rmk.
\end{split}
\end{align}
Because the morphisms and identity morphisms and tensors of morphisms
in $\zam$
are just the ones inherited from $\tilde\caM$,
properties corresponding to (\ref{aka}), (\ref{ao}) automatically holds.

The associators are also inherited from $\tilde\caM$:
\begin{align}
\begin{split}
a_{\rpc\rho p, \rpc \sigma q, \rpc \gamma r}:=
\alpha_{(\rho,p), (\sigma,q),(\gamma,r)}
=p\otimes q\otimes r.
\end{split}
\end{align}
From Lemma \ref{zabon}, this belongs to
the morphism in $\zam$ between
$$\lmk \rpc\rho p\ozt \rpc \sigma q\rmk\ozt \rpc \gamma r$$ and
$$\rpc\rho p\ozt \lmk \rpc \sigma q\ozt \rpc \gamma r\rmk.$$
This is an isomorphism, and its naturality and the
pentagon equation immediately follows from those in $\tilde \caM$ case.

Note that 
\begin{align}
\begin{split}
&\lmk (\pbd, \unit), \ti --\pbd \unit\rmk\ozt \rpc\rho p
=\rpc\rho p\\
&=\rpc\rho p \ozt \lmk (\pbd, \unit), \ti --\pbd \unit\rmk,
\end{split}
\end{align}
because of (\ref{somen}) and Lemma \ref{panda}.
Hence $\hat l_{\rpc \rho p} :=l_{(\rho,p)}=p$, 
$\hat r_{\rpc \rho p} :=r_{(\rho,p)}=p$
defines left/right constraint
using Lemma \ref{zabon} again.
Note that their naturality, as well as the triangle equality follow from
that in $\tilde\caM$.
The tensor unit is simple because $\pbd$ is irreducible.

Next we construct a direct sum for
$\rpc\rho p,\rpc \sigma q\in \ozam$.
Let $(\gamma, r)\in \om$ be the direct sum of
$(\rho,p), (\sigma, q)\in \om$, 
$up\in \mm\lmk (\rho, p), (\gamma,r)\rmk$,
$vq\in \mm\lmk (\sigma, q), (\gamma,r)\rmk$
in the proof of Proposition \ref{mtcat}.
By Lemma \ref{zabon}, we have $p,q\in\fbd$ hence
we have $r=upu^*+vqv^*\in\fbd$.
Therefore, by Lemma \ref{mikan},
$\lmk(\gamma,r),\ti --\gamma r\rmk$
is an object of $\zam$.
Because $up, vq\in \fbd$ as well,
they are morphisms in $\zam$
 from $\rpc\rho p$, $\rpc\sigma q$ to
$\lmk(\gamma,r),\ti --\gamma r\rmk$ respectively,
from Lemma \ref{zabon}.
They are isometries such that
$upu^*+vqv^*=r=\id_{(\gamma,r)}$.

To obtain a subobject of
$\rpc \rho p\in \ozam$
and a projection 
\begin{align}
\begin{split}
q\in\mzam\lmk \rpc \rho p,\rpc \rho p\rmk
=\mm\lmk(\rho,p), (\rho,p)\rmk\cap\fbd,
\end{split}
\end{align}
recall from the proof of Proposition \ref{mtcat}
that $(\rho, q)\in \om$ and
$q\in \mm\lmk(\rho,q), (\rho,p)\rmk$.
Because $q\in \fbd$,
we have
\begin{align}
\begin{split}
&\lmk (\rho, q), \ti --\rho q\rmk \in \ozam,\\
&q\in \mm\lmk (\rho,q), (\rho, p)\rmk\cap\fbd
=\mzam\lmk \lmk (\rho, q), \ti --\rho q\rmk, \rpc\rho p\rmk,
\end{split}
\end{align}
by Lemma \ref{mikan}, Lemma \ref{zabon}.
Because $q^*q=q$ and $qq^*=q$, this gives a subobject.
Hence $\zam$ is a $C^*$-tensor category.

Now we show that (\ref{pasta}) defines
a braiding on $\zam$.
For any $\rpc \rho p, \rpc \sigma q\in \ozam$,
Lemma \ref{zabon} implies $p,q\in\fbd$.
By Lemma \ref{panda} the definition (\ref{pasta})
means
\begin{align}\label{risu}
\begin{split}
&\epsilon_{\zam}\lmk \rpc\rho p, \rpc \sigma q\rmk:=
C_{(\sigma,q)}\lmk (\rho,p)\rmk=\ti\rho p \sigma q\\
&\in \mm\lmk
(\rho,p)\omt(\sigma,q),(\sigma,q)\omt (\rho,p)
\rmk\cap 
\fbd
\\
&=\mzam
\lmk
\rpc\rho p\ozt\rpc \sigma q,
\rpc\sigma q\ozt \rpc \rho p
\rmk.
\end{split}
\end{align}
Here we used $p,q\in\fbd$ and Lemma \ref{zabon}.
From this and the definition of the identity morphism
and (v) of Definition \ref{saru} for $C_{(\sigma,q)}$, it is straightforward
to see that this gives an isomorphism.
From the definition of objects and morphisms of $\zam$,
the naturality for the braiding hold.
 The definition of the tensor product (\ref{somen})
 and definition of the half-braiding Definition \ref{saru} (iii)
 guarantee the hexagon axiom.
 
 Hence $\zam$ is a braided $C^*$-tensor category.
\end{proof}
Finally, we get the following bulk-edge correspondence.
\begin{thm}\label{mainthm}
Let $\llz\in\pc$.
Consider the setting in subsection \ref{setting2}, and 
assume Assumption \ref{assum3}, Assumption \ref{assum80}, Assumption \ref{assum80l},
Assumption \ref{assump7}, Assumption \ref{aichi}, 
Assumption \ref{oo} and 
Assumption \ref{raichi}.
(Note that Assumption \ref{wakayama} automatically holds.)
Then the functor $\Psi: \Cabkl\to \zam$ given by
\begin{align}
\begin{split}
&\Psi(\rho):=
\lmk
\lmk F_0^{\llz}(\rho),\unit\rmk,
\ti --{F_0^{\llz}(\rho)}\unit
\rmk,\quad \rho\in \OUbkl,\\
&\Psi(R):=F_1^{\llz}(R),\quad
R\in (\rho,\sigma),\quad \rho,\sigma\in \OUbkl=\Obkl
\end{split}
\end{align}
is an equivalence of braided $C^*$-tensor categories.
Here $F^{\llz}$ is the functor from Theorem \ref{lem74}
$F^{\llz} : \Cabkl\to \Cabul$.
\end{thm}
\begin{proof}
First we check $\Psi$ is indeed a functor.
Clearly we have $\lmk F_0^{\llz}(\rho),\unit\rmk\in \om$ by Lemma \ref{lem17}
and
Lemma \ref{mikan} then implies $\Psi(\rho)\in \ozam$.
For any $\rho,\sigma\in \OUbkl$ and $R\in (\rho,\sigma)$,
we have $\Psi(R)\in \lmk F_0^{\llz}(\rho), F_0^{\llz}(\sigma)\rmk_U$
by Lemma \ref{lem17}, and by Lemma \ref{zabon},
$\lmk F_0^{\llz}(\rho), F_0^{\llz}(\sigma)\rmk_U$ 
is equal to $\mzam\lmk \Psi(\rho), \Psi(\sigma)\rmk$.
Hence $\Psi(R)\in\mzam\lmk \Psi(\rho), \Psi(\sigma)\rmk$.
For any object $\rho\in \OUbkl$, 
we have $\Psi\lmk\unit_\rho\rmk=\Psi(\unit)=\unit=\id_{\Psi(\rho)}$.
For any $R\in (\rho,\sigma)$, $S\in (\sigma,\gamma)$ for
$\rho,\sigma,\gamma\in \OUbkl$,
we have 
$\Psi(SR)=\Psi(S)\Psi(R)$.
Hence $\Psi$ is a functor. Clearly it is linear on the morphisms.

Next we see that $\Psi$ is a tensor functor.
First, 
\begin{align}
\begin{split}
\varphi_0:=\unit :
\lmk \lmk {\pbd}, \unit\rmk,
\ti --{\pbd}\unit\rmk
\to \Psi(\pbk)=
\lmk \lmk {\pbd}, \unit\rmk,
\ti --{\pbd}\unit\rmk
\end{split}
\end{align}\change{
is a morphism by Lemma \ref{zabon} {\it 2.}, hence an isomorphism.}
Note that
\begin{align}
\begin{split}
&\Psi(\rho)\ozt\Psi(\sigma)\\
&=\lmk
\lmk  F_0^{\llz}(\rho), \unit\rmk\omt
\lmk  F_0^{\llz}(\sigma), \unit\rmk,
\tilde\iota \lmk (-,-), \lmk  F_0^{\llz}(\rho), \unit\rmk\omt
\lmk  F_0^{\llz}(\sigma), \unit\rmk\rmk
\rmk\\
&=\lmk \lmk
\lmk F_0^{\llz}(\rho)\otimes F_0^{\llz}(\sigma),\unit  \rmk,
\tilde\iota \lmk (-,-), \lmk  F_0^{\llz}(\rho)\otimes
F_0^{\llz}(\sigma), \unit\rmk\rmk
\rmk\rmk
\\
&=\lmk
\lmk F_0^{\llz}(\rho\otimes\sigma),\unit \rmk,
\tilde\iota \lmk (--), \lmk F_0^{\llz}(\rho\otimes\sigma),\unit \rmk\rmk
\rmk=\Psi(\rho\otimes\sigma),
\end{split}
\end{align}
by Lemma \ref{lem17} and the calculation (\ref{somen}).
Hence 
$\varphi_2(\rho,\sigma):=\unit : \Psi(\rho)\ozt\Psi(\sigma)
\to \Psi(\rho\otimes\sigma)
$
is an isomorphism and it is natural
because the functor $F^{\llz}$ is \change{a tensor functor}:
\begin{align}
\begin{split}
\Psi(R\otimes S)=F_1^{\llz}(R\otimes S)=
F_1^{\llz}(R)\otimes F_1^{\llz}( S)
=\Psi(R)\ozt \Psi(S),
\end{split}
\end{align}
for $R\in (\rho,\rho')$, $S\in (\sigma,\sigma')$, $\rho,\rho',\sigma,\sigma'\in \Obkl$.
For any $\rho,\sigma,\gamma\in \Obkl$,
we have
\begin{align}
\begin{split}
\Psi\lmk \unit \rmk
&\varphi_2\lmk \rho\otimes \sigma, \gamma\rmk
\lmk \varphi_2\lmk\rho,\sigma \rmk\ozt \id_{\Psi(\gamma)}\rmk
=\unit
=\varphi_2\lmk \rho,\sigma\otimes\gamma\rmk
\lmk \id_{\Psi(\rho)}\ozt \varphi_2\lmk \sigma,\gamma\rmk\rmk
a_{\Psi(\rho), \Psi(\sigma),\Psi(\gamma)},\\
&\Psi(\unit)\varphi_2\lmk \pbk, \rho\rmk\lmk \varphi_0\ozt \id_{\Psi(\rho)}\rmk=\unit
=\hat l_{\Psi(\rho)}\\
&\Psi(\unit)\varphi_2\lmk\rho,  \pbk\rmk\lmk \id_{\Psi(\rho)} \ozt\varphi_0\rmk=\unit
=\hat r_{\Psi(\rho)}\\
\end{split}
\end{align}
Hence $\Psi$ is a tensor functor.

The tensor functor $\Psi$ is braided
because
\begin{align}
\begin{split}
&\Psi\lmk \epsilon_-^{(\lz)}(\rho:\sigma)\rmk=
F_1^{\llz}\lmk\epsilon_-^{(\lz)}(\rho:\sigma)\rmk
=\hb{F_0^{\llz}(\rho)}{F_0^{\llz}(\sigma)}\\
&=\ti{F_0^{\llz}(\rho)} \unit {F_0^{\llz}(\sigma)} \unit
=
\epsilon_{\zam}\lmk \Psi(\rho), \Psi(\sigma)\rmk,
\end{split}
\end{align}by
Lemma \ref{lem47}.

By Lemma \ref{lem17}, Lemma \ref{zabon} and the definition of $\Psi$,
$\Psi$ is fully faithful.
It is also essentially surjective.
To see this, let $\rpc \rho p \in \ozam$.
By Lemma \ref{zabon}, we have $p\in (\rho,\rho)_U$.
Therefore, by Lemma \ref{lem70},
there exists a $\gamma\in \Obul$ and an isometry
$v\in (\gamma,\rho)_U$ such that
$vv^*=p$.
For this $\gamma\in \Obul$, we get $G(\gamma)\in \Obkl$ by
Lemma \ref{lem86}.
Then we have
\begin{align}
\begin{split}
\Psi\lmk G(\gamma)\rmk
=\lmk (\gamma,\unit), \ti--\gamma \unit \rmk.
\end{split}
\end{align}
Note from Lemma \ref{zabon} that
\begin{align}
\begin{split}
v\in \mm\lmk (\gamma,\unit), (\rho,p)\rmk\cap \fbd
=\mzam\lmk \Psi\lmk G(\gamma)\rmk, \rpc \rho p\rmk.
\end{split}
\end{align}
Because
\begin{align}
\begin{split}
vv^*=p=\id_{\rpc\rho p},\quad v^*v=\unit,
\end{split}
\end{align}
it is an isomorphism.
Hence $\Psi$ is essentially surjective.
Hence $\Psi$ is an equivalence of braided $C^*$-tensor category.
\end{proof}

\section{Assumption \ref{assump7}}\label{gappedboundary}
In this section, we give some sufficient condition
for Assumption \ref{assump7} to hold.
We introduce a cone version of TQO (see \cite{bravyi2010topological} for original TQO):
\begin{defn}\label{def104}
Let $\Phi$ be a positive uniformly bounded finite range interaction on $\abk$
with a unique frustration-free ground state $\omega$.
Let $r\in\bbN$ be the interaction length of $\Phi$.
For each finite subset $S$ of $\bbZ^2$, we denote by
$S_r$ the set of all points $\bm x\in \bbZ^2$
whose Euclidean distance from $S$ are less than or equal to $r$.
We denote by $P_{S_r}$ the spectral projection of
$H_\Phi(S_r):=\sum_{X\subset S_r}\Phi(X)$ corresponding to eigenvalue $0$.
We say that $\Phi$ satisfies {\it strict cone TQO}
if for any $\theta_1, \theta_2$ with $0<\theta_1<\frac\pi 2<\theta_2<\pi$
and $-\infty<a\le b<\infty$,
there exists a number $m_{\theta_1,\theta_2,a, b}\in\bbN$
satisfying the following condition :
for any $c\ge r$ and $l\ge m_{\theta_1,\theta_2,a, b}$,
setting 
\begin{align}
\begin{split}
&S:=\lmk \lr{a}{\frac{\theta_2}2}\setminus\lr{b}{\frac{\theta_1}2}\rmk\cap \lmk \bbZ\times [c,c+l]\rmk,
\end{split}
\end{align}
we have
\begin{align}
\begin{split}
P_{S_r} AP_{S_r}=\omega(A)P_{S_r},\quad A\in \caA_S.
\end{split}
\end{align}
\end{defn}
Here is some sufficient condition for the strict cone TQO.
\begin{lem}
Let $\Phi$ be a positive uniformly bounded finite range  commuting  interaction
on $\abk$
with a unique frustration-free ground state $\omega$.
Let $r$ be the interaction length of $\Phi$.
Suppose that this $\omega$ satisfies the following property:
for any $\theta_1,\theta_2\in\bbR$ with
$0<\theta_1<\frac\pi 2<\theta_2<\pi$
and $-\infty<a<b<\infty$,
there exists an $m_{\theta_1,\theta_2,a,b}\in \bbN$
such that
\begin{align}
\begin{split}
s\lmk \left.\omega\right\vert_{\caA_{S_r}}\rmk
=P_{S_r},
\end{split}
\end{align} holds
for any $l\ge m_{\theta_1,\theta_2,a,b}$ and $c\ge r$.
Here we set $S:={\lmk \lr{a}{\frac{\theta_2}2}\setminus\lr{b}
{\frac{\theta_1}2}\rmk\cap\lmk\bbZ\times [c,c+l]\rmk}$, and 
$S_r$, $P_{S_r}$ are defined by the formulas in  Definition \ref{def104} 
with respect to our $l$ and $c$. 
The notation
$s\lmk \left.\omega\right\vert_{\caA_{S_r}}\rmk $ indicates the support projection of
the state $\left.\omega\right\vert_{\caA_{S_r}}$.
Then $\Phi$ satisfies the cone strict TQO.
\end{lem}
\begin{proof}
Let $(\caH,\pi, \Omega)$ be a GNS triple of $\omega$.
For any $\Gamma\subset\bbZ^2$, 
let $Q_\Gamma$ be an orthogonal projection onto
$\cap_{X\subset \Gamma}\ker\pi\lmk \Phi(X)\rmk$.
Note that $Q_S=\pi\lmk \Proj[H_{\Phi}(S)=0]\rmk\in\pi\lmk\caA_S\rmk$,
for finite $S\subset \bbZ^2$.
Here $\Proj[H_{\Phi}(S)=0]$ denotes the spectral
projection of $H_{\Phi}(S)$ corresponding to eigenvalue $0$.

For any $\Gamma\subset\bbZ^2$
and finite subsets $\Gamma_n\Subset\Gamma$ with $\Gamma_n\uparrow \Gamma$,
$Q_{\Gamma_n}$ is decreasing, hence there is a orthogonal projection
$\tilde Q$ such that
$Q_{\Gamma_n}\to \tilde Q$, $\sigma$-strongly.
In particular, we have $\tilde Q\in\pi(\caA_{\Gamma})''$.
Because $Q_\Gamma\le Q_{\Gamma_n}$ for any $n$,
we have $Q_\Gamma\le \tilde Q$. 
On the other hand, for any finite $X\subset \Gamma$, we have
\begin{align}
\pi(\Phi(X))\tilde Q=\pi(\Phi(X))Q_{\Gamma_n}\tilde Q=0
\end{align}
eventually.
Therefore, we have $\tilde Q\le Q_\Gamma$.
Hence we conclude $\tilde Q=Q_\Gamma$.
In particular, we have $Q_\Gamma\in\pi(\caA_{\Gamma})''$.

\change{Because our interaction is commuting,
for any finite $S_1,S_2\subset\bbZ^2$, $H_{\Phi}(S_1)$ and $H_{\Phi}(S_2)$ commute.
Therefore, we have
$[Q_{S_1}, Q_{S_2}]=0$.}
Therefore, for any $\Gamma_1,\Gamma_2\subset\bbZ^2$,
we have $[Q_{\Gamma_1}, Q_{\Gamma_2}]=0$.
In particular, $Q_{\Gamma_1}Q_{\Gamma_2}$ is an orthogonal projection.

For any finite $S\subset\bbZ^2$,
we have $Q_{S_r}Q_{S^c}=Q_{\bbZ^2}$.
To see this,
multiply $
Q_{S_r}\ge Q_{\bbZ^2}
$ by
$Q_{S^c}$ from left and right.
Then we obtain
\begin{align}
\begin{split}
Q_{S_r} Q_{S^c}=Q_{S^c} Q_{S_r} Q_{S^c}\ge Q_{S^c}Q_{\bbZ^2}Q_{S^c}=Q_{\bbZ^2}.
\end{split}
\end{align}
On the other hand, because $r$ is the interaction length of 
$\Phi$, for any finite subset $X$ of $\bbZ^2$ with $\Phi(X)\neq 0$,
we have $X\subset S_r$ or $X\subset S^c$.
Therefore, 
for any finite subset $X$ of $\bbZ^2$, we have 
$\Phi(X)Q_{S_r} Q_{S^c}=0$.
It means $Q_{S_r} Q_{S^c}\le Q_{\bbZ^2}$, showing the claim.

Because $\omega$ is the unique frustration-free state, it is pure and 
$\pi$ is irreducible. Therefore,
$Q_{\bbZ^2}$ is an orthogonal projection onto $\bbC\Omega$.
Therefore, for any finite $S\subset\bbZ^2$ and 
 for any $A\in\caA_S$, we have
\begin{align}
\begin{split}
\pi\lmk P_{S_r} A P_{S_r}\rmk\Omega
=Q_{S_r} \pi(A)\Omega
=Q_{S_r} \pi(A) Q_{S^c}\Omega
=Q_{S_r} Q_{S^c}\pi(A) \Omega
=Q_{\bbZ^2}\pi(A)\Omega
=\omega(A)\Omega.
\end{split}
\end{align}
In the third equality, we recalled $Q_{S^c}\in\pi(\caA_{S^c})''$.

Now we would like to remove $\Omega$ from the equality above
for $S$ described in the Lemma.
Let $\theta_1,\theta_2\in\bbR$ with
$0<\theta_1<\frac\pi 2<\theta_2<\pi$
and $-\infty<a<b<\infty$,
and let $m_{\theta_1,\theta_2,a,b}\in \bbN$ be the number
given by the assumption of the Lemma.
Let $l\ge m_{\theta_1,\theta_2,a,b}$ and $c\ge r$
and set 
$S:={\lmk \lr{a}{\frac{\theta_2}2}\setminus\lr{b}{\frac{\theta_1}2}\rmk\cap\lmk\bbZ\times [c,c+l]\rmk}$.
Let $\caK$ be a Hilbert space
and $U:\caH\to\caH_{S_r}\otimes \caK$ a unitary
such that
\begin{align}\label{kumamoto}
\begin{split}
&\Ad U\lmk\pi(A)\rmk=A\otimes\unit,\quad A\in \caA_{S_r},\\
&\Ad U\lmk \pi\lmk\caA_{S_r^c}\rmk''\rmk=\unit\otimes \caB(\caK).
\end{split}
\end{align}
Here we set $\caH_{S_r}:=\bigotimes_{z\in S_r}\bbC^d$.
Let 
\begin{align}
\begin{split}
U\Omega=\sum c_i\xi_i\otimes\eta_i\in \caH_{S_r}\otimes \caK,\quad c_i\neq 0
\end{split}
\end{align}
be a Schmidt decomposition.
By the assumption,
we have 
$
s\lmk \left.\omega\right\vert_{S_r}\rmk
=P_{S_r}$.
Note that
with the Schmidt decomposition above,
the density matrix of $\omega\vert_{S_r}$
is
\begin{align}
\sum_i|c_i|^2\ket{\xi_i}\bra{\xi_i}.
\end{align}
Therefore, $\{\xi_i\}$ is the CONS of $P_{S_r}\caH_{S_r}$.

Now using the Schmidt decomposition and (\ref{kumamoto}), we have
\begin{align}\begin{split}
&\sum c_i P_{S_r}A P_{S_r}\xi_i\otimes\eta_i
=\lmk P_{S_r}A P_{S_r}\otimes\unit\rmk U\Omega
=\Ad U \lmk \pi(P_{S_r} A P_{S_r})\rmk\cdot U\Omega\\
&=U\pi(P_{S_r} AP_{S_r})\Omega
=\omega(A) U\Omega
=\omega(A)\sum c_i\xi_i\otimes\eta_i,
\end{split}
\end{align}
for any $A\in\caA_S$.
Applying $\bra{\eta_i}$ on this, and using $c_i\neq 0$, we obtain
\begin{align}
\begin{split}
P_{S_r}AP_{S_r}\xi_i=\omega(A) \xi_i,
\end{split}
\end{align}
for any $A\in\caA_S$.
Because $\{\xi_i\}$ is a CONS of $P_{S_r}\caH_{S_r}$,
this means that
$P_{S_r}AP_{S_r}=\omega(A)P_{S_r}$,
for any $A\in\caA_S$.
This completes the proof. 
\end{proof}

\begin{prop}\label{hotaru}
Let $\Phi$ be a positive uniformly bounded finite range interaction on $\abk$
with a unique frustration-free state $\obk$, satisfying the strict cone TQO.
Let $r\in\bbN$ be the interaction length of $\Phi$.
Let $\Phi_{\mathrm {bd}}$ be 
a positive uniformly bounded finite range interaction on $\abd$
with a unique frustration-free state $\omega_{\mathrm {bd}}$.
Suppose there exists $r\le R\in\bbN$ such that
\begin{align}
\begin{split}
\Phi(X)=\Phi_{\mathrm {bd}}(X),\quad \text{for all}\quad X\Subset\bbZ\times[R,\infty).
\end{split}
\end{align}
Let $(\hbk,\pbk,\Omega_{\mathrm{bk}})$, $(\hbd,\pbd,\Omega_{\mathrm{bd}})$
be GNS triples of $\obk$, $\obd$ respectively.
Then for any $\ld\in \CUbk$,
$\pbk\vert_{\ld\cap\hu}$ and 
$\pbd\vert_{\ld\cap\hu}$ are quasi-equivalent.
\end{prop}
\begin{proof}
Fix  any $\ld\in \CUbk$.
Then there exist $\theta_1, \theta_2$ with $0<\theta_1<\frac\pi 2<\theta_2<\pi$
and $-\infty<a\le b<\infty$
such that $\ld \cap\hu\subset \lmk \lr{a}{\frac{\theta_2}2}\setminus\lr{b}{\frac{\theta_1}2}\rmk$.
We have
\begin{align}\label{miyazaki}
\begin{split}
\left. \obk\right\vert_{\caA_{\lmk \lr{a}{\frac{\theta_2}2}\setminus\lr{b}{\frac{\theta_1}2}\rmk\cap \lmk \bbZ\times [2R,\infty)\rmk}}
=\left. \obd\right\vert_{\caA_{\lmk \lr{a}{\frac{\theta_2}2}\setminus\lr{b}{\frac{\theta_1}2}\rmk\cap \lmk \bbZ\times [2R,\infty)\rmk}}.
\end{split}
\end{align}
In fact, for any local $A\in \caA_{\lmk \lr{a}{\frac{\theta_2}2}\setminus\lr{b}{\frac{\theta_1}2}\rmk\cap \lmk \bbZ\times [2R,\infty)\rmk}$,
there exists $l\ge m_{\theta_1\theta_1 a b}$ (given by Definition \ref{def104}),
such that
\begin{align}
A\in \caA_{\lmk \lr{a}{\frac{\theta_2}2}\setminus\lr{b}{\frac{\theta_1}2}\rmk\cap \lmk \bbZ\times [2R,2R+l]\rmk}.
\end{align}
Setting $S:={\lmk \lr{a}{\frac{\theta_2}2}\setminus\lr{b}{\frac{\theta_1}2}\rmk\cap \lmk \bbZ\times [2R,2R+l]\rmk}$,
by strict cone TQO, we have
$P_{S_r}A P_{S_r}=\obk(A)P_{S_r}$.
Because 
\begin{align}
\begin{split}
P_{S_r}=\Proj\left[
H_\Phi(S_r)=0
\right]
=\Proj\left[
H_{\Phi_{\mathrm{bd}}}(S_r)=0
\right],
\end{split}
\end{align}
we have 
\begin{align}
\begin{split}
\obd(A)=\obd\lmk P_{S_r} A P_{S_r}\rmk
=\obk(A)\obd\lmk P_{S_r}\rmk
=\obk(A).
\end{split}
\end{align}
As this holds for any local $A\in \caA_{\lmk \lr{a}{\frac{\theta_2}2}\setminus\lr{b}{\frac{\theta_1}2}\rmk\cap \lmk \bbZ\times [2R,\infty)\rmk}$,
we obtain (\ref{miyazaki}).

Set $\Gamma:={\lmk \lr{a}{\frac{\theta_2}2}\setminus\lr{b}{\frac{\theta_1}2}\rmk\cap \lmk \bbZ\times [2R,\infty)\rmk}$.
Then 
\begin{align}
\begin{split}
\caA_{\lmk \lr{a}{\frac{\theta_2}2}\setminus\lr{b}{\frac{\theta_1}2}\rmk}
=\caA_{\lmk \lr{a}{\frac{\theta_2}2}\setminus\lr{b}{\frac{\theta_1}2}\rmk\setminus \Gamma}\otimes\caA_\Gamma,
\end{split}
\end{align}
where $\lmk \lr{a}{\frac{\theta_2}2}\setminus\lr{b}{\frac{\theta_1}2}\rmk\setminus \Gamma$
is finite.
Because (\ref{miyazaki}) holds,
by \cite{BR2} Lemma 6.2.55, states
$\obd\vert_{\lmk \lr{a}{\frac{\theta_2}2}\setminus\lr{b}{\frac{\theta_1}2}\rmk}$
and $\obk\vert_{\lmk \lr{a}{\frac{\theta_2}2}\setminus\lr{b}{\frac{\theta_1}2}\rmk}$
are quasi-equivalent.

Because $\obk$, $\obd$ are pure,
$\pbk\lmk \caA_{\lmk \lr{a}{\frac{\theta_2}2}\setminus\lr{b}{\frac{\theta_1}2}\rmk}\rmk''$ and
$\pbd\lmk \caA_{\lmk \lr{a}{\frac{\theta_2}2}\setminus\lr{b}{\frac{\theta_1}2}\rmk}\rmk''$
are factors.
Hence, with GNS representations 
$\pi_{1}$, $\pi_2$
of states $\left. \obk\right\vert_{\caA_{\lmk \lr{a}{\frac{\theta_2}2}\setminus\lr{b}{\frac{\theta_1}2}\rmk}}$ and
 $\left. \obd\right\vert_{\caA_{\lmk \lr{a}{\frac{\theta_2}2}\setminus\lr{b}{\frac{\theta_1}2}\rmk}}$
 respectively 
 on $\caA_{\lmk \lr{a}{\frac{\theta_2}2}\setminus\lr{b}{\frac{\theta_1}2}\rmk}$,
 there are $*$-isomorphisms
 \begin{align}
 \begin{split}
 \tau_1: \pbk\lmk \caA_{\lmk \lr{a}{\frac{\theta_2}2}\setminus\lr{b}{\frac{\theta_1}2}\rmk}\rmk''\to
\pi_1\lmk \caA_{\lmk \lr{a}{\frac{\theta_2}2}\setminus\lr{b}{\frac{\theta_1}2}\rmk}\rmk'',\\
\tau_2: \pbd\lmk \caA_{\lmk \lr{a}{\frac{\theta_2}2}\setminus\lr{b}{\frac{\theta_1}2}\rmk}\rmk''\to
 \pi_2\lmk \caA_{\lmk \lr{a}{\frac{\theta_2}2}\setminus\lr{b}{\frac{\theta_1}2}\rmk}\rmk'',
 \end{split}
 \end{align}
 such that
 \begin{align}
 \begin{split}
 \tau_1\pbk\vert_{{\caA_{\lmk \lr{a}{\frac{\theta_2}2}\setminus\lr{b}{\frac{\theta_1}2}\rmk}}}=\pi_1,\quad
 \tau_2\pbd\vert_{{\caA_{\lmk \lr{a}{\frac{\theta_2}2}\setminus\lr{b}{\frac{\theta_1}2}\rmk}}}=\pi_2.
 \end{split}
 \end{align}
By the quasi-equivalence of 
$\obd\vert_{\caA_{\lmk \lr{a}{\frac{\theta_2}2}\setminus\lr{b}{\frac{\theta_1}2}\rmk}}$
and $\obk\vert_{\caA_{\lmk \lr{a}{\frac{\theta_2}2}\setminus\lr{b}{\frac{\theta_1}2}\rmk}}$,
there exists a $*$-isomorphism 
\begin{align}
\begin{split}
\tau_0 :
\pi_1\lmk \caA_{\lmk \lr{a}{\frac{\theta_2}2}\setminus\lr{b}{\frac{\theta_1}2}\rmk}\rmk''
\to 
\pi_2\lmk \caA_{\lmk \lr{a}{\frac{\theta_2}2}\setminus\lr{b}{\frac{\theta_1}2}\rmk}\rmk''
\end{split}
\end{align}
such that $\tau_0\pi_1=\pi_2$.

Hence we obtain a $*$-isomorphism
\begin{align}
\begin{split}
\tau_2^{-1}\tau_0\tau_1 : \pbk\lmk \caA_{\lmk \lr{a}{\frac{\theta_2}2}\setminus\lr{b}{\frac{\theta_1}2}\rmk}\rmk''
\to \pbd\lmk \caA_{\lmk \lr{a}{\frac{\theta_2}2}\setminus\lr{b}{\frac{\theta_1}2}\rmk}\rmk''
\end{split}
\end{align}
such that
\begin{align}
\begin{split}
\left.
\tau_2^{-1}\tau_0\tau_1\pbk
\right\vert_{\caA_{\lmk \lr{a}{\frac{\theta_2}2}\setminus\lr{b}{\frac{\theta_1}2}\rmk}}
=\left.\pbd\right\vert_{\caA_{\lmk \lr{a}{\frac{\theta_2}2}\setminus\lr{b}{\frac{\theta_1}2}\rmk}}.
\end{split}
\end{align}
Restricting this to $\caA_{\ld\cap\hu}$,
$\left. \pbd\right\vert_{\caA_{\ld\cap\hu}}$
is quasi-equivalent to
$\left. \pbk\right\vert_{\caA_{\ld\cap\hu}}$.
\end{proof}

{\bf Acknowledgment.}\\
The author is grateful to Dr. Daniel Ranard for asking her this question,
and to Professor Zhenghan Wang, Professor Liang Kong and Professor Yasuyuki Kawahigashi
for helpful discussions.

\appendix
\section{Basic Lemmas}\label{basic}
For $C^*$-algebras $\caA$, $\caB$ and $\varepsilon>0$, 
we write $\caA\subset_\varepsilon\caB$ if
for any $a\in \caA$, there exists a $b\in \caB$ such that $\lV b\rV\le \lV a\rV$
and $\lV a-b\rV\le \varepsilon\lV a\rV$.
For a $C^*$-algebra $\caB$, $\caU(\caB)$ denotes unitaries in $\caB$.
\begin{lem}\label{lem62}
For any $0<\varepsilon<1$, there exists $0<\delta_1(\varepsilon)<1$ satisfying the following:
For any unital $C^*$-algebra $\caA$, projections $p,q$ in $\caA$
such that $\lV q-qp \rV<\delta_1(\varepsilon)$,
there exists a partial isometry $w\in\caA$
such that $ww^*=q$, $w^*w\le p$ and $\lV w-q\rV<\varepsilon$.
\end{lem}
\change{
\begin{proof}
Lemma 2.5.2 of \cite{Lin} and its proof.
\end{proof}}
\begin{lem}\label{lem63}
For any $0<\varepsilon<1$, there exists $0<\delta_2(\varepsilon)<1$ satisfying the following:
For any unital $C^*$-algebra $\caA$, projections $p,q$ in $\caA$
such that $\lV q-p \rV<\delta_2(\varepsilon)$,
there exists a unitary $u\in \caA$ such that
$p=uqu^*$ and $\lV u-\unit\rV<\varepsilon$.
\end{lem}
\change{
\begin{proof}
Lemma  2.5.1 of \cite{Lin}.
\end{proof}
}
\begin{lem}\label{lem64}
For any $0<\varepsilon<1$, there exists $0<\delta_3(\varepsilon)<1$ satisfying the following:
For any Hilbert space $\caH$, a unital $C^*$-algebra $\caA$ on $\caH$, a projection
$p\in \caB(\caH)$, and a self-adjoint $x\in \caA$
such that $\lV p-x\rV<\delta_3(\varepsilon)$,
there exists a projection $q$ in $\caA$ such that
$\lV p-q\rV<\varepsilon$.
\end{lem}
\change{\begin{proof}
Lemma 2.5.4 of \cite{Lin}.
\end{proof}}
\section{Conditional expectations and approximations}
We frequently use the following Lemma.
\begin{lem}\label{lem12}
Let $\Gamma\subset\bbZ^2$ and $S$ a finite subset of $\Gamma$.
Let $(\caH,\pi)$ be a representation of $\caA_\Gamma$.
Let $\{e_{IJ} : I,J\in \caI\}$ be a system of matrix units of $\caA_S$ with indices $\caI$, and fix one 
$I_0\in\caI$.
Define $\caE_S:\caB(\caH)\to \caB(\caH)$ by
\begin{align}
\caE_S(x):=\sum_{I\in \caI} \pi(e_{II_0}) x\pi(e_{I_0I}),\quad x\in \caB(\caH).
\end{align}
Then $\caE_S\lmk \caB(\caH)\rmk=\pi(\caA_S)'$ and
$\caE_S:\caB(\caH)\to \pi(\caA_S)'$ is a projection of norm $1$.
Furthermore, $\caE^\Gamma_S :=\caE_S\vert_{\pi(\caA_\Gamma)''}
: \pi(\caA_\Gamma)''\to\pi(\caA_{\Gamma\setminus S})''$ is a surjection.
\end{lem}
\begin{proof}
That $\caE_S:\caB(\caH)\to \pi(\caA_S)'$ is a projection or norm $1$ is easy to check.
Because $\pi(\caA_{\Gamma\setminus S})''\subset \pi(\caA_S)'$,
we have $\pi(\caA_{\Gamma\setminus S})''=\caE_S\lmk \pi(\caA_{\Gamma\setminus S})''\rmk\subset \caE_S
\lmk\pi(\caA_{\Gamma})''\rmk$.
Any $A\in \alocg{\Gamma}$ can be written as
$A=\sum_{IJ}e_{IJ}a_{IJ}$ with some $a_{IJ}\in \caA_{\Gamma\setminus S}$.
We then have $\caE_S\lmk \pi(A)\rmk=\sum_{IJ}\caE_S(\pi(e_{IJ}))\pi(a_{IJ})=\pi(a_{I_0I_0})\in \pi\lmk \caA_{\Gamma\setminus S}\rmk''$.
Because $\caE_S$ is $\sigma$-weak-continuous, this implies 
$\caE_S\lmk \pi(\caA_{\Gamma})''\rmk\subset \pi(\caA_{\Gamma\setminus S})''$.
\end{proof}

\begin{lem}\label{lem14}
Let $\Lambda_1,\Lambda_2,\Gamma$ be non-empty subsets of $\bbZ^2$
satisfying $\Lambda_1,\Lambda_2\subset \Gamma$.
Let $(\caH,\pi)$ be a representation of $\caA_\Gamma$.
Suppose that $x\in\pi\lmk \caA_{(\Lambda_1)^c\cap\Gamma}\rmk'$,
$y\in \pi\lmk \caA_{\Lambda_2}\rmk''$,
$\varepsilon>0$ satisfy $\lV x-y\rV\le\varepsilon$.
Then there exists $z\in \pi\lmk \caA_{(\Lambda_1)^c\cap\Gamma}\rmk'\cap \pi\lmk \caA_{\Lambda_2}\rmk''$
satisfying $\lV x-z\rV\le\varepsilon$.
 If $y$ is self-adjoint, then $z$ can be taken to be self-adjoint.
\end{lem}
\begin{proof}
If $\lm 1^c\cap\Gamma\cap\lm 2=\emptyset$,
then $z=y$ satisfies the condition.
Suppose $\lm 1^c\cap\Gamma\cap\lm 2\neq\emptyset$.

Fix an increasing sequence $\{S_n\}_{n=1}^\infty$ of finite subsets of $\Lambda_1^c\cap \Gamma\cap\lm 2$
such that $S_n\uparrow \Lambda_1^c\cap \Gamma\cap\lm 2$.
For each $S_n$, there exists a projection of norm $1$, 
$\caE_{S_n}:\caB(\caH)\to \pi\lmk\caA_{S_n}\rmk'$, by Lemma \ref{lem12}.

By the definition of $\caE_{S_n}$, we have
\begin{align}
z_n:=\caE_{S_n}(y)=\sum_{I\in \caI} \pi(e_{II_0}) y\pi(e_{I_0I})\in \pi\lmk\caA_{S_n}\rmk',
\end{align}
with a system of matrix units $\{e_{IJ}\}$ of $\caA_{S_n}$.
Because $y\in \pi\lmk \caA_{\Lambda_2}\rmk''$ and 
\begin{align}
\pi(e_{II_0})\in \pi\lmk \caA_{S_n}\rmk\subset  \pi\lmk \caA_{\Lambda_1^c\cap \Gamma\cap\lm 2}\rmk
\subset \pi\lmk \caA_{\Lambda_2}\rmk,
\end{align}
we get $z_n\in  \pi\lmk \caA_{\Lambda_2}\rmk''\cap \pi\lmk\caA_{S_n}\rmk'$.
Because
$x\in\pi\lmk \caA_{(\Lambda_1)^c\cap\Gamma}\rmk'\subset \pi\lmk\caA_{S_n}\rmk'$,
we have $\caE_{S_n}(x)=x$.
Because $\caE_{S_n}$ is of norm $1$, we have
\begin{align}
\lV x-z_n\rV
=\lV \caE_{S_n}\lmk x-y\rmk\rV
\le \lV x-y\rV\le \varepsilon,\quad
\lV z_n\rV\le \lV y\rV.
\end{align}

Because $\{z_n\}$ is a bounded sequence in $\pi\lmk \caA_{\Lambda_2}\rmk''$,
there exists a subnet $\{z_{n_\alpha}\}$
converging to some $z\in \pi\lmk \caA_{\Lambda_2}\rmk''$ 
with respect to the $\sigma$-weak topology.
Because each $z_n$ belongs to $\pi\lmk\caA_{S_n}\rmk'$ and 
$S_n\uparrow \Lambda_1^c\cap \Gamma\cap\lm 2$, we have $z\in \pi\lmk\caA_{\Lambda_1^c\cap \Gamma\cap\lm 2}\rmk'$.
Hence we get 
\begin{align}
z\in \pi\lmk \caA_{\Lambda_2}\rmk''\cap \pi\lmk\caA_{\Lambda_1^c\cap \Gamma\cap\lm 2}\rmk'
=\pi\lmk \caA_{\Lambda_2}\rmk''\cap \pi\lmk\caA_{\Lambda_1^c\cap \Gamma\cap\lm 2}\rmk'
\cap\pi\lmk\caA_{\lm 2^c}\rmk'
\subset \pi\lmk \caA_{\Lambda_2}\rmk''\cap \pi\lmk\caA_{\Lambda_1^c\cap \Gamma}\rmk'.
\end{align}

For any $\xi,\eta\in\caH$, we have
\begin{align}
\lv
\braket{\eta}{(x-z)\xi}
\rv
=\lim_{\alpha}\lv
\braket{\eta}{(x-z_{n_\alpha})\xi}
\rv\le \varepsilon \lV \eta\rV\lV\xi\rV.
\end{align}
Hence we conclude $\lV x-z\rV\le\varepsilon$.
By the construction, $z$ is self-adjoint if $y$ is self-adjoint.
\end{proof}

\section{$\Obu$ in $\OrUbd$}\label{Orcen}
In this section, we show that $\Obu$ can be regarded as a center of $\OrUbd$
in a certain sense.
Recall in Lemma \ref{lem40}  we obtained
$
\hb{\rho_{\lz}} - : \OrUbdl\to \caU(\hbd)
$
for each
${\rho_{\lz}}\in \Obul$,
satisfying the naturality Lemma \ref{lem42}.
For each $u\in \caU(\caA_{\lz^c\cap\lzr})$, 
 $\sigma_u:=\pbd\Ad u$ defines an element in $\OrUbdl$.
 For any $\lm {r2}\in\Crbd$ with $\lz\leftarrow_r\lm {r2}$, $\lm {r2}\subset \lzr$, $t\ge 0$,
 we have $\pbd(u^*)\in  \VUbd{\sigma_u}{\lm{r2}(t)}$.
 Therefore, we have $\hb{\rho_{\lz}}{\sigma_u}
 =\pbd(u)\Tbd{\rho_{\lz}}{\lzr}\unit\lmk \pbd(u^*)\rmk=\pbd(u){\rho_{\lz}}(u^*)=\unit$.
We further note the following asymptotic property
of this map.
\begin{lem}\label{lem84}
Consider the setting in subsection \ref{setting2} and assume Assumption \ref{assum80}, \ref{assum80l}.
Let $(\lz,\lzr)\in \pc$ and ${\rho_{\lz}}\in \Obul$.
Then for any $\lm {r2}\in\Crbd$ with $\lz\leftarrow_r\lm {r2}$, $\lm {r2}\subset \lzr$
we have
\begin{align}
\begin{split}
\lim_{t\to\infty}
\sup_{\sigma\in O^{rU}_{\mathop{\mathrm{bd},\lm {r2}(t)}}}
\lV
\hb{{\rho_{\lz}}}{\sigma}-\unit
\rV=0.
\end{split}
\end{align}
\end{lem}
\begin{proof}
Fix $t\ge 0$. 
For any $\sigma\in O^{rU}_{\mathop{\mathrm{bd},\lm {r2}(t)}}$,
by (\ref{niigata}), we have
\begin{align}
\iota^{(\lz,\lzr)}\lmk{\rho_{\lz}}: \sigma\rmk=
\lim_{t_2\to\infty}\lmk \Vrl \sigma{\lm {r2}(t_2)}\rmk^*\Tbd{\rho_{\lz}}{\lzr}\unit
\lmk\Vrl \sigma{\lm {r2}(t_2)}\rmk,
\end{align}
and for each $t_2\ge t$,
$\Vrl \sigma{\lm {r2}(t_2)}\in \pbd\lmk \caA_{(\lm {r2}(t))^c\cap\hu}\rmk'\cap\fbd$,
hence
\begin{align}
\begin{split}
\lV
\iota^{(\lz,\lzr)}\lmk{\rho_{\lz}}: \sigma\rmk-\unit
\rV\le
\lV
\left. \lmk\Tbd{\rho_{\lz}}{\lzr}\unit-\id\rmk\right\vert_{\pbd\lmk \caA_{(\lm {r2}(t))^c\cap\hu}\rmk'\cap\fbd}
\rV.
\end{split}
\end{align}
Combining with Lemma \ref{lem82}, we obtain the claim.
\end{proof}
\begin{lem}\label{lem83}
Consider the setting in subsection \ref{setting2}.
Let $\rho\in \OrUbd$.
Suppose that for each 
$(\lz,\lzr)\in \pc$, there exists $\rho_{\lz}\in \OrUbdl$
with
$\hat\iota(\rho_{\lz}: -):\OrUbdl \to \caU(\hbd)$ satisfying the following.
\begin{description}
\item[(i)]
There exists $\lm {r2}\in\Crbd$ with $\lz\leftarrow_r\lm {r2}$, $\lm {r2}\subset \lzr$
such that
\begin{align}\label{okayama}
\lim_{t\to\infty}
\sup_{\sigma\in O^{rU}_{\mathop{\mathrm{bd},\lm {r2}(t)}}}
\lV
\hat\iota({\rho_{\lz}} :\sigma)-\unit
\rV=0.
\end{align}
\item[(ii)]
For all $\sigma,\sigma'\in \OrUbdl$,
$S\in (\sigma,\sigma')_U$,
\begin{align}
\hat\iota({\rho_{\lz}} :\sigma')(\unit_{\rho_{\lz}}\otimes S)
=(S\otimes \unit_{\rho_{\lz}})\hat\iota({\rho_{\lz}} :\sigma).
\end{align}
\item[(iii)]
For each $u\in \caU(\caA_{\lz^c\cap\lzr})$, 
$\hat\iota(\rho_{\lz}: \sigma_u)=\unit$,
with $\sigma_u:=\pbd\Ad u$.
\item[(iv)] $\rho\simeq_U\rho_{\lz}$.
\end{description}
Then $\rho$ belongs to $\Obu$.
\end{lem}
\begin{rem}
We already saw that $\rho\in \Obu$ satisfies the assumptions of the Lemma,
under Assumption \ref{assum80} and Assumption \ref{assum80l}.
\end{rem}
\begin{proof}
For any $\ld\in \CUbk$, there exists $(\lz,\lzr)\in \pc$
such that $\lz\subset\ld\cap\hu$.
For this $(\lz,\lzr)\in\pc$, consider the representation $\rho_{\lz}\in \OrUbdl$ in the assumption.
For $\lm{r2}$ in (i) and 
 any $\sigma\in \OrUbdl$, 
take $\Vrl\sigma{\lm {r2}(t)}\in \VUbd\sigma{\lm{r2}(t)}$, $t\ge 0$.
Then we have $\sigma(t):=\Ad\lmk \Vrl\sigma{\lm {r2}(t)}\rmk\sigma\in\OrUbdn{\lm {r2}(t)}\subset  \OrUbdl$ and
$\Vrl\sigma{\lm {r2}(t)}\in \lmk \sigma,\sigma(t)\rmk_U$.
Hence by the naturality (ii), we have
\begin{align}
\hat\iota({\rho_{\lz}} :\sigma)=
(\Vrl\sigma{\lm {r2}(t)}\otimes \unit_{\rho_{\lz}})^*\hat\iota({\rho_{\lz}} :\sigma(t))(\unit_{\rho_{\lz}}\otimes \Vrl\sigma{\lm {r2}(t)})
=(\Vrl\sigma{\lm {r2}(t)})^*\hat\iota({\rho_{\lz}} :\sigma(t))\Tbd{\rho_{\lz}}{\lzr}\unit\lmk  \Vrl\sigma{\lm {r2}(t)}\rmk.
\end{align}
Taking $t\to\infty$, from (i),  we obtain
\begin{align}
\begin{split}
\hat\iota({\rho_{\lz}} :\sigma)
=\lim_{t\to\infty} 
(\Vrl\sigma{\lm {r2}(t)})^*\Tbd{\rho_{\lz}}{\lzr}\unit\lmk  \Vrl\sigma{\lm {r2}(t)}\rmk.
\end{split}
\end{align}
For each 
 $\sigma_u:=\pbd\Ad u$ with $u\in \caU(\caA_{\lz^c\cap\lzr})$, 
and $\pbd(u^*)\in  \VUbd{\sigma_u}{\lm{r2}(t)}$, we apply this formula.
Then, with (iii), we obtain
\begin{align}
\begin{split}
\unit=\hat\iota({\rho_{\lz}} :\sigma_u)
=\pbd(u)\Tbd{\rho_{\lz}}{\lzr}\unit\lmk \pbd(u^*)\rmk
=\pbd(u)\rho_{\lz}(u^*),
\end{split}
\end{align}
for all $u\in \caU(\caA_{\lz^c\cap\lzr})$.
 This, combined with $\rho_{\lz}\in \OrUbdl$,
 implies $\rho_{\lz}\vert_{\caA_{\lz^c\cap \hu}}=\pbd\vert_{\caA_{\lz^c\cap\hu}}$,
 hence $\rho_{\lz}\vert_{\caA_{\ld^c\cap \hu}}=\pbd\vert_{\caA_{\ld^c\cap\hu}}$.
 From (iv), we conclude $\rho\in\Obu$.

\end{proof}

\noindent{Funding and/or Conflicts of interests/Competing interests This work was supported by JSPS KAKENHI Grant Number 19K03534 and 22H01127.
It was also supported by JST CREST Grant Number JPMJCR19T2.
The author has no competing interests to declare that are relevant to the content of this article.}\\\\
\noindent{DATA AVAILABILITY
The data that support the findings of this study are available within the article.}

\bibliographystyle{amsplain}
\bibliography{BA}
\end{document}